\documentclass[reqno]{tran-l}

\usepackage[utf8]{inputenc}
\usepackage{ifthen, mathrsfs, amsmath, amsfonts, amssymb, amsthm, braket, bm, graphicx, mathtools, enumitem, color}
\usepackage[nocompress]{cite}
\usepackage{microtype}
\usepackage[english]{babel}
\usepackage[nodayofweek,long]{datetime}
\usepackage[colorlinks=true, allcolors=blue]{hyperref}
\hypersetup{final}
\usepackage[hang,flushmargin]{footmisc}
\usepackage{tikz-cd, upgreek}

\newcommand{\numorbitals}{N_\mathrm{b}}
\newcommand{\ctTBprefactor}{h_0}
\newcommand{\ctTBexponent}{{\gamma_0}}
\newcommand{\ctHamregularity}{\nu}
\newcommand{\ctnoninterpen}{\mathfrak{m}}
\newcommand{\ctCT}{\gamma_\mathrm{CT}}
\newcommand{\ctGamma}{{\Upsilon}}

\newcommand{\ctDist}{\mathsf{d}}

\newcommand{\Ham}{\mathcal{H}}
\newcommand{\W}{{\dot{\mathscr W}}^{1,2}}
\renewcommand{\Re}[1]{\mathrm{Re}\left(#1\right)}
\renewcommand{\Im}[1]{\mathrm{Im}\left(#1\right)}
\renewcommand{\leq}{\leqslant}\renewcommand{\geq}{\geqslant}
\renewcommand{\above}[2]{\genfrac{}{}{0pt}{}{#1}{#2}}

\newcommand{\ubar}{\overline{u}}
\DeclareMathOperator*{\argmin}{arg\,min}
\newcommand{\ep}{\varepsilon}
\newcommand{\dxx}[1]{\mathrm{d}#1}

\newcommand{\TRstar}{T_{R^\star}}
\newcommand{\TRtil}{T_{\tilde{R}}}
\newcommand{\distR}[2]{{r_{#1}^{#2}}}

\newtheorem*{fact*}{Fact}
\newtheorem{theorem}{Theorem}[section]
\newtheorem*{theorem*}{Theorem}

\newtheorem{lemma}[theorem]{Lemma}
\newtheorem{prop}[theorem]{Proposition}

\theoremstyle{remark}
\newtheorem{remark}{Remark}

\newtheoremstyle{assumptionStyle}
  {0.5em}
  {0.5em}
  {}            
  {0.5cm}       
  {\bfseries}   
  {}            
  {0.2cm}       
  {}            

\newcommand{\thistheoremname}{}
\theoremstyle{assumptionStyle}\newtheorem*{genericthm*}{\textup{\thistheoremname}}
\newenvironment{assumption}[2][1.6cm]{
	\renewcommand{\thistheoremname}{#2}%
	\begin{genericthm*}\hangindent=#1
	\setlength{\parindent}{#1+1.5em}} 
	{\end{genericthm*}}

\numberwithin{equation}{section}

\usepackage{geometry}
\usepackage{cleveref}
\crefformat{equation}{\textup{(#2#1#3)}}
\crefformat{section}{\S#2#1#3} 
\crefformat{subsection}{\S#2#1#3}
\crefformat{subsubsection}{\S#2#1#3}
\crefformat{thm}{Theorem~#2#1#3}

\newcommand{\refCE}[3]{$(\hyperref[#1]{\mathrm{CE}^{{#2},{#3}}})$}
\newcommand{\refGCE}[4]{$(\hyperref[#1]{\mathrm{GCE}_{{#4}}^{{#2},{#3}}})$}

\begin{document}

\title[Point Defects in Tight Binding Models for Insulators]{Point Defects in Tight Binding Models for Insulators}

\author{Christoph Ortner}
\address{Mathematics Institute,
Zeeman Building,
University of Warwick,
Coventry, UK}
\curraddr{}
\email{C.Ortner@warwick.ac.uk}

\author{Jack Thomas}
\address{Mathematics Institute,
Zeeman Building,
University of Warwick,
Coventry, UK}
\curraddr{}
\email{J.Thomas.1@warwick.ac.uk}

\begin{abstract}
We consider atomistic geometry relaxation in the context of linear tight binding models for point defects. A limiting model as Fermi-temperature is sent to zero is formulated, and an exponential rate of convergence for the nuclei configuration is established. We also formulate the thermodynamic limit model at zero Fermi-temperature, extending the results of [H.~Chen, J.~Lu, C.~Ortner. \textit{Arch.~Ration.~Mech.~Anal.}, 2018]. We discuss the non-trivial relationship between taking zero temperature and thermodynamic limits in the finite Fermi-temperature models.  
\end{abstract}

\maketitle


\section{Introduction}
\label{sec:introduction}
Electronic structure calculations are widely used to model and understand a variety of electronic, optical and magnetic phenomena observed in solids \cite{bk:finnis,bk:martin04}. Tight binding models are minimalistic electronic structure models that lie, both in terms of computational cost and accuracy, between \textit{ab initio} calculations and empirical interatomic potential models. Thus, tight binding models are useful in a number of situations where \textit{ab initio} calculations are impossible to undertake due to the large system size but are advantageous over empirical methods due to the increased accuracy and significant underlying quantum mechanical effects present. Although interesting in their own right, tight binding models also serve as case studies for the technically much more challenging density functional theory.{\let\thefootnote\relax\footnote{
\textit{Date.} \today \\
\textit{2000 Mathematics Subject Classification.} 65N12, 65N25, 74E15, 74G65, 81V45.\\
\textit{Key words and phrases.} point defects; tight binding model; zero temperature limit; thermodynamic limit; insulators.\\
CO is supported by EPSRC Grant EP/R043612/1 and Leverhulme Research Project Grant RPG-2017-191.\\
JT is supported by EPSRC as part of the MASDOC DTC, Grant No. EP/HO23364/1.
}}

In the present work, we consider a linear tight binding model and formulate the limiting model as Fermi-temperature tends to zero. Further, we consider the thermodynamic limit at zero Fermi-temperature and explore the extent to which these two limits commute. 

The simulation of local defects in solids remains a major issue in materials science and solid state physics \cite{Pisani1994,bk:Stoneham2001}. For a mathematical review of some works related to the modelling of point defects in materials science see \cite{CancesLeBris2013}. Progress on local defects in the context of Thomas-Fermi-von-Weizs\"acker (TFW) and reduced Hartree-Fock (rHF) models has been made in \cite{LiebSimon1977,BlancLeBrisLions2007,CancesEhrlacher2011} and \cite{CancesDeleurenceLewin2008,GontierLahbabi17}, respectively. 

Thermodynamic limit (or bulk limit) problems have been widely studied in the literature. The case of a perfect crystalline lattice has been studied in \cite{LiebSimon1977} for the Thomas-Fermi (TF) model, \cite{bk:CattoLeBrisLions1998}  for the TFW model and in \cite{CattoLeBrisLions2002} and \cite{CattoLeBrisLions2001} for Hartree and Hartree-Fock type models, respectively. In these papers, the limit of the ground state energy per unit volume and minimising electronic density as domain size tends to infinity are identified in the cases where the energy functionals are convex (that is, for the TF, TFW, restricted Hartree and rHF models). For more general Hartree and Hartree-Fock type models, periodic models have been proposed and shown to be well-posed. In the setting of the rHF model, an exponential rate of convergence for the supercell energy per unit cell is obtained in the case of insulators \cite{GontierLahbabi2016convergence}.

It is important to note that in all of the papers mentioned above the nuclei degrees of freedom are fixed on a periodic lattice or with a given defect. Preliminary results concerning the simultaneous relaxation of the electronic structure together with geometry equilibriation (of the nuclei positions) can be found in \cite{NazarOrtner17} for the TFW model and \cite{ChenOrtner16,ChenLuOrtner18} for tight binding models.

The present work is motivated by \cite{ChenLuOrtner18} which establishes the following two results: (i) under a mild condition on the prescribed number of electrons in the sequence of finite domain approximations, the Fermi level is shown to converge in the thermodynamic limit to that of a perfect crystal and is thus independent of the electron numbers and the defect. This result enables, (ii) the formulation of a unique limiting model in the grand-canonical ensemble for the electrons with chemical potential fixed at the perfect crystal level. The purpose of the present work is to explore the extent to which these results can be extended to the zero Fermi-temperature case, as well as consider the zero Fermi-temperature limit of the model described in \cite{ChenLuOrtner18}.

\subsection*{Summary of Results}

The main convergence results of \cite{ChenLuOrtner18} and the present work are summarised in Figure~\ref{fig:cd-1}. 

\usetikzlibrary{decorations.pathmorphing}
\begin{figure}[htbp]
 	\centering
 	\begin{tikzcd}[row sep=large, column sep=large]
 	{\boxed{\above{\text{CE}}{\beta<\infty, R<\infty}}}  
 	\arrow[rr, "\above
 	                {R\to\infty}
 	                {\left[\textcolor{blue}{10}\right]}
 	            " description] 
 	\arrow[dd, "\above{\beta\to\infty}{
 	    \above
 	        {\text{Thm.}~\ref{cor:strong_0T_limit_infinite_domain},}
 	        {\text{Prop.}~\ref{cor:weak_0T_limit_infinite_domain}}
 	    }" description] &  &
 	%
 	\boxed{\above{\text{GCE}}{{\beta<\infty,R=\infty}}}   
 	\arrow[dd, "\above{\beta\to\infty}{
 	    \above
 	        {\text{Thm.}~\ref{thm:0T_limit_infinite_domain},}
 	        {\text{Prop.}~\ref{prop:0T_limit_infinite_domain}}
        }" description] &  &
 	%
 	{\boxed{\above{\text{GCE}}{\beta<\infty,R<\infty}}}   
 	\arrow[ll, "\above
 	                {R\to\infty}
 	                {\left[\textcolor{blue}{10}\right]}
 	           " description] 
 	\arrow[dd, "\above
 	                {\beta\to\infty}
 	                {\above
 	                    {\text{Thm.}~\ref{thm:0T_limit_infinite_domain},}
 	                    {\text{Prop.}~\ref{prop:0T_limit_infinite_domain}}
 	                }
 	            " description] 
 	\\  &  &  &  &  \\
 	%
 	{\boxed{\above{\text{CE}}{\beta=\infty,R<\infty}}}  
 	\arrow[rr,rightsquigarrow, "\above
 	                                {R\to\infty}
 	                                {\above
 	                                   {\text{Thm.}~\ref{thm:0T_thermodynamic_limit_CE_GCE},}
 	                                   {\text{Prop.~}\ref{prop:0T_thermodynamic_limit_CE_GCE}}
 	                                }
 	                            " description] &  &
 	%
 	{\boxed{\above{\text{GCE}}{\beta=\infty,R=\infty}}}   &  & 
 	%
 	{\boxed{\above{\text{GCE}}{\beta=\infty,R<\infty}}}        
 	\arrow[ll, "\above
 	                {R\to\infty}
 	                {\above
 	                    {\text{Thm.}~\ref{thm:0T_thermodynamic_limit},}
 	                    {\text{Prop.}~\ref{prop:0T_thermodynamic_limit}}
 	                }
 	            " description]                                          
 	\end{tikzcd}
 	\caption{Diagram to illustrate the main results of \cite{ChenLuOrtner18} and the present work. Here, ``GCE'' denotes the grand canonical ensemble and ``CE'' the canonical ensemble. The top two thermodynamic limits represent the results of \cite{ChenLuOrtner18} and the results of this paper are indicated on the remaining arrows.}
 	\label{fig:cd-1}
\end{figure}
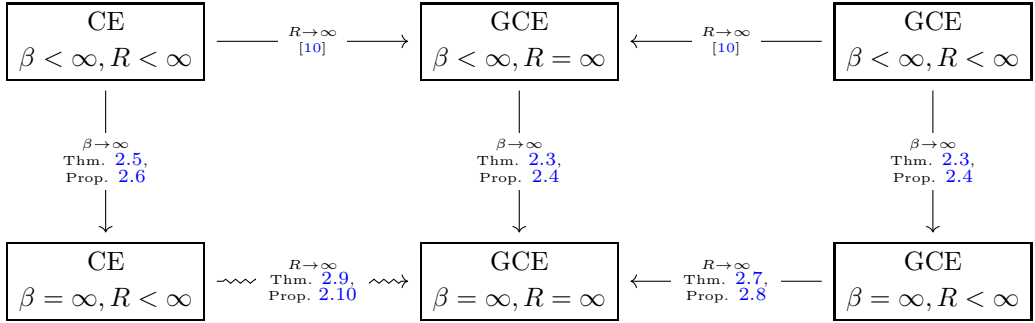
%

\subsubsection*{Thermodynamic Limit.} Since we are interested in the bulk properties of a material with local defects, it is convenient to consider an extended system of infinitely many nuclei. However, to simulate such a system, we must of course restrict ourselves to finite computational domains and impose an artificial boundary condition. Throughout this paper we shall consider periodic boundary conditions for the nuclei (that is, a supercell model) in the form of a torus tight binding model and show that the thermodynamic limit is well defined under some appropriate choice of electron numbers in the sequence of finite domain approximations.

More precisely, we consider linear tight binding models with electrons in the canonical ensemble and at zero Fermi-temperature. We show that, because the zero temperature Fermi levels depend globally on each eigenvalue (and not just on the limiting density of states as in the case of finite Fermi-temperature), the zero temperature Fermi levels only converge in the thermodynamic limit under strict conditions on the number of electrons imposed in the sequence of finite domain approximations. Moreover, the limiting Fermi level depends upon the polluted band structure and consequently also on the defect. Using this, the thermodynamic limit model is shown to be a grand canonical model with chemical potential fixed (but defect-dependent) and given by the limit of the sequence of finite domain Fermi levels. That is, the number of electrons imposed in the sequence of finite domain approximations is critical in identifying a limiting model. This analysis clarifies questions left open in \cite{ChenLuOrtner18} about the effect of Fermi-temperature on the convergence.

\subsubsection*{Zero Temperature Limit.} A key feature of zero temperature electronic structure models is the sharp cut-off between unoccupied and occupied electronic states. In practice (e.g.~\cite{ChenLu2016} for density functional theory), a low but positive Fermi-temperature may be chosen in order to approximate the sharp cut-off with a smooth Fermi-Dirac distribution (alternatively, artificial smearing methods may be used). One can then show that the error committed does not drastically affect the simulation; see \cite{CancesEhrlacherGontierLevittLombardi2020} for an in-depth error analysis for typical observables (including the Fermi level, total energy and the density). Choosing a finite Fermi-temperature has the additional benefit that there is a unique Fermi level (see \cref{eq:FL}) solving the electron number constraint which is advantageous in numerical simulations \cite{ChenLu2016,Yang1991}. 

In the present work, we give a comprehensive justification of this approach; assuming the electrons are in finite Fermi-temperature and the nuclei degrees of freedom are determined by minimising the grand potential associated with the electrons, we uniquely identify the limiting model as Fermi-temperature tends to zero by a grand canonical model for the electrons at zero Fermi-temperature. We quantify an exponential rate of convergence for the nuclei configuration.

\subsubsection*{Strong Energy Locality.} A key tool in \cite{ChenLuOrtner18} and the present work is a strong energy locality of the tight binding model, first proved in \cite{ChenOrtner16}, extended to other quantities of interest in \cite{ChenLuOrtner18} and strengthened for the case of point defects in insulating systems in \cite{ChenOrtnerThomas2019}. This locality allows for the definition of a renormalised energy functional on the infinite lattice with embedded point defect and thus allows for the formulation of a limiting model \cite{ChenNazarOrtner19}.

\subsection*{Notation}
For matrices or Hilbert–Schmidt operators, we denote the Frobenius or Hilbert–Schmidt norm by $\|\cdot\|_\mathrm{F}$. The $\ell^2$ norm on sequence spaces and Euclidean norm on $\mathbb R^n$ or $\mathbb C$ will be denoted by $\|\cdot\|_{\ell^2}$ and $|\cdot|$, respectively. For operators defined on $\ell^2$, we denote by $\|\cdot\|_{\ell^2 \to \ell^2}$ the operator norm. For a multi-index $\theta \in \mathbb N^d$ and $\alpha \in \mathbb Z^d$, we will write 
$|\theta|_1 \coloneqq \sum_{j=1}^d \theta_j$ and $|\alpha|_\infty \coloneqq \max_{j} |\alpha_j|$. 
For $A,B \subset \mathbb C$, $r>0$ and $x \in \mathbb C$, the distance between $x$ and $A$ is defined by $\mathrm{dist}(x,A) \coloneqq \inf_A |x - \cdot\,|$, the Hausdorff distance between $A$ and $B$ is denoted by $\mathrm{dist}(A,B) \coloneqq \max\{ \sup_{a \in A}\mathrm{dist}(a,B), \sup_{b \in B}\mathrm{dist}(b,B) \}$ and the (open) ball of radius $r$ about $A$ is defined by 
$B_r(A) \coloneqq \{y \in \mathbb C\colon \mathrm{dist}(y,A) < r\}$.
For a normed space $(X,\|\cdot\|)$ and $x \in X$, we denote by $B_r(x;\|\cdot\|) \coloneqq \{y \in X\colon \|x - y\| < r\}$ the (open) ball of radius $r$ about $x$. When considering the Euclidean norm on $\mathbb R^n$ or $\mathbb C$, we omit the norm in the notation, $B_r(x;\|\cdot\|) = B_r(x)$, and we write $B_r \coloneqq B_r(0)$ for balls centred at the origin.

For a sequence of sets $A_n \subset \mathbb R$, we define
$
    \liminf_{n \to \infty} A_n \coloneqq 
        \{ a \colon \exists\, a_n \in A_n \textrm{ s.t. } a_n \to a \} 
$
and
$
    \limsup_{n \to \infty} A_n \coloneqq 
        \{ a \colon \exists\, a_n \in A_n \textrm{ s.t. } a_{n} \to a \textrm{ along a subsequence}\}.
$
The (topological) limit of $A_n$, denoted $\lim_{n\to\infty} A_n$, is defined and equal to $\liminf_{n \to \infty} A_n$ and $\limsup_{n \to \infty} A_n$ in the case that these two limits agree.

The set of strictly positive real numbers will be denoted by $\mathbb R_+ \coloneqq \{r \in \mathbb R \colon r > 0\}$. We write $b + A = \{ b + a \colon a \in A \}$ and similarly for $A - b$. For a subset $A_0 \subset A$, we denote by $\chi_{A_0}\colon A \to \{0,1\}$ the characteristic function of $A_0$. For a function $f\colon A \to B$, we denote by $f|_{A_0}$ the restriction of $f$ to $A_0$. If $B \subset \mathbb C$, we denote by $\mathrm{supp}(f)$ the support of $f$. For a finite set $A$, we denote by $\#A$, the cardinality of $A$.

For $\mathcal G \in C^3(X)$ and $u,v,w,z \in X$, we let $\Braket{\delta \mathcal G(u),v}$, $\Braket{\delta^2 \mathcal G(u)v,w}$ and $\Braket{\delta^3 \mathcal G(u)v,w,z}$ denote the first, second and third variations of $\mathcal G$, respectively. We write $\|\delta \mathcal G(u)\|$ for the operator norm of $\delta\mathcal G(u)$.

For a self-adjoint operator $T$, we let $\sigma(T)$ denote the spectrum of $T$ and $\sigma_\mathrm{disc}(T)$ the discrete spectrum of $T$ (that is, the set of isolated eigenvalues of finite multiplicity) while $\sigma_\mathrm{ess}(T) \coloneqq \sigma(T)\setminus \sigma_\mathrm{disc}(T)$ is known as the essential spectrum.

The symbol $C$ will denote a generic positive constant that may change from one line to the next.  In calculations, $C$ will always be independent of Fermi-temperature and computational domain size. The dependencies of $C$ will be clear from context or stated explicitly. When convenient to do so we write $f \lesssim g$ to mean $f \leq C g$ for some generic positive constant as above.

\section{Results}
\label{sec:results}
\subsection{Point Defect Reference Configurations}
\label{pg:pointdefects}
\newcommand{\asPointDefect}{\textup{\textbf{(P)}}}

We fix a locally finite reference configuration $\Lambda^\mathrm{ref} \subset \mathbb R^d$ and consider the configuration, $\Lambda$, obtained from $\Lambda^\mathrm{ref}$ by embedding a point defect:
\begin{assumption}[1.3cm]{\asPointDefect}
    There exists an invertible matrix $\mathsf{A} \in \mathbb R^{d\times d}$ and a \textit{unit cell} $\Gamma \subset \mathbb R^d$ such that $\Gamma$ is finite, contains the origin and
    $$\Lambda^\mathrm{ref} = \bigcup_{\gamma \in \mathbb{Z}^d} (\Gamma + \mathsf{A}\gamma).$$ 
    There exists $R_\textrm{def}>0$ such that $B_{R_\textrm{def}}\cap\Lambda$ is finite and $\Lambda^\mathrm{ref} \setminus B_{R_\textrm{def}} = \Lambda \setminus B_{R_\textrm{def}}$.
\end{assumption}
We will always think of $\Lambda^\mathrm{ref}$ as a ground state insulating multilattice material.

For displacements $u \colon \Lambda \to \mathbb R^d$, we suppose that the following non-interpenetration condition is satisfied:
\newcommand{\asNonInter}{\textup{\textbf{(L)}}}
\begin{assumption}[1.3cm]{\asNonInter}
	There exists $\mathfrak{m}>0$ such that $r_{\ell k}(u) \geq \mathfrak{m}|\ell - k|$ for all $\ell, k \in \Lambda$ where we use the notation ${\bm r}_{\ell k}(u) \coloneqq \ell + u(\ell) - k - u(k)$ and $r_{\ell k}(u) \coloneqq |\bm{r}_{\ell k}(u)|$.
\end{assumption}
\noindent When clear from the context, we shall drop the argument $(u)$ in $\bm{r}_{\ell k}(u)$ and $r_{\ell k}(u)$.

\subsection{Linear Tight Binding Hamiltonian}
\label{sec:tight_binding_hamiltonian}

We suppose there are $\numorbitals$ atomic orbitals per atom (indexed by $1\leq a,b\leq\numorbitals$) and consider a two-centre linear tight binding model with Hamiltonian, $\Ham(u)$, with matrix entries
\begin{align*}
	\left[\Ham(u)\right]_{\ell k}^{ab} = h^{ab}_{\ell k}(\bm{r}_{\ell k}(u)) \quad 
	\text{for } \ell, k \in \Lambda \text{ and } 1 \leq a,b \leq \numorbitals,
\end{align*}
where the functions $h^{ab}_{\ell k}\colon\mathbb R^d \to \mathbb R$ satisfy the following tight binding assumptions:
\newcommand{\asTB}{\textup{\textbf{(TB)}}}
\begin{assumption}{\asTB}
    We suppose that $h^{ab}_{\ell k}\colon \mathbb R^d \to \mathbb R$ are $\ctHamregularity$ times continuously differentiable with $\nu \geq 3$, and that there exist $\ctTBprefactor,\ctTBexponent>0$ with
	\begin{equation}
    	\label{a:TB}
    	\big|h^{ab}_{\ell k}(\xi)\big| \leq \ctTBprefactor e^{-\ctTBexponent|\xi|} \quad \text{and} \quad 
    	\big|\partial^\theta h^{ab}_{\ell k}(\xi)\big| \leq \ctTBprefactor e^{-\ctTBexponent|\xi|} 
    	\quad \forall\,\xi \in \mathbb R^d,
	\end{equation}
	for all multi-indices $\theta \in \mathbb N^d$ with $|\theta|_1 \leq \ctHamregularity$. 

	Moreover, we suppose that the Hamiltonian is symmetric (i.e.~$h^{ab}_{\ell k}(\xi) = h^{ba}_{k\ell}(-\xi)$ for all $1\leq a,b \leq \numorbitals$, $\ell, k \in \Lambda$ and $\xi \in \mathbb R^d$) and the following translational invariance condition is satisfied: $h^{ab}_{\ell + \mathsf{A}\gamma, k + \mathsf{A}\eta} = h^{ab}_{\ell k}$ for all $1 \leq a,b \leq \numorbitals$, $\ell, k \in \Lambda^\mathrm{ref}$ and $\gamma,\eta \in \mathbb Z^d$ such that $\ell + \mathsf{A}\gamma, k + \mathsf{A}\eta \in \Lambda^\mathrm{ref}$.
\end{assumption}
It is important to stress here that the constants $\ctTBprefactor,\ctTBexponent>0$ in \cref{a:TB} are independent of the atomic sites $\ell, k \in \Lambda$.

Most linear tight binding models impose a finite cut-off radius and so the pointwise bound on $|h^{ab}_{\ell k}|$ in \cref{a:TB} is normally automatically satisfied. For $|\theta|_1=1$, this assumption states that there are no long range interactions in the model. In particular, we suppose that the Coulomb interactions have been screened, a typical assumption in practical tight binding codes \cite{cohen94,mehl96,papaconstantopoulos15}.

The symmetry property of \asTB~means that the spectrum of $\Ham(u)$ is real. In practice, the Hamiltonian satisfies further symmetry properties which we briefly discuss in Appendix~\ref{sec:symmetries}.

The translational invariance property states that atoms $\ell, k \in \Lambda$ for which there exists a $\gamma \in \mathbb Z^d$ such that $\ell - k = \mathsf{A}\gamma$ are of the same species.

In practice, the number of atomic orbitals per atom depends on the atomic species. This notational complication can easily be avoided by redefining the Hamiltonian by taking $\numorbitals$ to be the maximum number of atomic orbitals per atom and shifting the spectrum away from $\{0\}$ as outlined in Appendix~\ref{app:numorbitals}. 

\subsection{Torus Tight Binding Model}
\label{sec:torus_TB}

In order to simulate the model described in \S\ref{sec:tight_binding_hamiltonian}, we must restrict ourselves to finite computational domains and impose artificial boundary conditions. A popular choice for simulating the far-field behaviour of point defects is periodic boundary conditions which we now introduce.

For $R>0$, we consider a sequence, $\Lambda_R$, of computational cells given by the following:
\newcommand{\asRefR}{\textup{\textbf{($\textrm{REF}_R$)}}}
\begin{assumption}[2.08cm]{\asRefR}
    For $R> 0$, we suppose $\Omega_R \subset \mathbb R^d$ is a bounded connected domain with $B_R \subset \Omega_R$. Further, we take an invertible matrix 
    $\mathsf{M}_R = (\mathsf{m}_1,\dots,\mathsf{m}_d) \in \mathbb R^{d\times d}$
    such that $\mathsf{m}_j \in \Lambda^\mathrm{ref}$ and $\mathbb R^d$ is the disjoint union of the shifted domains 
    $\Omega_R + \mathsf{M}_R\alpha$ for $\alpha \in \mathbb Z^d$. 
    The computational cell is defined to be 
    $\Lambda_R \coloneqq \Omega_R\cap\Lambda$.
\end{assumption}

When employing periodic boundary conditions, we shall consider displacements $u \colon \Lambda_R \to \mathbb R^d$ satisfying the following uniform non-interpenetration condition:
\newcommand{\asNonInterR}{\textup{\textbf{($\textrm{L}_R$)}}}
\begin{assumption}[1.5cm]{\asNonInterR}
    There exists $\ctnoninterpen > 0$ such that 
    $|\bm{r}_{\ell k}(u) + \mathsf{M}_R\alpha| \geq \ctnoninterpen|\ell - k + \mathsf{M}_R\alpha|$
    for all $\alpha \in \mathbb Z^d$ and atomic positions $\ell, k \in \Lambda_R$.
\end{assumption}
Throughout this paper, we shall say that displacements satisfying \asNonInterR~are admissible periodic displacements. We shall also identify $u$ with its periodic extension to $\bigcup_\alpha \left(\Lambda_R + \mathsf{M}_R\alpha\right)$. In this sense, the above assumption agrees with \asNonInter~from \cref{sec:tight_binding_hamiltonian} with $\Lambda$ replaced with $\bigcup_\alpha \left(\Lambda_R + \mathsf{M}_R\alpha\right)$.  

In order to simplify notation, we shall define
$\distR{\ell k}{\#}(u)\coloneqq \min_{\alpha\in\mathbb Z^d} |\bm{r}_{\ell k}(u) + \mathsf{M}_R\alpha|$ 
to be the torus distance between atomic positions $\ell + u(\ell)$ and $k + u(k)$. If it is clear from context, we shall drop the argument $(u)$.

The torus tight binding Hamiltonian is given by
\begin{equation}
    \label{a:two-centre}
	\big[\Ham^R(u)\big]_{\ell k}^{ab} 
		= \sum_{\alpha \in \mathbb Z^d} h^{ab}_{\ell k}(\bm{r}_{\ell k}(u) + \mathsf{M}_R\alpha) 
		\quad \textup{for } \ell, k \in \Lambda_R \text{ and } 1\leq a,b\leq \numorbitals.
\end{equation}

We note that, under \asTB~and \asNonInterR, the spectrum of the Hamiltonian, $\sigma(\Ham^R(u))$, is uniformly bounded in some interval that depends only on $\ctnoninterpen,d,\ctTBprefactor,\ctTBexponent$ and is independent of $R$ and the displacement $u$ satisfying \asNonInterR, with the constant $\ctnoninterpen$ \cite[Lemma~2.1]{ChenOrtner16}. This can also be seen as a direct corollary of Lemma~\ref{lem:perturbation_spec}, below. 

\begin{remark}
In this paper, we consider periodic boundary conditions in order to avoid spectral pollution that is known to occur when using clamped boundary conditions, see for example \cite{CancesEhrlacherMaday2012}. Indeed, in Lemma~\ref{lem:spectral_pollution} we prove that spectral pollution does not occur in our setting.
\end{remark}

\subsubsection{Potential Energy at Finite Fermi-Temperature}
\label{sec:finiteTenergy}

We consider a particle system containing $N_{e,R}$ electrons and nuclei described by some admissible displacement $u \colon \Lambda_R \to \mathbb{R}^d$. For simplicity of notation we define 
$N_R \coloneqq \#\Lambda_R\cdot N_\textrm{b}$
(and so $\Ham^R(u)\in \mathbb R^{N_R \times N_R}$). 

We first suppose that the electrons are in a canonical ensemble. That is, we fix the number of particles in the system, the volume and Fermi-temperature $T>0$. The particle number functional is given by summing the electronic occupation numbers according to the Fermi-Dirac distribution:
\begin{align}
    \label{eq:particleno}
    N^{\beta,R}(u;\tau) &\coloneqq 2\sum_{s=1}^{N_R} f_\beta(\lambda_s - \tau) 
    \quad \text{ and } \quad 
    f_\beta(\ep) \coloneqq \frac{1}{1 + e^{\beta \ep}},
\end{align}
where $\{\lambda_s\}_{s=1}^{N_R}$ is some enumeration of $\sigma(\Ham^R(u))$. Here, the factor of two accounts for the spin and 
$\beta$
is known as the inverse Fermi-temperature. Since $\tau \mapsto N(u;\tau)$ is strictly increasing, for any electron number $N_{e,R} \in (0,2N_R)$, the Fermi level, $\ep_\mathrm{F}^{\beta,R}(u)= \ep_\mathrm{F}^{\beta,R}$, at finite Fermi-temperature solving 
\begin{align}
    \label{eq:FL}
    N^{\beta,R}(u;\ep_\mathrm{F}^{\beta,R}) = N_{e,R}
\end{align}
is well-defined (see Appendix~\ref{sec:pot_energy_surface} for more details).

The Helmholtz free energy is then given by 
\begin{align}
    \label{def:Helmholtz_free_energy}
    \begin{split}
        E^{\beta,R}(u) 
        &= \sum_{s=1}^{N_R} \mathfrak{e}^\beta(\lambda_s,\ep^{\beta,R}_\mathrm{F}) \quad \text{where}\\
        \mathfrak{e}^\beta(\lambda;\tau) 
        &= 2 \lambda f_\beta(\lambda - \tau) + \frac{2}{\beta}S(f_\beta(\lambda - \tau))\\
        &= 2 \tau f_\beta(\lambda-\tau) + \frac{2}{\beta} \log\left(1-f_\beta(\lambda-\tau)\right)
    \end{split}
\end{align}
where $S(f) = f \log(f) + (1 - f) \log (1-f)$ is the entropy contribution. More details regarding the derivation of this energy can be found in Appendix~\ref{sec:pot_energy_surface}.

In the following, it will be useful to consider the Helmholtz free energy as a function of both the configuration and Fermi level. That is, we abuse notation and define $E^{\beta,R}(u;\tau) \coloneqq \sum_s \mathfrak{e}^\beta(\lambda_s;\tau)$. In this notation, we have 
$E^{\beta,R}(u) \equiv E^{\beta,R}(u,\ep^{\beta,R}_\mathrm{F}(u))$. 

For a many-particle system that is free to exchange particles with a reservoir, it is useful to consider the {grand canonical ensemble}. This framework will also allow us to formulate the limiting models as $R\to\infty$. In this case, the Fermi-temperature, volume and chemical potential, $\mu$, are fixed model parameters $(TV\mu)$ and, instead of the Helmholtz free energy, we subtract the contribution resulting from varying the particle number and consider the {grand potential}:
\begin{align}
    \label{eq:grand_canonical_ensemble}
    \begin{split}
        G^{\beta,R}(u;\mu) &\coloneqq E^{\beta,R}(u;\mu) 
            - \mu N^{\beta,R}(u; \mu) 
            = \sum_{s = 1}^{N_R} \mathfrak{g}^\beta(\lambda_s, \mu) 
                \quad \text{where} \\
        \mathfrak{g}^\beta(\lambda;\tau) &\coloneqq 
            \mathfrak{e}^\beta(\lambda;\tau) - 2\tau f_\beta(\lambda-\tau) 
            = \frac{2}{\beta} \log\left(1-f_\beta(\lambda-\tau)\right).
    \end{split}
\end{align}

When it is clear from the context, we will drop $\mu$ in the argument for the particle number functional and grand potential: that is, $N^{\beta,R}(u;\mu) = N^{\beta,R}(u)$ and $G^{\beta,R}(u;\mu) = G^{\beta,R}(u)$.

\subsubsection{Potential Energy Surface at Zero Fermi-Temperature}
\label{sec:0Tenergy}
 
We now consider the Helmholtz free energy and grand potential at zero Fermi-temperature. We simply take the pointwise limit of \cref{eq:grand_canonical_ensemble} as $\beta \to \infty$ to obtain
\begin{align}
    \label{eq:zero_temp_finite_domian}
    \begin{split}
        G^{\infty,R}(u;\mu) &\coloneqq \sum_{s=1}^{N_R} \mathfrak{g}(\lambda_s;\mu) \quad \text{where}\\
        \mathfrak{g}(\lambda;\tau) &= 2(\lambda - \tau)\chi_{(-\infty,\tau)}(\lambda).
    \end{split}
\end{align}
See Lemma~\ref{lem:g_convergence_2} for justification of this limit. Moreover, we define the zero temperature particle number as the limit as $\beta \to \infty$:
\[
    N^{\infty,R}(u;\tau) \coloneqq 2\#\{ \lambda \in \sigma(\Ham^R(u)) \colon \lambda < \tau \} + \#\{ \lambda \in \sigma(\Ham^R(u)) \colon \lambda = \tau \}.
\]

For the Helmholtz free energy, we must also consider the Fermi level constraint. However, taking the $\beta \to \infty$ limit of the particle number functional yields a step function. This means that there may not be a unique solution to the particle number constraint \cref{eq:FL} at zero Fermi-temperature. We define the zero temperature Fermi level as the zero temperature limit of the finite temperature Fermi levels:
\begin{lemma}[Fermi Level at Zero Fermi-Temperature]
    \label{lem:fermilevelconverge}
    Suppose $u\colon\Lambda_R \to \mathbb R^d$ satisfies \asNonInterR~and that 
    $\ep_\mathrm{F}^{\beta,R}$
    is the corresponding Fermi level solving \cref{eq:FL}. Then, 
    \[ 
        \lim_{\beta\to\infty}\ep_\mathrm{F}^{\beta,R} = \ep_\mathrm{F}^{\infty,R} \coloneqq 
            \begin{cases}
                \underline{\ep} &\text{if }N^{\infty,R}(u;\tfrac{1}{2}(\underline{\ep} + \overline{\ep})) > N_{e,R} \\
                \frac{1}{2}\left( \underline{\ep} + \overline{\ep}\right) &\text{if }N^{\infty,R}(u;\tfrac{1}{2}(\underline{\ep} + \overline{\ep})) = N_{e,R}\\
                \overline{\ep} &\text{if } N^{\infty,R}(u;\tfrac{1}{2}(\underline{\ep} + \overline{\ep})) < N_{e,R}  
            \end{cases}
    \] 
    where 
    $\underline{\ep} \coloneqq \arg\max_{\sigma(\Ham^R(u))} \{N^{\infty,R}(u;\cdot) \leq N_{e,R}  \}$
    and 
    $\overline{\ep} \coloneqq \argmin_{\sigma(\Ham^R(u))} \{ N^{\infty,R}(u;\cdot) \geq N_{e,R}  \}$.
\end{lemma} 
\begin{proof}
    This well-known result is elementary but, for the sake of completeness, a full proof is presented in \Cref{sec:appendix_FL}.
\end{proof}

This result suggests that we may formally define the zero Fermi-temperature Helmholtz free energy by considering the pointwise limit as $\beta\to\infty$ and fixing the Fermi level as in Lemma~\ref{lem:fermilevelconverge}. That is, we define
\begin{equation}
    \label{eq:Helmholtz_zero_temperature}
    \begin{split}   
        &E^{\infty,R}(u) \equiv E^{\infty,R}(u;\ep_\mathrm{F}^{\infty,R}(u)) 
            \coloneqq \sum_{s} \mathfrak{e}^\infty(\lambda_s;\ep_\mathrm{F}^{\infty,R})\\
        &\quad\text{where}\quad 
        \mathfrak{e}^\infty(\lambda;\tau) 
            \coloneqq 2 \lambda \chi_{(-\infty,\tau)}(\lambda) + \lambda \chi_{\{\tau\}}(\lambda).
    \end{split}
\end{equation} 

\subsubsection{Equilibiration of Nuclei Positions}
\label{sec:geometryeq}

For finite computational cells, $\Lambda_R$, $\beta \in (0,\infty]$ and fixed chemical potential, $\mu$, we consider the geometry optimisation problems
\begin{align}
    \label{problem:canonical_finitedomain}
    \tag{$\text{CE}^{{\beta},R}$}
    \overline{u} &\in \argmin\left\{E^{\beta,R}(u) \colon u \text{ satisfies \asNonInterR}
    \right\}, \quad \textup{and}\\
    \label{problem:grand_canonical_finitedomain}
    \tag{$\text{GCE}^{{\beta},{R}}_{\mu}$}
    \overline{u} &\in \argmin\left\{G^{\beta,R}(u;\mu) \colon u \text{ satisfies \asNonInterR}
    \right\}.
\end{align}
Here and throughout, ``$\argmin$'' denotes the set of local minimisers. We denote these problems by \refCE{problem:canonical_finitedomain}{\beta}{R} and \refGCE{problem:grand_canonical_finitedomain}{\beta}{R}{\mu} so that we can reference the problem and associated parameters using a single compact notation.

We will assume that solutions $\overline{u}$ to \refGCE{problem:grand_canonical_finitedomain}{\beta}{R}{\mu} are \textit{strongly stable} in the following sense:
\begin{align}
    \label{eq:strongly_stable_R}
		\Braket{\delta^2 G^{\beta,R}(\ubar;\mu)v,v} \geq c_0 \|Dv\|_{\ell^2_\ctGamma}^2 
\end{align}
for some positive constant $c_0>0$. In order to differentiate the energy, in the case $\beta = \infty$, we require $\mu \not\in \sigma(\Ham^R(\overline{u}))$.

\subsection{Thermodynamic Limit}
\label{sec:thermodynamiclimit}

In this section, we introduce the limiting model on an infinite point defect reference configuration for finite Fermi-temperature as formulated in \cite{ChenLuOrtner18} and state the analogous model at zero Fermi-temperature. 

\subsubsection{Band Structure of the Homogeneous Crystal}
\label{sec:homcrystal}

In the case $\Lambda = \Lambda^\mathrm{ref}$, we denote the Hamiltonian of the reference configuration by $\Ham^\mathrm{ref}$. That is, 
\[
    \big[\Ham^\mathrm{ref}\big]_{\ell k}^{ab} 
    \coloneqq h^{ab}_{\ell k}(\ell - k)\quad \forall\, \ell, k \in \Lambda^\mathrm{ref} 
    \text{ and } 1\leq a,b\leq \numorbitals.
\]	
By exploiting the translational symmetry of the reference configuration, we may apply the Bloch transform \cite{bk:kittel} to conclude that the spectrum of the reference Hamiltonian is banded in the sense that it is a union of at most finitely many spectral bands, see Appendix~\ref{sec:band_structure} for more details.

For the remainder of this paper, we assume that the system is an insulator:
\newcommand{\asGap}{\textup{\textbf{(GAP)}}}
\begin{assumption}[1.95cm]{\asGap}
	There is a gap in the reference spectrum, $\sigma(\Ham^\mathrm{ref})$, and the chemical potential, $\mu$, is fixed in the interior of this gap. 
\end{assumption}

\subsubsection{Energy Spaces of Displacements}
\label{sec:energy_spaces_of_displacements}

Following \cite{ChenNazarOrtner19,ChenLuOrtner18,ChenOrtner16,ChenOrtnerThomas2019}, we introduce a space of finite energy displacements which restricts the class of admissible configurations. 

Given $\ell \in \Lambda$ and $\rho \in \Lambda - \ell$, we define the finite difference $D_\rho u(\ell) \coloneqq u(\ell + \rho) - u(\ell)$. The full (infinite) finite difference stencil is then defined to be $Du(\ell) \coloneqq \left( D_\rho u(\ell)\right)_{\rho\in\Lambda - \ell}$ and for $\ctGamma > 0$, we define the $\ell^2_\ctGamma$ semi-norm by
\[ 
	\|Du\|_{\ell^2_\ctGamma} \coloneqq 
		\bigg(\sum_{\ell \in \Lambda} \sum_{\rho\in\Lambda - \ell} e^{-2\ctGamma|\rho|}|D_\rho u(\ell)|^2\bigg)^{1/2}. 
\]
Since all of the semi-norms $\|D\cdot\|_{\ell^2_\ctGamma}$, for $\ctGamma > 0$, are equivalent \cite{ChenNazarOrtner19}, we will fix $\ctGamma>0$ for the remainder of this paper and define the following function space of finite energy displacements:
\[ 
	\W(\Lambda) \coloneqq 
	\{ u \colon \Lambda \to \mathbb R^d \colon \|Du\|_{\ell^2_\ctGamma} < \infty \} \textup{ with semi-norm } \|D\cdot\|_{\ell^2_\ctGamma}.
\]

For $u \in \W(\Lambda)$, it is possible to approximate $\Ham(u)$ by a finite rank update of $\Ham^\mathrm{ref}$ which leads to the following result describing a decomposition of $\sigma(\Ham(u))$ \cite[Lemma~3]{ChenOrtnerThomas2019}:
\begin{lemma}[Decomposition of the Spectrum]
\label{lem:spectrum}
    Suppose that $u \in \W(\Lambda)$ is a displacement satisfying \asNonInter. Then, $\sigma_\mathrm{ess}(\Ham(u)) = \sigma(\Ham^\mathrm{ref})$. Moreover, outside any neighbourhood of $\sigma(\Ham^\mathrm{ref})$, there are at most finitely many isolated eigenvalues: for all $\delta > 0$, there exists $S_\delta > 0$ such that 
    \[
        \#\left(\sigma(\Ham(u)) \setminus B_\delta(\sigma(\Ham^\mathrm{ref}))\right) \leq S_\delta.
    \]
\end{lemma}

\subsubsection{Grand Potential Difference Functional}
\label{sec:grandpotdiff}

For a finite system, $\Lambda_R$, $\beta \in (0,\infty]$, and $u \colon \Lambda_R \to \mathbb R^d$ satisfying \asNonInterR~we may write
\begin{align}
	\label{eq:site_decomp}
	\begin{split}
		&G^{\beta,R}(u) = 
			\sum_{\ell \in \Lambda_R} G^{\beta,R}_\ell(u) \quad \text{where} \quad G_\ell^{\beta,R}(u) \coloneqq 
			\sum_{s} \mathfrak{g}^{\beta}(\lambda_s;\mu)\sum_a[\psi_s]_{\ell a}^2.
	\end{split}
\end{align}

By \cite[Theorem~10]{ChenLuOrtner18}, the site energies defined in \cref{eq:site_decomp} can be extended to the infinite domain, $\Lambda$, by taking a sequence of computational cells, $\Lambda_R$, and defining the site energy as the limit along this sequence: $G_\ell^\beta(u) = G_{\ell}^{\beta,\infty}(u) \coloneqq \lim_{R\to\infty} G_\ell^{\beta,R}(u|_{\Lambda_R})$.

It will be convenient to rewrite the site energies as a function of the full interaction stencil, $Du(\ell)$. This can be done since the site energies inherit the translational invariance from the Hamiltonian operators (as in \cite{ChenOrtner16,ChenLuOrtner18}):
\begin{align}\label{def:site_energy_stencil}
    \mathcal G^{\beta, R}_\ell(Du(\ell)) \coloneqq G^{\beta, R}_\ell(u) 
\end{align}
for each $\beta \in (0,\infty]$ and $R \in (0,\infty]$. In the case where $R = \infty$, we simply write $\mathcal G^{\beta}_\ell \coloneqq \mathcal G^{\beta, \infty}_\ell$.

With this definition in hand, we may renormalise the total energy and, formally at first, define the following grand potential difference functional: for $\beta \in (0,\infty]$,
\begin{equation}\label{eq:grand_pot_diff}
    \mathcal G^\beta(u) \coloneqq \sum_{\ell \in \Lambda}
    \bigg[ \mathcal G^\beta_\ell( Du(\ell) ) - \mathcal G^\beta_\ell(\bm{0}) \bigg].
\end{equation}
Strong locality estimates for the site energies \cite[Theorems~4 and 5]{ChenOrtnerThomas2019}, together with the results of \cite{ChenNazarOrtner19} allow us to conclude that the grand potential difference functional is well defined on the following space of admissible displacements:
\begin{equation}
	\label{eq:adm_zero_temp}
	\mathrm{Adm}(\Lambda) \coloneqq \{ u \in \dot{\mathscr{W}}^{1,2}(\Lambda) \textup{ satisfying } \asNonInter \colon \mu \not\in \sigma(\Ham(u)) \}.
\end{equation}
We may therefore consider the following problems: for $\beta \in (0, \infty]$,
\begin{align}
	\label{problem:infinite_domain}
	\tag{$\mathrm{GCE}_\mu^{\beta,\infty}$}
	\overline{u} &\in 
		\argmin \{ \mathcal G^\beta(u) \colon u \in \mathrm{Adm}(\Lambda) \}.
\end{align}
In the main results of this paper, we will assume that solutions $\overline{u}$ to \refGCE{problem:infinite_domain}{\beta}{\infty}{\mu} are \textit{strongly stable} in the following sense:
\begin{align}
    \label{eq:strongly_stable}
		\Braket{\delta^2 \mathcal G^\beta(\ubar)v,v} \geq c_0 \|Dv\|_{\ell^2_\ctGamma}^2 
		\quad \forall \,v \in \dot{\mathscr W}^{1,2}(\Lambda)
\end{align}
for some positive constant $c_0>0$.

For $\beta \in (0,\infty]$, we denote by $\mathcal G^\beta_\mathrm{ref}$ the reference grand potential which is given by \cref{eq:grand_pot_diff} but with $\Lambda$ replaced with $\Lambda^\mathrm{ref}$. We assume that the reference configuration is an equilibrium state and stable in the sense that: there exists ${c}_\mathrm{stab} > 0$ such that
\begin{align}
\label{eq:reference_stability}
\delta\mathcal G^\beta_\mathrm{ref}(\bm{0}) = 0 
\quad \textup{and} \quad
\Braket{\delta^2\mathcal G^\beta_\mathrm{ref}(\bm{0})v,v} \geq c_\mathrm{stab} \|Dv\|_{\ell^2_\ctGamma(\Lambda^\mathrm{ref})}^2 \quad\forall\, v \in \dot{\mathscr W}^{1,2}(\Lambda^\mathrm{ref}).
\end{align}

We will often drop the superscript in the site energy and grand potential difference functional in the case of zero Fermi-temperature: that is, $\mathcal G_\ell = \mathcal G_\ell^\infty$ and $\mathcal G = \mathcal G^\infty$. 

\subsection{Main Results}
\label{sec:main_results}

The results of this paper can be summarised in Figure~\ref{fig:cd-1}. Each arrow represents the following mathematical statements: \textit{(i) Strong limit:} for any strongly stable solution to the limit problem, there exists a sequence of solutions to the finite domain or finite temperature problem that converges to the solution to the limit problem; \textit{(ii) Weak limit:} for any bounded sequence of solutions to the finite domain or finite temperature problem, there is a weak limit along a subsequence that is a critical point of the limiting energy functional.

This paper generalises the results of \cite{ChenLuOrtner18} to the zero temperature case but also to the case where $\Lambda^\mathrm{ref}$ is not necessarily a Bravais lattice. 

When we say that the limiting model is given by a grand-canonical model, we have to be careful in specifying the limiting chemical potential. For example, the limit of \refCE{problem:canonical_finitedomain}{\beta}{R} as $R\to \infty$ for $\beta < \infty$ is given by \refGCE{problem:infinite_domain}{\beta}{\infty}{\mu_\#} where $\mu_\#$ is the Fermi level for the homogeneous crystal \cite[Theorems~A.2 and A.3]{ChenLuOrtner18}. This subtlety means that Figure~\ref{fig:cd-1} \emph{cannot} be seen as a commutative diagram.

We also stress that these results only hold in the case where $\mu \not\in\sigma(\Ham(\ubar))$. This is simply because the site energies and hence the grand potential difference functional are not differentiable if the chemical potential is an eigenvalue. In this case the main techniques used in this paper cannot be applied. We will study the consequences of this restriction in future work.

\subsubsection{Zero Temperature Limit}

First, we state the zero Fermi-temperature limit result for the grand canonical ensemble model:
\begin{theorem}
[Strong Zero Temperature Limit,
\refGCE{problem:infinite_domain}{\beta}{\infty}{\mu}
$\to$ 
\refGCE{problem:infinite_domain}{\infty}{\infty}{\mu}]\label{thm:0T_limit_infinite_domain}
Suppose that $\mu$ satisfies \asGap~and $\ubar$ is a solution to
\refGCE{problem:infinite_domain}{\infty}{\infty}{\mu}
that is strongly stable \cref{eq:strongly_stable}. Then, there exist solutions $\ubar_\beta$ to 
\refGCE{problem:infinite_domain}{\beta}{\infty}{\mu}
such that
	\[
		\| D\ubar_\beta - D\ubar \|_{\ell^2_\ctGamma} \leq C e^{- \frac{1}{12}\ctDist\beta} \quad \textup{and} \quad 
		|\mathcal G^\beta(\overline{u}_\beta) - \mathcal G(\overline{u})| \leq C e^{-\frac{1}{12}\mathsf{d}\beta}
	\]
	where $\ctDist \coloneqq \mathrm{dist}(\mu, \sigma(\Ham(\overline{u})))$. 
\end{theorem}
\begin{prop}[Weak Zero Temperature Limit,
\refGCE{problem:infinite_domain}{\beta}{\infty}{\mu}
$\to$ 
\refGCE{problem:infinite_domain}{\infty}{\infty}{\mu}]
\label{prop:0T_limit_infinite_domain}
	Suppose that $\ubar_{\beta_j}$ is a bounded sequence (with $\beta_j \to \infty$) of solutions to
	\refGCE{problem:infinite_domain}{\beta_j}{\infty}{\mu}
	each satisfying \asNonInter~with an accumulation parameter uniformly bounded below by $\mathfrak{m}>0$ and such that $\mu$ is (eventually) uniformly bounded away from $\sigma(\Ham(\overline{u}_{\beta_j}))$. Then, there exists $\ubar\in\W(\Lambda)$ such that along a subsequence
	\begin{equation}
	\label{eq:theorem1part2_1}
	    D_\rho\ubar_{\beta_j}(\ell) \to D_\rho\ubar(\ell) 
	    \quad \forall \ell \in \Lambda,\,\rho \in \Lambda - \ell.
	\end{equation}
	Moreover, $\ubar$ is a critical point of $\mathcal G$.
\end{prop}
The exact same arguments can be made in the finite domain case. That is, every strongly stable solution to
\refGCE{problem:grand_canonical_finitedomain}{\infty}{R}{\mu}
is an accumulation point of a sequence of solutions to
\refGCE{problem:grand_canonical_finitedomain}{\beta}{R}{\mu}. 
Moreover, for every bounded sequence of solutions to 
\refGCE{problem:grand_canonical_finitedomain}{\beta}{R}{\mu},
with spectrum uniformly bounded away from the chemical potential, up to a subsequence, there exists a weak limit. Moreover, the limit is a critical point of the limiting functional.

Further, we have an analogous result in the canonical ensemble. Here we only consider $R<\infty$ since $(\textup{CE}^{\infty,\infty})$ is not well defined.
\begin{theorem}[Strong Zero Temperature Limit, \refCE{problem:canonical_finitedomain}{\beta}{R}
$\to$ 
\refCE{problem:canonical_finitedomain}{\infty}{R}]
\label{cor:strong_0T_limit_infinite_domain}
Suppose $\ubar_R$ is a solution to 
\refCE{problem:canonical_finitedomain}{\infty}{R}
with $\mu \coloneqq \ep_\mathrm{F}^{\infty,R}(\overline{u}_R) \not\in \sigma(\Ham^R(\overline{u}_R))$ and such that \cref{eq:strongly_stable_R} is satisfied with the chemical potential $\mu$ and $\beta = \infty$. Then, there exist solutions $\ubar_{R,\beta}$ to 
\refCE{problem:canonical_finitedomain}{\beta}{R}
such that
	\[
		\| D\ubar_{R,\beta} - D\ubar_{R} \|_{\ell^2_\ctGamma} + 
		\big| \ep_\mathrm{F}^{\beta,R}(\ubar_{R,\beta}) - \ep_\mathrm{F}^{\infty,R}(\ubar_{R}) \big| 
		    \leq C e^{-\frac{1}{12}\ctDist\beta}
	\]
	where $\ctDist = \mathrm{dist}(\mu, \sigma(\Ham^R(\overline{u}_R)))$.
\end{theorem}
\begin{prop}[Weak Zero Temperature Limit, \refCE{problem:canonical_finitedomain}{\beta}{R}
$\to$ 
\refCE{problem:canonical_finitedomain}{\infty}{R}]
\label{cor:weak_0T_limit_infinite_domain}
    Suppose that $\ubar_{R,\beta_j}$ is a bounded sequence (with $\beta_j \to \infty$) of solutions to 
    \refCE{problem:canonical_finitedomain}{\beta_j}{R}
    each satisfying \asNonInterR~with an accumulation parameter uniformly bounded below by $\mathfrak{m}>0$. Then, there exists $\ubar_R \colon \Lambda_R \to \mathbb R^d$ such that along a subsequence $\overline{u}_{R,\beta_j} \to \overline{u}_R$ and $\ep_\mathrm{F}^{\beta_j,R}(\overline{u}_{R,\beta_j}) \to \ep_\mathrm{F}^{\infty,R}(\overline{u}_R)$ as $j \to \infty$.
    %
	%
	
	Moreover, if 
	$\ep_\mathrm{F}^{\beta_j,R}(\overline{u}_{R,\beta_j})$ is (eventually) uniformly bounded away from $\sigma(\Ham^R(\overline{u}_{R,\beta_j}))$,
	%
	%
	then $\ubar_R$ is a critical point of $G^{\infty,R}(\,\cdot\,;\mu)$ with $\mu \coloneqq \ep_\mathrm{F}^{\infty,R}(\overline{u}_R)$.
\end{prop}

\begin{remark}[Convergence Rates]
\label{rem:0T_convergence_rates}
The convergence rates in Theorems~\ref{thm:0T_limit_infinite_domain} and~\ref{cor:strong_0T_limit_infinite_domain} are obtained by a suitable consistency estimate and application of the inverse function theorem. For example, in the $R<\infty$ case, we may use the fact that $\overline{u}$ is an equilibrium state for $\beta = \infty$ to conclude that,
\begin{align}\label{eq:rates_R}\begin{split}
    \left|\frac{\partial G^{R,\beta}(\overline{u})}{\partial u(\ell)}\right|
    &\leq \sum_{s\colon \lambda_s < \mu} 
    2\left( 1- 
    f_\beta(\lambda_s - \mu) \right)
    \left|\frac{\partial \lambda_s(\overline{u})}{\partial \overline{u}(\ell)}\right| + 
    \sum_{s\colon \lambda_s > \mu} 
    2 
    f_\beta(\lambda_s - \mu)    \left|\frac{\partial \lambda_s(\overline{u})}{\partial \overline{u}(\ell)}\right|
\end{split}\end{align}
where $\lambda_s = \lambda_s(\overline{u})$ is some enumeration of $\sigma(\Ham(\overline{u}))$. The dominant contribution in \cref{eq:rates_R} is exponentially small in the distance from the closest eigenvalue to the chemical potential. Unless there is significant cancellation in this summation (which we have no reason to expect), this simple calculation suggests that the convergence rate obtained in Theorems~\ref{thm:0T_limit_infinite_domain} and~\ref{cor:strong_0T_limit_infinite_domain} depend on the defect states within the band gap and can be no better than a constant multiple of $\mathrm{dist}(\mu, \sigma\left(\Ham(\overline{u}))\right)$. 
\end{remark}

\subsubsection{Thermodynamic Limit}

Now we move on to consider the thermodynamic limit results. 

\begin{theorem}[Strong Thermodynamic Limit, 
\refGCE{problem:grand_canonical_finitedomain}{\infty}{R}{\mu} $\to$ \refGCE{problem:infinite_domain}{\infty}{\infty}{\mu}]
\label{thm:0T_thermodynamic_limit}
Suppose that $\mu$ is fixed such that \asGap~is satisfied and $\ubar$ is a solution to
\refGCE{problem:infinite_domain}{\infty}{\infty}{\mu} 
that is strongly stable \cref{eq:strongly_stable}. Then, there exist solutions $\ubar_R$ to
\refGCE{problem:grand_canonical_finitedomain}{\infty}{R}{\mu}
such that $\ubar_R \to \ubar$ in $\dot{\mathscr W}^{1,2}(\Lambda)$.
\end{theorem}
\begin{prop}[Weak Thermodynamic Limit, 
\refGCE{problem:grand_canonical_finitedomain}{\infty}{R}{\mu} $\to$ \refGCE{problem:infinite_domain}{\infty}{\infty}{\mu}]
\label{prop:0T_thermodynamic_limit}
Suppose that $\ubar_{R_j}$ is a bounded sequence (with $R_j \to \infty$) of solutions to
\refGCE{problem:grand_canonical_finitedomain}{\infty}{R_j}{\mu}
each satisfying \textup{\textbf{($\textrm{L}_{R_j}$)}} with an accumulation parameter uniformly bounded below by $\mathfrak{m}>0$ and such that $\mu$ is uniformly bounded away from $\sigma(\Ham^{R_j}(\overline{u}_{R_j}))$. Then, there exists $\ubar\in\dot{\mathscr{W}}^{1,2}(\Lambda)$ such that along a subsequence
	\begin{equation}
	    D_\rho\ubar_{R_j}(\ell) \to D_\rho\ubar(\ell) \quad \forall \ell \in \Lambda,\,\rho \in \Lambda - \ell.
	\end{equation}
	Moreover, $\ubar$ is a critical point of $\mathcal G$.
\end{prop}

We now turn our attention to the thermodynamic limit of the canonical model. Here, we see that the prescribed number of particles in the sequence of finite domain approximations is vital in identifying a limiting model.
\begin{theorem}[Strong Thermodynamic Limit, 
\refCE{problem:canonical_finitedomain}{\infty}{R} $\to$ \refGCE{problem:infinite_domain}{\infty}{\infty}{\mu}]
\label{thm:0T_thermodynamic_limit_CE_GCE}
	Suppose that $\mu$ is fixed such that \asGap~is satisfied and $\ubar$ is a solution to \refGCE{problem:infinite_domain}{\infty}{\infty}{\mu} that is strongly stable \cref{eq:strongly_stable}. Then, there exists a sequence $N_{e,R}$ and solutions $\ubar_R$ to \refCE{problem:canonical_finitedomain}{\infty}{R} such that $\ubar_R \to \ubar$ in $\dot{\mathscr W}^{1,2}(\Lambda)$. 
		
	Moreover, $\ep_\mathrm{F}^{\infty,R}(\ubar_R) \to \nu$ as $R\to\infty$ for some $\nu \in \mathbb R$ and $\overline{u}$ is also a solution to \refGCE{problem:infinite_domain}{\infty}{\infty}{\nu}. 
\end{theorem}
\begin{prop}[Weak Thermodynamic Limit, 
\refCE{problem:canonical_finitedomain}{\infty}{R} $\to$ \refGCE{problem:infinite_domain}{\infty}{\infty}{\mu}]
\label{prop:0T_thermodynamic_limit_CE_GCE}
    Suppose that $\ubar_{R_j}$ is a bounded sequence of solutions to 
    \refCE{problem:canonical_finitedomain}{\infty}{R}
    each satisfying \textup{\textbf{(L$_{R_j}$)}}~with an accumulation parameter uniformly bounded below by $\mathfrak{m}>0$. 
    Then, there exists $\overline{u} \in \dot{\mathscr W}^{1,2}(\Lambda)$ and $\mu \in \mathbb{R}$ such that along a subsequence 
    \[
    D_\rho \overline{u}_{R_j}(\ell) \to D_\rho \overline{u}(\ell) \quad \forall\,\ell \in \Lambda, \rho \in \Lambda - \ell
    \quad \textup{and} \quad 
    \ep_\mathrm{F}^{\infty, R_j}(\overline{u}_{R_j}) \to \mu.
    \]

    Moreover, if $\ep_\mathrm{F}^{\infty,R_j}(\overline{u}_{R_j})$ is (eventually) uniformly bounded away from $\sigma(\Ham^{R_j}(\overline{u}_{R_j}))$, then 
	%
	%
	$\ubar$ is a critical point of $\mathcal G(\,\cdot\,;\mu)$. 
\end{prop}

\begin{remark}
Every strongly stable solution $\overline{u}$ to \refGCE{problem:infinite_domain}{\infty}{\infty}{\mu} also solves \refGCE{problem:infinite_domain}{\infty}{\infty}{\nu} for all $\nu$ in some maximal interval $(\underline{\nu},\overline{\nu})$, displayed in Figure~\ref{fig:spectrum}. In particular, $\overline{u}$ is a strongly stable solution to \refGCE{problem:infinite_domain}{\infty}{\infty}{\nu} with $\nu \coloneqq \tfrac{1}{2}(\underline{\nu} + \overline{\nu})$. Theorem~\ref{thm:0T_thermodynamic_limit_CE_GCE} states that, under some appropriate choice of particle number, there exist solutions $\overline{u}_R$ solving \refCE{problem:canonical_finitedomain}{\infty}{R} such that $\overline{u}_R \to \overline{u}$ and $\ep_\mathrm{F}^{\infty,R}(\overline{u}_R) \to \nu$ as $R\to\infty$.

\begin{figure}[htbp]
	\centering
	\resizebox{\columnwidth} {!} {
		\tikzset{every picture/.style={line width=0.75pt}} 

\begin{tikzpicture}[x=0.75pt,y=0.75pt,yscale=-1,xscale=1]

\draw  [fill={rgb, 255:red, 0; green, 0; blue, 0 }  ,fill opacity=1 ] (31,41) -- (236.8,41) -- (236.8,55.2) -- (31,55.2) -- cycle ;
\draw  [fill={rgb, 255:red, 0; green, 0; blue, 0 }  ,fill opacity=1 ] (415.8,41) -- (621.6,41) -- (621.6,55.2) -- (415.8,55.2) -- cycle ;
\draw [line width=3.75]    (267.8,38.2) -- (268,57) ;

\draw [line width=3.75]    (355.8,38.2) -- (356,57) ;

\draw  [dash pattern={on 4.5pt off 4.5pt}]  (283.8,20.2) -- (284,75) ;

\draw  [dash pattern={on 4.5pt off 4.5pt}]  (309.8,20.2) -- (310,75) ;

\draw (284,80) node   {$\mu $};
\draw (311,80) node   {$\nu $};
\draw (269,32) node   {$\underline{\nu }$};
\draw (357,31.2) node   {$\overline{\nu }$};

\end{tikzpicture} }%
	\caption{Cartoon depicting an approximation of $\sigma(\Ham(\overline{u}))$ together with the limiting chemical potential, $\nu$, from Theorem~\ref{thm:0T_thermodynamic_limit_CE_GCE}.}
	\label{fig:spectrum}
\end{figure}
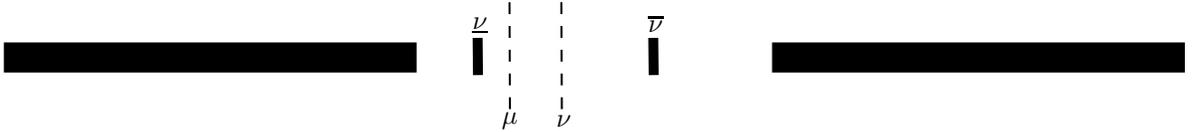

We are not implying that any of the problems \refGCE{problem:infinite_domain}{\infty}{\infty}{\nu} are equivalent for $\nu \in (\underline{\nu},\overline{\nu})$, only that they are locally equivalent around the fixed displacement $\overline{u}$. This is simply because there are no eigenvalues between $\underline{\nu}$ and $\overline{\nu}$, as depicted in Figure~\ref{fig:spectrum}, and we are considering the case of zero Fermi-temperature.
\end{remark}

\begin{remark}[Convergence Rates]
\label{rem:convergence_rates}
Since we are considering the case where $\Lambda^\mathrm{ref}$ is more general than a Bravais lattice $\mathrm{B}\mathbb Z^d$, we do not prove any convergence rates for $R\to\infty$. In the case where $\Lambda^\mathrm{ref} = \mathsf{B}\mathbb Z^d$, strongly stable solutions $\overline{u}$ to \refGCE{problem:infinite_domain}{\infty}{\infty}{\mu} satisfy the following far field decay estimate \cite{ChenNazarOrtner19}:
\begin{align}\label{eq:far_field_decay}
    \bigg(\sum_{\rho \in \Lambda - \ell} e^{-2\ctGamma|\rho|}|D_\rho \overline{u}(\ell)|^2\bigg)^{1/2} 
    \lesssim (1 + |\ell|)^{-d} \quad \text{for all }\ell \in \Lambda.
\end{align}
If the estimate \cref{eq:far_field_decay} holds in the case where $\Lambda^\mathrm{ref} \not = \mathsf{B}\mathbb Z^d$, we can simply repeat the proofs of Theorem~\ref{thm:0T_thermodynamic_limit} and~\ref{thm:0T_thermodynamic_limit_CE_GCE} verbatim and obtain the following convergence rate:
\[
    \|D\overline{u}_R - D\overline{u}\|_{\ell^2_\ctGamma} \lesssim R^{-d/2}.
\]
For finite interaction range models and in the case of multilattices, far-field decay estimates of the form \cref{eq:far_field_decay} are satisfied \cite{OlsonOrtner2017}. In light of \cite{ChenNazarOrtner19}, which extends \cite{EhrlacherOrtnerShapeev16} to infinite interaction range models, it is safe to assume that \cref{eq:far_field_decay} can be extended to our setting.  
\end{remark}
\begin{remark}
In the weak convergence results (Propositions~\ref{prop:0T_limit_infinite_domain},~\ref{cor:weak_0T_limit_infinite_domain},~\ref{prop:0T_thermodynamic_limit}~and~\ref{prop:0T_thermodynamic_limit_CE_GCE}) we assume that the chemical potential (or Fermi level) is uniformly bounded away from the spectrum. By the spectral pollution results (Lemma~\ref{lem:spectral_pollution}, below), this implies that the chemical potential (or the limit of the Fermi level) is not in the limiting spectrum. We opt to make the assumption on the sequence of solutions rather than imposing a condition on the (a priori unknown) weak limit.
\end{remark}

\section{Conclusions}
\label{sec:conclusions}
In this paper, we have formulated the zero Fermi-temperature limit models for geometry relaxation problems in the context of simple two-centre linear tight binding models for point defects and quantified an exponential rate of convergence for the nuclei positions.

Further, we have extended the results of \cite{ChenLuOrtner18} to the case of zero Fermi-temperature under the assumption that the chemical potential is not an eigenvalue of the Hamiltonian. That is, we have formulated zero Fermi-temperature models in the grand canonical ensemble for the electrons for general point defects. We have shown that, under an assumption on the number of electrons imposed in the sequence of finite domain approximations, this is a limiting model as domain size is sent to infinity in a tight binding model in the canonical ensemble for the electrons and at zero Fermi-temperature. In contrast to the finite Fermi-temperature results of \cite{ChenLuOrtner18}, we have shown that a specific choice of electron number in the sequence of finite domain approximations is crucial in identifying the limiting model.

A consequence of these results is that, in general, the zero Fermi-temperature and thermodynamic limits of the geometry optimisation problem do not commute. In particular, taking the thermodynamic limit first, we obtain a limiting model with fixed chemical potential at the reference domain level. On the other hand, if we take the zero Fermi-temperature limit first, the limit model is a grand canonical model but the fixed chemical potential depends on the sequence of solutions to the finite domain problems. The limit of the Fermi levels depends on the polluted band structure and so there is no reason why the limiting Fermi level agrees with the reference Fermi level.

We stress again here that the weak convergence results of Propositions~\ref{prop:0T_limit_infinite_domain},~\ref{cor:weak_0T_limit_infinite_domain},~\ref{prop:0T_thermodynamic_limit}~and~\ref{prop:0T_thermodynamic_limit_CE_GCE} are weaker than the analogous results in the finite Fermi-temperature case \cite{ChenLuOrtner18}. Indeed, by assuming the chemical potential (or sequence of Fermi levels) is bounded away from the spectrum, we are ensuring that the limiting Fermi level is not an eigenvalue of the Hamiltonian, an assumption that we cannot justify in general. We do this to ensure differentiability of the limiting site energies which is required to define the zero Fermi-temperature grand potential difference functional. Exploring the consequences of lifting this technical assumption is beyond the scope of this paper.

Thus, we have completed the diagram in Figure~\ref{fig:cd-1} however, we must reiterate here that some care is needed in order to interpret this diagram correctly.

\section{Acknowledgements}
Stimulating discussions with Huajie Chen are gratefully acknowledged.

\section{Proofs}\label{sec:proofs}
Throughout this section we will use the following notation: for $\beta \in (0,\infty]$, $R \in (0,\infty]$, $\ell \in \Lambda$, $\bm{m} = (m_1,\dots,m_j) \in \Lambda^j$ and $\bm{\rho} = (\rho_1,\dots,\rho_j) \in (\Lambda - \ell)^j$, we write
\begin{align*}
    \Ham_{,\bm{m}}(u) &\coloneqq \frac{\partial^j \Ham(u)}{\partial u(m_1)\dots \partial u(m_j)},
    \quad \textrm{and} \\
    \mathcal G_{\ell, \bm{\rho}}^{\beta,R}(Du(\ell)) 
    &\coloneqq \frac{\partial^j \mathcal G_\ell^{\beta,R}(Du(\ell))}{\partial D_{\rho_1}u(\ell)\dots\partial D_{\rho_j}u(\ell)} 
    = \frac{\partial^j G_\ell^{\beta,R}(u)}{\partial u(\ell+\rho_1)\dots\partial u(\ell+\rho_j)} 
\end{align*}
and similarly for $\Ham^R_{,\bm{m}}(u)$ (with appropriate $\bm{m}$ and $u$).

\subsection{Preliminaries}

In order to prove the main theorems of this paper, we require some preliminary results. First, we state perturbation results for the Hamiltonian operators, resulting in the description of the corresponding spectra. We then move on to explicitly define the site energies by use of resolvent calculus and state a Combes-Thomas estimate for the resolvent operators. Finally, we show that, because we are considering periodic boundary conditions, spectral pollution cannot happen.

\subsubsection{Spectrum of the Hamiltonian}

In this section, we describe $\sigma(\Ham^\mathrm{ref})$, $\sigma(\Ham(u))$ and $\sigma(\Ham^R(u))$ for admissible displacements $u$. 

We first show that small perturbations in the atomic configuration results in small perturbations in the Hamiltonian (in the sense of Frobenius norm) and thus also small perturbations in the corresponding spectra (in the sense of Hausdorff distance).

We start with the case that $R < \infty$:

\begin{lemma}\label{ham_converge_R}
	Suppose $u_1, u_2 \colon \Lambda_R \to \mathbb R^d$ satisfy \asNonInter~for both $l = 1,2$ and some $\mathfrak m >0$. Then, for $\ell,k \in \Lambda_R$, if $|D_{k-\ell}(u_1 - u_2)(\ell)| \leq \mathfrak{m}|\ell - k|$, then
	\begin{align} 
	\label{eq:ham_diff_1_R}
	\begin{split}
	\left|\left[\Ham^R(u_1) - \Ham^R(u_2)\right]_{\ell k}^{ab}\right| 
	    &\leq  C e^{-c\gamma_0 \min_{\alpha\in\mathbb Z^d}|\ell - k + \mathsf{M}_R\alpha|}
	        \left| D_{k-\ell}(u_1-u_2)(\ell)\right|
	\end{split}
	\end{align}
	where $c = \tfrac{\mathfrak{m}\sqrt{3}}{4}$. 
	
	Moreover, if $\|D(u_1 - u_2)\|_{\ell^2_\ctGamma}$ is sufficiently small, then 
	\begin{align} 
	\label{eq:ham_diff_2_R}
	\begin{split}
	\mathrm{dist}\big(\sigma(\Ham^R(u_1)), \sigma(\Ham^R(u_2))\big)
	\leq 
	\|\Ham^R(u_1) - \Ham^R(u_2)\|_\mathrm{F} 
	    &\leq  C \| D(u_1-u_2)\|_{\ell^2_\ctGamma}.
	\end{split}
	\end{align}
\end{lemma}
\begin{proof}
After extending $\Ham^R$ by periodicity, we may assume $|\ell - k + \mathsf{M}_R\alpha| \geq |\ell - k|$ for all $\alpha \in \mathbb Z^d$. Applying Taylor's theorem we can conclude that there exists %
$\xi^\alpha = (1 - \theta)( \bm{r}_{\ell k}(u_1) + \mathsf{M}_R\alpha) + \theta( \bm{r}_{\ell k}(u_2) + \mathsf{M}_R\alpha)$
for some $\theta = \theta(\alpha, a,b,\ell,k) \in [0,1]$ such that
\begin{align*}
    \left|\left[\Ham^R(u_1) - \Ham^R(u_2)\right]_{\ell k}^{ab}\right| 
    &= \bigg|\sum_{\alpha \in \mathbb Z^d}\nabla h_{\ell k}^{ab}(\xi^\alpha) \cdot D_{k - \ell}(u_1 - u_2)(\ell)\bigg| \\
    &\leq \ctTBprefactor \sum_{\alpha \in \mathbb Z^d} e^{-\gamma_0|\xi^\alpha|}|D_{k - \ell}(u_1 - u_2)(\ell)|.
\end{align*}
Since $|\bm{r}_{\ell k}(u_l) + \mathsf{M}_R\alpha| \geq \mathfrak{m}|\ell - k|$ for $l = 1,2$ and $|D_{k-\ell}(u_1 - u_2)(\ell)| \leq \mathfrak{m}|\ell - k|$ we can conclude that $|\xi^\alpha| \geq \tfrac{\mathfrak{m}\sqrt{3}}{2}|\ell - k|$ (here, we have used the following: if $x,y \in \mathbb R^d$ with $|x|, |y| \geq r$ and $|x-y| \leq r$, then $|tx + (1-t)y| \geq \sqrt{r^2 - (\frac{r}{2})^2} = \tfrac{\sqrt{3}}{2}r$ for all $t \in [0,1]$). Therefore, after summing over $\alpha \in \mathbb Z^d$, we obtain \cref{eq:ham_diff_1_R}.

We suppose that $\|D(u_1 - u_2)\|_{\ell^2_\ctGamma}$ is sufficiently small such that 
    $|D_\rho (u_1 - u_2)(\ell)| \leq \mathfrak{m} |\rho|$
for all $\ell \in \Lambda$ and $\rho \in \Lambda - \ell$. This can be done as the semi-norm defined by $\sup_{\ell,\rho} |D_\rho v(\ell)| / |\rho|$ is equivalent to $\|D\cdot\|_{\ell^2_\ctGamma}$ \cite{ChenNazarOrtner19}.

We extend $\Ham^R$ and $u_1,u_2$ by periodicity and so, for each $\ell \in \Lambda_R$, we can sum over the set $\Lambda_R(\ell)$ of all $k \in \bigcup_{\alpha} (\Lambda_R + \mathsf{M}_R\alpha)$ for which $|\ell - k| = \min_{\alpha} |\ell - k + \mathsf{M}_R\alpha|$:
\begin{align}\label{eq:Ham_converge_trick}\begin{split}
    \|\Ham^R(u_1) - \Ham^R(u_2)\|_\mathrm{F}^2 
    &\leq C \sum_{\ell \in \Lambda_R} \sum_{k \in \Lambda_R(\ell)} e^{-2c\gamma_0|\ell - k|}|D_{k-\ell}(u_1 - u_2)(\ell)|^2\\
    &\leq C \sum_{\ell \in \widetilde{\Lambda}_R} \sum_{k \in \widetilde{\Lambda}_R} e^{-2c\gamma_0|\ell - k|}|D_{k-\ell}(u_1 - u_2)(\ell)|^2 \\
    &\leq C \|D(u_1 - u_2)\|_{\ell^2_\ctGamma(\widetilde{\Lambda}_R)}^2 \leq C_d \|D(u_1 - u_2)\|_{\ell^2_\ctGamma(\Lambda_R)}^2 
\end{split}\end{align}
where $\widetilde{\Lambda}_R \coloneqq \bigcup_{\ell \in \Lambda_R} \Lambda_R(\ell)$.

The perturbation in the spectrum \cref{eq:ham_diff_2_R} follows directly from \cref{eq:ham_diff_1_R} since small perturbations in the Frobenius norm give rise to small perturbations in the spectrum \cite{Kato95}. 
\end{proof}

We now consider the $R = \infty$ case:
%
\begin{lemma}\label{ham_converge}
	Suppose $u_1, u_2 \in \mathrm{Adm}(\Lambda)$ satisfying \asNonInter~for both $l = 1,2$ and some $\mathfrak m >0$. Then, for $\ell,k,m \in \Lambda$, if $|D_{k-\ell}(u_1 - u_2)(\ell)| \leq \mathfrak{m}|\ell - k|$, we have 
	\begin{align} \label{eq:ham_diff_1}
	\begin{split}
	\left|\left[\Ham(u_1) - \Ham(u_2)\right]_{\ell k}^{ab}\right| 
	    &\leq  C e^{-c\gamma_0 |\ell - k|}
	        \left| D_{k-\ell}(u_1-u_2)(\ell)\right|, \quad \text{and}\\
	\left|\left[\Ham_{,m}(u_1) - \Ham_{,m}(u_2)\right]_{\ell k}^{ab}\right| 
	    &\leq  C e^{-c\gamma_0 (|\ell - m| + |m - k|)}
	        \left| D_{k-\ell}(u_1-u_2)(\ell)\right|.
	\end{split}
	\end{align}
	where $c = \frac{\mathfrak{m}\sqrt{3}}{2}$. 
	
	In particular, if $\|D(u_1 - u_2)\|_{\ell^2_\ctGamma}$ is sufficiently small, we have
	\begin{gather}
	\label{eq:Ham_diff_2}
	    \mathrm{dist}\left(\sigma(\Ham(u_1)), \sigma(\Ham(u_2))\right) \leq \|\Ham(u_1) - \Ham(u_2)\|_\mathrm{F} \leq C \|D(u_1 - u_2)\|_{\ell^2_{\ctGamma}}. 
	\end{gather}
\end{lemma}
\begin{proof}
    Using the same idea as in Lemma~\ref{ham_converge_R}, we have 
	\begin{gather*}
	    \big| \left[\Ham(u_1) - \Ham(u_2)\right]_{\ell k}^{ab} \big|
	        \leq \ctTBprefactor e^{-\ctTBexponent |\xi_0|} |D_{k-\ell}(u_1 - u_2)(\ell)| \quad \textrm{and}\\
	  \big|  \left[\Ham_{,m}(u_1) - \Ham_{,m}(u_2)\right]_{\ell k}^{ab} \big| 
	        \leq \ctTBprefactor e^{-\ctTBexponent |\xi_1|} |D_{k-\ell}(u_1 - u_2)(\ell)|
	\end{gather*}
	where $\xi_j = (1-\theta_j) \bm{r}_{\ell k}(u_1) +  \theta_j \bm{r}_{\ell k}(u_2)$ for some $\theta_j = \theta_j(a,b,\ell,k) \in [0,1]$ and both $j = 1,2$. Now, since ${r}_{\ell k}(u_l) \geq \mathfrak{m}|\ell - k|$ for both $l = 1,2$ and $|D_{k-\ell}(u_1-u_2)(\ell)| \leq \mathfrak{m}|\ell - k|$, we necessarily have that $|\xi_l| \geq \frac{\sqrt{3}}{2}\mathfrak{m}|\ell - k|$. Therefore, we obtain \cref{eq:ham_diff_1} and thus \cref{eq:Ham_diff_2} as in the proof of Lemma~\ref{ham_converge_R}.
\end{proof}

We next approximate the Hamiltonian $\Ham(u)$ by a finite rank update of the reference Hamiltonian $\Ham^\mathrm{ref}$. However, we must first redefine $\Ham(u)$ and $\Ham^\mathrm{ref}$ so that they are defined on the same spatial domain.

For $\ell, k \in \Lambda\cup\Lambda^\mathrm{ref}$, we define
\begin{align}\label{eq:extend_hamiltonian}
	\widetilde{\Ham}(u)_{\ell k}^{ab} \coloneqq 
		\begin{cases}
			\Ham(u)_{\ell k}^{ab} 	& \text{if } \ell, k \in \Lambda \\
			0	& \text{if } \ell \in \Lambda^\mathrm{ref} \setminus \Lambda \text{ or } k \in \Lambda^\mathrm{ref} \setminus \Lambda.
		\end{cases}
\end{align}
Similarly, for the reference Hamiltonian, we define 
\[	
	\big(\widetilde{\Ham}^\mathrm{ref}\big)_{\ell k}^{a b} \coloneqq 
	\begin{cases}
		\big(\Ham^\mathrm{ref}\big)_{\ell k}^{ab} 	& \text{if } \ell, k \in \Lambda^\mathrm{ref} \\
		0	& \text{if } \ell \in \Lambda \setminus \Lambda^\mathrm{ref} \text{ or } k \in \Lambda \setminus \Lambda^\mathrm{ref}.
	\end{cases}
\]	
The new operators defined in this way are obtained from the original ones by adding a finite number of zero columns and rows. This means that, apart from adding zero as an eigenvalue of finite multiplicity, the spectrum is unchanged. We may ignore this subtlety by shifting the spectrum away from $\{0\}$ (by artificially adding on a multiple of the identity to the Hamiltonian operators). A more detailed explanation of this idea can be found in \cite[\S4.3]{ChenOrtnerThomas2019}.

%
%
The following decomposition of the Hamiltonian has been shown in \cite[Lemma~9]{ChenOrtnerThomas2019}:
\begin{lemma}[Decomposition of the Hamiltonian, $R=\infty$]
\label{lem:decomp_Ham_infty}
Fix $u \in \mathrm{Adm}(\Lambda)$. Then, for all $\delta > 0$, there exists a constant, $R_\delta > 0$, and operators $P(u),Q(u)$ such that 
\begin{align}\label{eq:decomp_Ham}\begin{split}
   \widetilde{\Ham}(u) &= \widetilde{\Ham}^{\mathrm{ref}} + P(u) + Q(u),
\end{split}\end{align}
where $\|P(u)\|_\mathrm{F} \leq \delta$ and $Q(u)_{\ell k}^{ab} =0$ for all $(\ell, k) \not\in B_{R_\delta} \times B_{R_\delta}$. 
\end{lemma}
We now discuss the corresponding $R<\infty$ case which allows us to describe $\sigma(\Ham^R(u_R))$ (in Lemma~\ref{lem:perturbation_spec}, below) and the limiting spectrum as $R\to \infty$ (in Lemma~\ref{lem:spectral_pollution}, below):
\begin{lemma}[Decomposition of the Hamiltonian, $R<\infty$]
\label{lem:decomp_Ham_R}
Suppose that $u_R\colon \Lambda_R \to \mathbb R^d$ satisfies \asNonInterR, with some constant uniformly bounded below by $\mathfrak{m}>0$, and with $\sup_R \|Du_R\|_{\ell_\ctGamma^2} < \infty$. Then, for all $\delta>0$, there exists an $R$ independent constant, $R_\delta > 0$, a constant $R_\infty$ with $R_\infty \to \infty$ as $R\to\infty$ and operators 
$P^R_\delta, P_{\mathrm{loc}}^R$, $P^R_\infty$ 
such that 
\begin{align}\label{eq:decomp_Ham_R}\begin{split}
    \widetilde{\Ham}^R(u_R) &= \widetilde{\Ham}^{\mathrm{ref},R} + P_\delta^R + P_\mathrm{loc}^R + P^R_\infty
\end{split}\end{align}
where $\|P^R_\delta\|_\mathrm{F}\leq \delta$ and $P^R_\mathrm{loc}, P^R_\infty$ are finite rank operators with rank independent of $R$, and matrix entries non-zero only on $B_{R_\delta}\times B_{R_\delta}$ and $(\Lambda_R\setminus B_{R_\infty}) \times (\Lambda_R\setminus B_{R_\infty})$, respectively. 

Moreover, if $u_R \rightharpoonup u$, then $P_\mathrm{loc}^R \to P^\infty_\mathrm{loc}$ where $[P^\infty_\mathrm{loc}]_{\ell k}^{ab} \coloneqq [\widetilde{\Ham}(u) - \widetilde{\Ham}^\mathrm{ref}]_{\ell k}^{ab}$ as $R\to\infty$ for all $\ell, k \in B_{R_\delta}$.
%
%
\end{lemma}
\begin{proof}
The construction is similar to that of Lemma~\ref{lem:decomp_Ham_infty}. We give a sketch of the argument in \Cref{app:decomp_Ham} for completeness.
\end{proof}

As shown in \cite[Lemma~3]{ChenOrtnerThomas2019}, Lemma~\ref{lem:decomp_Ham_infty} is sufficient to prove Lemma~\ref{lem:spectrum}. We now state and prove an analogous $R<\infty$ result:
%
%
\begin{lemma}[Decomposition of the Spectrum]\label{lem:perturbation_spec}
    For $\delta > 0$ and $u_R$ satisfying \asNonInterR, with some constant uniformly bounded below by $\mathfrak{m}>0$, and $\sup_R \|Du_R\|_{\ell^2_\ctGamma} < \infty$, there exists $S_\delta$ such that
	\begin{align}
	    \#\left( \sigma(\Ham^R(u_R)) \setminus B_{\delta}(\sigma(\Ham^{\mathrm{ref}})) \right) \leq S_{\delta}.\label{eq:decomposition_spectrum_2}
	\end{align}
	The constant $S_\delta$ depends on $\delta$ and $(u_R)$ but not on $R$.

\end{lemma}
\begin{proof}
	We first note that $\sigma(\Ham^{\mathrm{ref},R}) \subset \sigma(\Ham^{\mathrm{ref}})$. This can be shown by writing $\sigma(\Ham^\mathrm{ref})$ as the union of energy bands defined on the Brillouin zone and noting that $\sigma(\Ham^{\mathrm{ref},R})$ can then be written as a union of these energy bands over a discretised Brillouin zone (see Appendix~\ref{sec:band_structure} for the details).
	
	As a result of Lemma~\ref{lem:decomp_Ham_R}, we can approximate $\widetilde{\Ham}^R(u_R)$ by a finite rank update of $\widetilde{\Ham}^\mathrm{ref}$. That is, there exists an operator $Q^R(u_R)$ of finite rank (independent of $R$) such that
	\begin{align*}
	\mathrm{dist}\left(\sigma(\Ham^R(u_R)), \sigma(\Ham^{\mathrm{ref},R} + Q^R(u_R))\right) \leq \delta. 
	\end{align*}
	Since $Q^R(u_R)$ is of finite rank, we may apply an interlacing theorem for the eigenvalues of rank one updates \cite{DancisDavis1987} finitely many times (independently of $R$) to conclude that,
	\begin{align*}
	    \#\left(
	    \sigma(\Ham^{\mathrm{ref},R} + Q^R(u_R))
	    \setminus \sigma(\Ham^{\mathrm{ref}})
	    \right) 
	    \leq S_\delta
	\end{align*}
	for some $S_\delta$ independent of $R$. 
\end{proof}

\subsubsection{Resolvent Calculus}

In order to write the finite Fermi-temperature site energies using resolvent calculus, we need $\mathfrak{g}^\beta(\cdot\,;\mu)$ to extend to holomorphic functions on some open neighbourhood of the spectrum \cite[Lemma~6]{ChenOrtnerThomas2019}:
\begin{lemma}[Analytic Continuation]
	\label{lem:analytic-cont}
	Fix $\beta \in(0,\infty)$. Then, $z\mapsto\mathfrak{g}^\beta(z;\mu)$ can be analytically continued to the set $\mathbb C \setminus \left\{\mu + ir \colon r \in \mathbb R, |r| \geq \pi\beta^{-1}\right\}$.
\end{lemma}
From now on, we will denote the analytic continuation again by $z\mapsto\mathfrak{g}^\beta(z;\mu)$.

For fixed $u \in \mathrm{Adm}(\Lambda)$, we suppose that $\mathscr C^-$ and $\mathscr C^+$ are simple closed contours encircling $\sigma(\Ham(u)) \cap (-\infty,\mu)$ and $\sigma(\Ham(u)) \cap (\mu,\infty)$, respectively, and avoiding the line $\mu + i\mathbb R$, see Figure~\ref{fig:contours}. Further, we may suppose that 
\begin{gather}
\label{eq:distance_contour_0T}
	\mathrm{dist}(z, \sigma(\Ham(u))) \geq \frac{1}{2}\ctDist(u) \quad \text{and} \quad
	|\Re{z}-\mu|\geq \frac{1}{2}\ctDist(u) \quad \forall\, z \in \mathscr C^- \cup \mathscr C^+\\
\label{def:du}
    \textrm{where} \quad 
    \ctDist(u) \coloneqq \mathrm{dist}(\mu, \sigma(\Ham(u))) 
    \quad \textrm{and} \quad
    \ctDist^\mathrm{ref} \coloneqq \mathrm{dist}(\mu, \sigma(\Ham^\mathrm{ref})).
\end{gather}
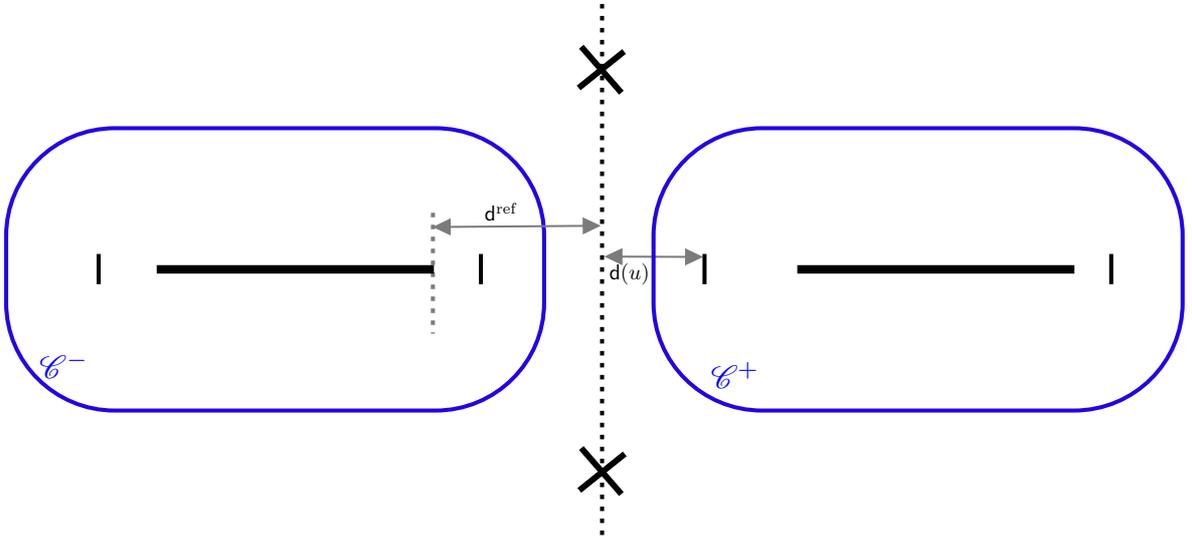
\begin{figure}[ht]
	\centering
	\resizebox{\columnwidth} {!} {
		\tikzset{every picture/.style={line width=0.75pt}} 

\begin{tikzpicture}[x=0.75pt,y=0.75pt,yscale=-1,xscale=1]

\draw [color={rgb, 255:red, 0; green, 0; blue, 0 }  ,draw opacity=1 ][line width=1.5]  [dash pattern={on 1.69pt off 2.76pt}]  (334.5,12) -- (334.5,279) ;

\draw  [color={rgb, 255:red, 34; green, 4; blue, 230 }  ,draw opacity=1 ][line width=1.5]  (41.45,126.96) .. controls (41.45,97.49) and (65.34,73.6) .. (94.81,73.6) -- (252.64,73.6) .. controls (282.11,73.6) and (306,97.49) .. (306,126.96) -- (306,160.24) .. controls (306,189.71) and (282.11,213.6) .. (252.64,213.6) -- (94.81,213.6) .. controls (65.34,213.6) and (41.45,189.71) .. (41.45,160.24) -- cycle ;
\draw [line width=1.5]    (86.9,135.96) -- (86.9,150.96) ;

\draw [line width=2.25]    (324.3,232.2) -- (344,255) ;

\draw [line width=2.25]    (323.3,253.2) -- (345.65,235.1) ;

\draw [line width=2.25]    (324.3,33.2) -- (344,56) ;

\draw [line width=2.25]    (322.97,53.65) -- (345.32,35.55) ;

\draw [color={rgb, 255:red, 123; green, 123; blue, 123 }  ,draw opacity=1 ][line width=1.5]  [dash pattern={on 1.69pt off 2.76pt}]  (251.34,115.52) -- (251.3,175.5) ;

\draw  [fill={rgb, 255:red, 0; green, 0; blue, 0 }  ,fill opacity=1 ] (116,142) -- (251.3,142) -- (251.3,145.2) -- (116,145.2) -- cycle ;
\draw  [fill={rgb, 255:red, 0; green, 0; blue, 0 }  ,fill opacity=1 ] (431,142) -- (566.3,142) -- (566.3,145.2) -- (431,145.2) -- cycle ;
\draw [line width=1.5]    (274.9,135.96) -- (274.9,150.96) ;

\draw [line width=1.5]    (584.9,135.96) -- (584.9,150.96) ;

\draw [color={rgb, 255:red, 123; green, 123; blue, 123 }  ,draw opacity=1 ]   (331.9,121.98) -- (253.15,122.58) ;
\draw [shift={(251.15,122.6)}, rotate = 359.56] [fill={rgb, 255:red, 123; green, 123; blue, 123 }  ,fill opacity=1 ][line width=0.75]  [draw opacity=0] (8.93,-4.29) -- (0,0) -- (8.93,4.29) -- cycle    ;
\draw [shift={(333.9,121.96)}, rotate = 179.56] [fill={rgb, 255:red, 123; green, 123; blue, 123 }  ,fill opacity=1 ][line width=0.75]  [draw opacity=0] (8.93,-4.29) -- (0,0) -- (8.93,4.29) -- cycle    ;
\draw [color={rgb, 255:red, 123; green, 123; blue, 123 }  ,draw opacity=1 ]   (337.6,137.26) -- (382.25,137.43) ;
\draw [shift={(384.25,137.44)}, rotate = 180.22] [fill={rgb, 255:red, 123; green, 123; blue, 123 }  ,fill opacity=1 ][line width=0.75]  [draw opacity=0] (8.93,-4.29) -- (0,0) -- (8.93,4.29) -- cycle    ;
\draw [shift={(335.6,137.25)}, rotate = 0.22] [fill={rgb, 255:red, 123; green, 123; blue, 123 }  ,fill opacity=1 ][line width=0.75]  [draw opacity=0] (8.93,-4.29) -- (0,0) -- (8.93,4.29) -- cycle    ;
\draw [line width=1.5]    (384.9,135.96) -- (384.9,150.96) ;

\draw  [color={rgb, 255:red, 34; green, 4; blue, 230 }  ,draw opacity=1 ][line width=1.5]  (359.8,126.96) .. controls (359.8,97.49) and (383.69,73.6) .. (413.16,73.6) -- (566.64,73.6) .. controls (596.11,73.6) and (620,97.49) .. (620,126.96) -- (620,160.24) .. controls (620,189.71) and (596.11,213.6) .. (566.64,213.6) -- (413.16,213.6) .. controls (383.69,213.6) and (359.8,189.71) .. (359.8,160.24) -- cycle ;

\draw (70,191) node [scale=1.2,color={rgb, 255:red, 0; green, 0; blue, 255 }  ,opacity=1 ]  {$\mathscr{C}^{-}$};
\draw (285,115) node [scale=0.8,color={rgb, 255:red, 0; green, 0; blue, 0 }  ,opacity=1 ]  {$\mathsf{d}^{\mathrm{ref}}$};
\draw (348,146) node [scale=0.8,color={rgb, 255:red, 0; green, 0; blue, 0 }  ,opacity=1 ]  {$\mathsf{d}(u)$};
\draw (400,196) node [scale=1.2,color={rgb, 255:red, 0; green, 0; blue, 255 }  ,opacity=1 ]  {$\mathscr{C}^{+}$};

\end{tikzpicture} }%
	\caption{
	Cartoon depicting an approximation of $\sigma(\Ham(u))$ for $u \in \mathrm{Adm}(\Lambda)$ (on the real axis) and the contours $\mathscr C^-$ and $\mathscr C^+$. The positive constants $\ctDist(u)$ and $\ctDist^\mathrm{ref}$ are also displayed.}
	\label{fig:contours}
\end{figure}

With these contours defined, we may write the site energies for both finite and zero Fermi-temperature defined in \cref{eq:site_decomp} using resolvent calculus: for $\ell \in \Lambda$ and $\beta \in (0,\infty]$
\begin{equation}
	\label{eq:site_energy_resolvent}
	\mathcal G_\ell^{\beta}(Du(\ell)) = - \frac{1}{2\pi i} \sum_{a=1}^{\numorbitals}
	\oint_{\mathscr{C}^- \cup \mathscr C^+}
	\mathfrak{g}^\beta(z;\mu)\left[\mathscr (\Ham(u) - z)^{-1}\right]_{\ell \ell}^{a a} \mathrm{d}z 
\end{equation}
We shall often simplify notation and write $\mathcal G_\ell \coloneqq \mathcal G_\ell^\infty$. 

Similarly, for the site energies defined for an admissible periodic displacement $u \colon \Lambda_R \to \mathbb R^d$ for $R>0$, we can fix contours $\mathscr C^-, \mathscr C^+$ satisfying \cref{eq:distance_contour_0T} and \cref{def:du} with $\Ham$ replaced with $\Ham^R$, and write
\begin{equation}
	\label{eq:site_energy_resolvent_R}
	G_\ell^{\beta,R}(u) = - \frac{1}{2\pi i} \sum_{a=1}^{\numorbitals} \oint_{\mathscr{C}^-\cup \mathscr{C}^+}
	\mathfrak{g}^\beta(z;\mu)\left[\mathscr (\Ham^R(u) - z)^{-1}\right]_{\ell \ell}^{a a} \mathrm{d}z.
\end{equation}

From now on, we shall denote the resolvent operators in \cref{eq:site_energy_resolvent} and \cref{eq:site_energy_resolvent_R} by $\mathscr R_z(u) \coloneqq (\Ham(u) - z)^{-1}$ and $\mathscr R^R_z(u) \coloneqq (\Ham^R(u) - z)^{-1}$, respectively. This resolvent calculus approach has been widely used for the tight binding model \cite{ChenLuOrtner18,ChenOrtner16,ELu10,Goedecker1995}.

The following Combes-Thomas type resolvent estimates will be useful in the main proofs:
\begin{lemma}[Combes-Thomas Resolvent Estimates]
	\label{lem:CT}
	Fix $u \in \mathrm{Adm}(\Lambda)$, $\mathfrak{d}^\mathrm{ref}>0$ and $z \in \mathbb C$ such that $\mathrm{dist}(z, \sigma(\Ham^\mathrm{ref})) \geq \mathfrak{d}^\mathrm{ref}$. Then, for all $0\leq j\leq\ctHamregularity$, $\ell,k \in \Lambda$ and $\bm{m} = (m_1,\dots,m_j) \in \Lambda^j$, there exist positive constants $C_0 = C_0(\ell,k)$ and $C_j = C_j(\ell,\bm{m})$ such that
	\begin{gather}
		 \big| \big[\mathscr R_z(u)\big]^{ab}_{\ell k} \big|
		 \leq C_0(\ell,k) e^{-\gamma_\mathrm{CT} |\ell - k|} 
		 \quad\text{and} \label{eq:resolvent}\\
		 \left| \frac{\partial^j \big[\mathscr R_z(u)\big]^{aa}_{\ell\ell}}{\partial [u(m_1)]_{i_1}\dots \partial [u(m_j)]_{i_j}} \right|
		 \leq C_j(\ell,\bm{m}) e^{-\gamma_\mathrm{CT} \sum_{l=1}^j |\ell - m_l|}. \label{eq:der_resolvent}
	\end{gather}
    where $\gamma_{\mathrm{CT}} = c \min\{ 1, \mathfrak{d}^\mathrm{ref} \}$ and $c$ depends on $j,\ctTBprefactor,\ctTBexponent,\mathfrak{m},d$. 
    
    If $\Lambda = \Lambda^\mathrm{ref}$ and $u = 0$, then the constants $C_j^\mathrm{ref} \coloneqq C_j$ are independent of $\ell,\bm{m}$.  
	
	Moreover, $C_0(\ell,k) \to C_0^\mathrm{ref}$ as $|\ell| + |k| - |\ell - k| \to \infty$ and  $C_j(\ell,\bm{m}) \to C_j^\mathrm{ref}$ as the subsystem containing $\ell,m_1,\dots,m_j$ moves away from the defect core together.
\end{lemma}
\begin{proof}
	The first bound \cref{eq:resolvent} for $\Lambda = \Lambda^\mathrm{ref}$ and $u = 0$ is a standard resolvent estimate \cite[Lemma~6]{ChenOrtner16} where the exponent is explicitly calculated in terms of $\mathfrak{d}^\mathrm{ref}$ as in \cite[Lemma~8]{ChenOrtnerThomas2019}. The estimates for the derivatives of the resolvent \cref{eq:der_resolvent} for $\Lambda = \Lambda^\mathrm{ref}$ and $u=0$ follow from \cref{eq:resolvent} as in equations (4.6) and (4.8) of \cite[\S~4.2]{ChenOrtnerThomas2019}. The same estimates can be derived after replacing $\mathscr R_z^\mathrm{ref}$ with $\mathscr R_z(u)$ and $\mathfrak{d}^\mathrm{ref}$ with $\mathfrak{d}$ where $\mathrm{dist}(z,\sigma(\Ham(u))) \geq \mathfrak{d}$.
	
	The improved locality estimates (that is, with the exponent only depending on the reference band gap and not on the discrete spectrum) of \cref{eq:resolvent} and \cref{eq:der_resolvent} are derived in equations (4.19)$-$(4.24) of \cite[\S 4.4]{ChenOrtnerThomas2019}.
\end{proof}

\begin{remark}\label{rem:CT_R}
The same result holds in the case $R<\infty$ if $r_{\ell k}(u)$ is replaced with the torus distance $r_{\ell k}^\#(u)$, and the proof follows in the exact same way as in the corresponding $R=\infty$ result. The improved locality results can be proved by following equations (4.19)$-$(4.24) of \cite[\S 4.4]{ChenOrtnerThomas2019} and using the decomposition of the Hamiltonian \cref{eq:decomp_Ham_R}.
\end{remark}

\subsubsection{Spectral Pollution}

It is well known (see \cite{DaviesPlum2004, LewinSere2009} and references therein) that, in general, when approximating the spectrum of an operator with a sequence of finite dimensional spaces, spurious eigenvalues may be present in the limit. That is, accumulation points of eigenvalues along the sequence are not necessarily contained in the spectrum of the limit operator. In this section, we discuss the extent to which spectral pollution occurs when approximating $\sigma(\Ham(u))$ with $\sigma(\Ham^R(u_R))$.

More specifically, we are able to show that, if $u_R \to u$ strongly, then spectral pollution does not occur and, in the case that $u_R \rightharpoonup u$, we show that the spectral pollution is very mild. That is, we may use Lemma~\ref{lem:decomp_Ham_R} to conclude that there are at most finitely many additional eigenstates in the band gap which arise due to finitely many $O(1)$ distortions of the lattice. These distortions are sent to infinity as $R\to\infty$ and so the additional eigenstates are not present in the limit.

We remark here that the use of periodic boundary conditions prevents spectral pollution that is known to occur in the case of clamped boundary conditions, for example, see \cite{CancesEhrlacherMaday2012} for a proof in the case of local defects in a crystalline material in a PDE setting.

%
%
%
\begin{lemma}
\label{lem:spectral_pollution}
Suppose $u_R\colon \Lambda_R \to \mathbb R^d$ is a bounded sequence satisfying \asNonInterR~with some uniform constant $\mathfrak{m}>0$ and $u_R \rightharpoonup u$ for some $u\in \dot{\mathscr{W}}^{1,2}(\Lambda)$. Let $P^{R}_\infty(u_R)$ be the finite rank operator from Lemma~\ref{lem:decomp_Ham_R} with the constant $\delta > 0$. Then, 
\begin{itemize}
\setlength\itemsep{0em}

\item[(i)]
$\sigma(\Ham(u)) \subset \liminf\limits_{R\to\infty} \sigma\big(\Ham^R(u_R) - P^{R}_\infty(u_R) \big),
$

\item[(ii)] 
$
\sigma(\Ham(u)) \subset  \liminf\limits_{R\to\infty} \sigma\big(\Ham^R(u_R)\big),
$
%

\item[(iii)]
$
    \sigma(\Ham(u)) \supset 
    \limsup\limits_{R\to\infty} \big[
        \sigma\big(\Ham^R(u_R) - P^{R}_\infty(u_R)\big) \setminus B_{2\delta}(\sigma(\Ham^\mathrm{ref}))
    \big],
$
\item[(iv)] If $u_R \to u$ strongly, then 
$ \sigma(\Ham(u)) = \lim\limits_{R\to\infty} \sigma(\Ham^R(u_R))$.
\end{itemize}
\end{lemma}
\begin{remark}\label{rem:spectral_pollution}
Following the proof of Lemma~\ref{lem:spectral_pollution}, one can easily see that if $u_\beta, u \in \dot{\mathscr{W}}^{1,2}(\Lambda)$, satisfying \asNonInter~with a uniform constant $\mathfrak{m}$, and $u_\beta \rightharpoonup u$ as $\beta \to \infty$, then $\sigma(\Ham(u)) \subset \liminf\limits_{\beta \to \infty} \sigma(\Ham(u_\beta))$.
\end{remark}
\begin{figure}[ht]
	\centering
	\resizebox{\columnwidth} {!} {
		\input{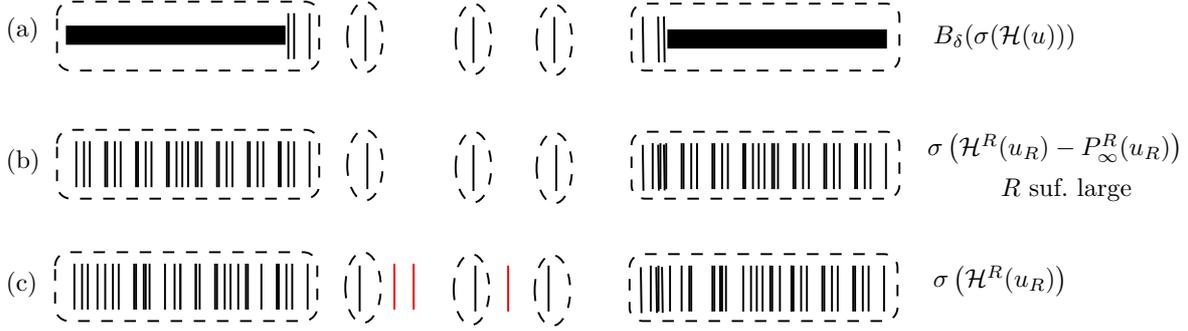} }%
	\caption{Cartoon illustrating Lemma~\ref{lem:spectral_pollution}. (a) is qualitatively similar to $\sigma(\Ham(u))$ for $u \in \dot{\mathscr W}^{1,2}(\Lambda)$ as asserted in Lemma~\ref{lem:spectrum}. (b) illustrates Lemma~\ref{lem:spectral_pollution}~\textit{(i)}~and~\textit{(iii)}: eigenvalues of $\Ham(u)$ lying in the band gap can be approximated by eigenvalues of $\Ham^R(u_R) - P^{R}_\infty(u_R)$ and every accumulation point of $\sigma\big(\Ham^R(u_R) - P^{R}_\infty(u_R)\big)$ is contained in $\sigma(\Ham(u))$. (c) illustrates Lemma~\ref{lem:spectral_pollution}~\textit{(ii)} where the finitely many eigenvalues outside $B_\delta(\sigma(\Ham(u)))$ are the ``defect states'' that arise when including the far-field contribution $P_\infty^{R}(u_R)$. These defect states vanish in the weak limit.}\label{fig:spectral_pollution}
\end{figure}
\begin{proof}
The first part of this proof loosely follows the first part of \cite[Proof of Thm.~3.1]{CancesEhrlacherMaday2012}. 

\textit{(i).} Take $\lambda \in \sigma(\Ham(u))$. For every $\tau > 0$, we may choose $\psi$ of compact support such that $\|\psi\|_{\ell^2} = 1$, $\mathrm{supp}(\psi) \subset B_{R_0}$ for some $R_0>0$, and 
\begin{align*}
    \| (\Ham(u) - \lambda) \psi \|_{\ell^2} \leq \tau.
\end{align*}
For $R\geq R_0$, we let $\psi_R \coloneqq \psi|_{\Lambda_R}$ and calculate: for $\ell \in \Lambda_R$ and $1\leq a\leq \numorbitals$,
\begin{align}\label{eq:sp_part_1}
    &\big[\big(\Ham^R(u_R) - P^{R}_\infty(u_R) \big) \psi_R\big]_{\ell}^a 
    = [\Ham(u)\psi]_\ell^a 
    -   \sum_{1 \leq b \leq \numorbitals}
    \sum_{k \in \Lambda_{R} \cap B_{R_0} }
    P^{R}_\infty(u_R)_{\ell k}^{ab}\psi_R(k;b)  \\ 
    &\,+\sum_{\above{k \in \Lambda_{R} \cap B_{R_0}}{1 \leq b \leq \numorbitals}}
    \left(h_{\ell k}^{ab}(\bm{r}_{\ell k}(u_R)) - h_{\ell k}^{ab}(\bm{r}_{\ell k}(u))\right) \psi(k;b) 
     + \sum_{\above{k \in \Lambda_{R} \cap B_{R_0}}{1 \leq b \leq \numorbitals}}
    \sum_{\above{\alpha \in \mathbb Z^d}{\alpha \not=0}} h_{\ell k}^{ab}(\bm{r}_{\ell k}(u_R) + \mathsf{M}_R\alpha) \psi_R(k;b). \nonumber
\end{align}
Therefore, after choosing $R$ sufficiently large such that $P_\infty^{R}(u_R)_{\ell k}^{ab} = 0$ for all $k \in \Lambda_R \cap B_{R_0}$, squaring, summing over $\ell \in \Lambda_R$ and applying Lemma~\ref{ham_converge}, we have: for sufficiently large $R$, 
\begin{align}\label{eq:HamRminuslambda}\begin{split}
    &\big\|\big(\Ham^R(u_R) - P^{R}_\infty(u_R) - \lambda\big)\psi_R\big\|_{\ell^2(\Lambda_R)} \\
    &\qquad\leq 
    \big\|\big(\Ham^R(u_R) - P^{R}_\infty(u_R)\big) \psi_R - \Ham(u) \psi \big\|_{\ell^2(\Lambda_R)} 
    + \|(\Ham(u) - \lambda)\psi\|_{\ell^2(\Lambda_R)} \\
    &\qquad\leq C\|D(u_R - u)\|_{\ell^2_\ctGamma(\Lambda_R \cap B_{2R_0})} + C e^{- \gamma_0\mathfrak{m} R_0 } + C e^{-\frac{1}{2}\ctTBexponent \mathfrak{m} (R - R_0)} + \tau.
\end{split}\end{align}
Here, we have used the fact that for $\ell\in \Lambda_R$ and $k \in B_{R_0}$, we have $|\ell - k + \mathsf{M}_R\alpha| \geq R - R_0$ for all $\alpha \in \mathbb Z^d \setminus \{0\}$.

Therefore, by choosing $R_0$ and then $R$ sufficiently large, either $\lambda \in \sigma\big(\Ham^R(u_R)- P^{R}_\infty(u_R)\big)$ or we can write 
\begin{align*}
    1 = \|\psi_R\|_{\ell^2(\Lambda_R)} 
    &\leq 
    \|(\Ham^R(u_R) - P^{R}_\infty(u_R) - \lambda)^{-1}\|_{\ell^2\to\ell^2} \|(\Ham^R(u_R) - P^{R}_\infty(u_R) - \lambda)\psi_R\|_{\ell^2(\Lambda_R)} \\
    &\leq \|(\Ham^R(u_R) - P^{R}_\infty(u_R) - \lambda)^{-1}\|_{\ell^2\to\ell^2} \cdot 2\tau.
\end{align*}
That is, in the case that $\lambda \not\in \sigma\big(\Ham^R(u_R)- P^{R}_\infty(u_R)\big)$, we know that $(\Ham^R(u_R) - P^{R}_\infty(u_R) - \lambda)^{-1}$ defines a bounded linear operator and so 
\begin{align*}
    \mathrm{dist}\left(\lambda, \sigma\big(\Ham^R(u_R) - P^{R}_\infty(u_R)\big)\right) %
    = \frac{1}{\|(\Ham^R(u_R) - P^{R}_\infty(u_R) - \lambda)^{-1}\|_{\ell^2\to\ell^2}} \leq 2\tau.  
\end{align*}
Here, we have used the fact that, for a bounded normal operator, the operator norm equals the spectral radius.

\textit{(ii).} The exact same arguments may be made for the operator $\Ham^R(u_R)$. In this case, the second term in \cref{eq:sp_part_1} is omitted and we obtain $\|(\Ham^R(u_R) - \lambda)\psi_R\|_{\ell^2} \leq 2\tau$ for all $R$ sufficiently large as in the proof of \textit{(i)}.

\textit{(iii).} We suppose that $\lambda \in \limsup_{R\to\infty} \sigma\big(\Ham^R(u_R) - P^{R}_\infty(u_R)\big)$
with $\lambda \not\in B_{2\delta}(\sigma(\Ham^\mathrm{ref}))$. By Lemma~\ref{lem:decomp_Ham_R} and~\ref{lem:perturbation_spec}, there exists $S_\delta > 0$ such that 
\begin{gather}
    \sigma\big(\Ham^{\mathrm{ref},R} + P_\delta^R(u_R)\big) \cap B_\delta(\lambda) = \emptyset \label{eq:sp_pt3_1} \\
    \# \left( \sigma\big(\Ham^R(u_R) - P^{R}_\infty(u_R)\big) \cap B_\delta(\lambda) \right) \leq S_\delta \quad \forall \, R \label{eq:sp_pt3_2}
\end{gather}
where $P^R_\delta(u_R)$ is the perturbation arising in Lemma~\ref{lem:decomp_Ham_R} with $\|P^R_\delta(u_R)\|_\mathrm{F} \leq \delta$.

By \cref{eq:sp_pt3_1} and \cref{eq:sp_pt3_2}, we may let $\mathscr C = \partial B_\delta(\lambda)$ be the positively oriented circle of radius $\delta$ centred at $\lambda$, and obtain
\begin{align}
\begin{split}\label{eq:spec_projector}
    &\sum_j \lambda^{(j)}_R \psi_R^{(j)} \otimes \psi_R^{(j)} = - \frac{1}{2\pi i} \oint_\mathscr{C}(\Ham^R(u_R)- P^{R}_\infty(u_R) - z)^{-1} \mathrm{d}z \\
    &\qquad=  - \frac{1}{2\pi i} \oint_\mathscr{C}(\Ham^R(u_R) - P^{R}_\infty(u_R) - z)^{-1} - (\Ham^{\mathrm{ref},R} + P_\delta^R(u_R) - z)^{-1}\mathrm{d}z
\end{split}
\end{align}
where $\mathrm{span} \{\psi_R^{(j)}\}_j$ is the eigenspace corresponding to the eigenvalues $\lambda_R^{(j)} \in \sigma\big(\Ham^R(u_R) - P^{R}_\infty(u_R)\big)$ with $\lambda_R^{(j)} \in B_\delta(\lambda)$ and $\|\psi_R^{(j)}\|_{\ell^2} = 1$. By Lemma~\ref{lem:decomp_Ham_R}, for sufficiently large $R$, we have
\begin{align}
    &\left|\left[(\Ham^R(u_R) - P^{R}_\infty(u_R)- z)^{-1} - (\Ham^{\mathrm{ref},R} + P_\delta^R(u_R) - z)^{-1}\right]_{\ell k}^{ab}\right| \nonumber \\ 
    &\qquad 
    = \left|\left[
    (\Ham^R(u_R) -P^{R}_\infty(u_R)- z)^{-1}
   P^{R}_\mathrm{loc}(u_R)
    (\Ham^{\mathrm{ref},R} + P_\delta^R(u_R) - z)^{-1}
    \right]_{\ell k}^{ab}\right| 
        \nonumber\\
    &\qquad
    \leq C \sum_{\ell_1,\ell_2 \in \Lambda_R \cap B_{R_\delta}}
        e^{-\ctCT(r^\#_{\ell \ell_1} + r^\#_{\ell_2 k})} 
    %
    %
    \leq Ce^{-\ctCT (|\ell| + |k|)}.
        \label{eq:RminusRdef}
\end{align}
Equation~\cref{eq:RminusRdef} is analogous to the $R = \infty$ result shown in \cite[Eq.~(4.19)]{ChenOrtnerThomas2019}. Therefore, by applying \cref{eq:spec_projector}, we have
\begin{align}
\label{eq:eigenvector_decay}
    |\psi_R^{(j)}(\ell; a)| \leq C e^{-\ctCT|\ell|} \quad \text{for all }\ell \in \Lambda_R \text{ and } 1\leq a \leq \numorbitals.
\end{align}

Now, after defining $\widetilde{\psi}_R^{(j)}$ to be equal to $\psi_R^{(j)}$ on $\Lambda_R$ and extending by zero to $\Lambda$, we have: for sufficiently large $R$,
\begin{align*}
    \|(\Ham(u) - \lambda^{(j)}_R)\widetilde{\psi}_R^{(j)}\|_{\ell^2} \leq C\|D(u_R - u)\|_{\ell^2_\ctGamma(\Lambda \cap B_{R_0})} + C \big( e^{-\eta_1 R_0} + e^{-\eta_2 (R - R_0)} + e^{-\eta_3 R_\infty} \big)   
\end{align*}
where $\eta_j > 0$ for each $j \in \{1,2,3\}$ and $R_\infty$ is the constant from Lemma~\ref{lem:decomp_Ham_R} (that is, $P^{R}_\infty(u_R)$ zero on $(\Lambda_R \setminus B_{R_\infty})^2$ with $R_\infty \to \infty$ as $R\to\infty$). This calculation is analogous to \cref{eq:HamRminuslambda} where, instead of exploiting the fact the (approximate) eigenvectors are of compact support, we now use the exponential decay of the eigenvectors \cref{eq:eigenvector_decay}.

For a strictly increasing sequence $(R_n)_n \subset \mathbb N$ and sequence of indices $(j_n)_n$, we define the subsequence $(\lambda_n, \psi_n) \coloneqq (\lambda^{(j_n)}_{R_n}, \widetilde{\psi}^{(j_n)}_{R_n})$. We can conclude that if $\lambda_n \to \lambda$ as $n \to \infty$, we have 
\begin{align*}
    \big\|
    (\Ham(u) - \lambda){\psi}_n
    \big\|_{\ell^2} &\leq 
    \big\|
    (\Ham(u) - \lambda_n){\psi}_n
    \big\|_{\ell^2} + 
    |\lambda - \lambda_n| 
    \to 0 \quad \text{as }n \to \infty.
\end{align*}
Therefore, by applying Weyl's criterion \cite[Ch.~7]{HislopSigal1995}, we can conclude that $\lambda \in \sigma(\Ham(u))$.

\textit{(iv).} In the case that $u_R \to u$ strongly as $R\to\infty$, we have: for all $\delta>0$, there exists $R_\delta > 0$ such that $\|Du_R\|_{\ell^2_\ctGamma(\Lambda_R\setminus B_{R_\delta})} \leq \delta$ for all $R$ sufficiently large. Following the proof of Lemma~\ref{lem:decomp_Ham_R}, we can conclude that $P^R_\infty(u_R) = 0$.
\end{proof}

\subsubsection{Limits of the Site Energies}

We now state that $\mathfrak{g}^\beta(\,\cdot\,;\mu)$ converges exponentially as $\beta \to \infty$ which is used in the convergence of the site energies in the zero temperature limit.
\begin{lemma}\label{lem:g_convergence_2}
	Fix $z \in \mathbb C$ such that 
	$\ctDist \coloneqq \frac{1}{2}|\mathrm{Re}(z) - \mu|>0$. 
	Then, for all $\beta_0 > 0$, there exists a positive constant $C_{\beta_0\ctDist}$ such that
	\begin{align*}
		|\mathfrak{g}^\beta(z;\mu) - \mathfrak{g}(z;\mu)| \leq C_{\beta_0\ctDist} \beta^{-1} e^{-\frac{1}{3}\beta|\Re{z} - \mu|} \quad \forall\, \beta > \beta_0.
	\end{align*}
\end{lemma}
\begin{proof}
    The proof is lengthy but elementary and so is given in \Cref{proof_gbeta_convergence}.
\end{proof}
 
We now apply Lemma~\ref{lem:g_convergence_2} together with the resolvent estimates of Lemma~\ref{lem:CT} to show that the site energies and their derivatives converge in the zero temperature limit:
\begin{lemma}[Zero Temperature Limit of the Site Energies]
	\label{lem:convergence_site_energy}
	Let $u \in \emph{Adm}(\Lambda)$. Then, for each $\beta_0 > 0$, $0 \leq j \leq \ctHamregularity$, $\ell \in \Lambda$, $\bm{m}= (m_1,\dots,m_j) \in \Lambda^j$ and $1\leq i_1,\dots,i_j \leq d$, there exists a constant $C$ depending on $\beta_0, \numorbitals, \ctDist(u), d$ such that, for all $\beta > \beta_0$, we have
	\begin{gather*}
	\left|\frac{\partial^j \mathcal G_{\ell}^\beta(Du(\ell))}{\partial [u(m_1)]_{i_1} \dots \partial [u(m_j)]_{i_j}} -  
	\frac{\partial^j \mathcal G_{\ell}(Du(\ell))}{\partial [u(m_1)]_{i_1} \dots \partial [u(m_j)]_{i_j}} \right| 
	\leq C \beta^{-1} e^{- \frac{1}{6}\ctDist(u) \beta} e^{-\gamma_\mathrm{CT} \sum_{l=1}^j r_{\ell m_l} }.
	\end{gather*}	
\end{lemma}

\begin{proof}
With the resolvent estimates of Lemma~\ref{lem:CT} and the convergence of the integrand $\mathfrak{g}^\beta$ shown in Lemma~\ref{lem:g_convergence_2}, this proof is a simple corollary.

Using \cref{eq:site_energy_resolvent}, we may note that
\begin{align*}
	\left|\mathcal G^\beta_{\ell}(Du(\ell)) - \mathcal G_{\ell}(Du(\ell))\right|  
	&= \frac{1}{2\pi}\left|\sum_a \oint_{\mathscr C^-\cup\mathscr C^+} 
		\left[ \mathfrak{g}^\beta(z;\mu) - \mathfrak{g}(z;\mu)\right] 
		\left[\mathscr R_z(u)\right]_{\ell\ell}^{aa} \mathrm{d}{z}\right| \\
	&\lesssim  C_0 \max_{z \in \mathscr C^- \cup \mathscr C^+} 
		\big|\mathfrak{g}^\beta(z;\mu) - \mathfrak{g}(z;\mu)\big| \\
	&\lesssim C_0C_{\beta_0\ctDist} \beta^{-1} e^{-\frac{1}{6}\ctDist(u)\beta }  
\end{align*}
where $C_0 = C_0(\ell,\ell)$ is the constant from Lemma~\ref{lem:CT} and $C_{\beta_0\ctDist}$ is the constant from Lemma~\ref{lem:g_convergence_2}. 
	
Moreover, 
\begin{align*}
	&\bigg|\frac{\partial^j \mathcal G^\beta_{\ell}(Du(\ell))}{\partial[u(m_1)]_{i_1} 
		\dots \partial[u(m_j)]_{i_j}} - \frac{\partial^j \mathcal G_{\ell}(Du(\ell))}{\partial[u(m_1)]_{i_1} \dots \partial[u(m_j)]_{i_j}}\bigg| \\
	&\qquad= \frac{1}{2\pi}\left|\sum_a \oint_{\mathscr C^-\cup\mathscr C^+} 
		\left[ \mathfrak{g}^\beta(z;\mu) - \mathfrak{g}(z;\mu)\right] 
		\frac{\partial^j \left[\mathscr R_z(u)\right]_{\ell\ell}^{aa}}{\partial[u(m_1)]_{i_1} 
			\dots \partial[u(m_j)]_{i_j}} \mathrm{d}{z}\right| \\
	&\qquad\lesssim 
	C_{\beta_0\ctDist} \beta^{-1} e^{-\frac{1}{6}\ctDist(u) \beta} \cdot C_j e^{-\ctCT \sum_{l=1}^j r_{\ell m_l} }
\end{align*}
where $C_j = C_j(\ell,\bm{m})$ is the constant from Lemma~\ref{lem:CT}. 
\end{proof}

The corresponding $R\to\infty$ result is as follows:
\begin{lemma}[Thermodynamic Limit of the Site Energies]\label{eq:thermodynamic_site_energies}
Let $u \in \mathrm{Adm}(\Lambda)$ be of compact support and fix $\beta \in (0,\infty]$. Then, for sufficiently large $R$ and each $0\leq j \leq \nu$, $\ell \in \Lambda_R$, $\bm{m} = (m_1,\dots,m_j) \in \Lambda_R^j$ and $1\leq i_1,\dots,i_j \leq d$, we have 
\begin{align*}
    \left| \frac{\partial^j  G_{\ell}^{\beta,R}(u)}{\partial [u(m_1)]_{i_1} \dots \partial [u(m_j)]_{i_j}} 
	    -  \frac{\partial^j G^\beta_{\ell}(u)}{\partial [u(m_1)]_{i_1} \dots \partial [u(m_j)]_{i_j}} \right|
	\leq C e^{-\eta \big(\mathrm{dist}(\ell, \mathbb R^d \setminus \Omega_R) +  \sum_{l=1}^j r_{\ell m_l}^\#(u)\big)}   
\end{align*}
where $\eta \coloneqq \frac{1}{2}\mathfrak{m}\min\{\ctCT, \frac{1}{2}\ctTBexponent\}$.
\end{lemma}

\begin{proof}
Similar to the calculations in the proof of Lemma~\ref{lem:convergence_site_energy}, we write the site energies using resolvent calculus. Using the fact $\mathfrak{g}^\beta(z;\mu)$ is uniformly bounded along the contour $\mathscr C^-\cup\mathscr C^+$, it is sufficient to prove that the derivatives of the resolvent operators converge in the thermodynamic limit. A full proof is given in Appendix~\ref{app:thermodynamic_limit_site_energies}.
\end{proof}

\subsection{Zero Temperature Limit}

We are now in a position to prove the first main convergence results:

\subsubsection{Proof of Theorem~\ref{thm:0T_limit_infinite_domain}: \texorpdfstring{$\beta \to \infty$}{Zero Temperature Limit} in the Grand Canonical Ensemble}

    We may choose $r>0$ such that $B_r(\overline{u};\|D\cdot\|_{\ell^2_\ctGamma}) \subset \mathrm{Adm}(\Lambda)$. Now, since $\mathcal G^\beta \in C^3(B_r(\overline{u};\|D\cdot\|_{\ell^2_\ctGamma}))$, we know that $\delta^2 \mathcal G^\beta$ is Lipschitz in a neighbourhood of $\overline{u}$.
    
    For the remainder of the proof, we fix $\beta_0 > 0$. By Lemma~\ref{lem:convergence_site_energy}, for all 
	 $v,w \in \W(\Lambda)$,
	 we have,
	\begin{align}
	\begin{split}\label{Thm1Part1:2ndvariationerror}
		&\Braket{\left(\delta^2\mathcal G^\beta(\ubar) - \delta^2\mathcal G(\ubar)\right)v, w} \\
		&\qquad= \sum_{\ell \in \Lambda} \sum_{\rho_1,\rho_2 \in \Lambda - \ell} 
				D_{\rho_1}v(\ell)^T \left(\mathcal G_{\ell,{\rho_1\rho_2}}^\beta(D\ubar(\ell)) - 
				\mathcal G_{\ell,{\rho_1\rho_2}} (D\ubar(\ell))\right) D_{\rho_2}w(\ell) \\
		&\qquad\leq C C_{\beta_0 \ctDist} \beta^{-1} e^{-\frac{1}{6}\beta \ctDist(\overline{u})} \cdot C_2 
		\sum_{\ell \in \Lambda}\sum_{\rho_1,\rho_2\in\Lambda-\ell} e^{-\ctCT ( |\rho_1| + |\rho_2| )} |D_{\rho_1} v(\ell)||D_{\rho_2} w(\ell)| 
		\\ 
		&\qquad\leq C\beta^{-1} e^{-\frac{1}{6}\beta \ctDist(\overline{u})} \|Dv\|_{\ell^2_\ctGamma}\|Dw\|_{\ell^2_\ctGamma}
	\end{split}
	\end{align}
	for all $\beta > \beta_0$. The constant $C$ in the final line depends on $\ctDist\coloneqq \ctDist(\overline{u})$ but this dependence is suppressed for notational simplicity. By the assumed strong stability \cref{eq:strongly_stable} and \cref{Thm1Part1:2ndvariationerror}, we immediately obtain the following stability estimate,
	\begin{align}
	\label{eq:stability_zero_temp_limit}
	\begin{split}
		\Braket{\delta^2\mathcal G^\beta(\ubar)v, v} &= 
			\Braket{\delta^2\mathcal G(\ubar)v, v} + 
			\Braket{\left(\delta^2\mathcal G^\beta(\ubar) - \delta^2\mathcal G(\ubar)\right)v, v} \\ 
		&\geq \left( c_0 - C \beta^{-1} e^{-\frac{1}{6}\beta\ctDist(\overline{u})}\right) \|Dv\|^2_{\ell^2_\ctGamma}.
	\end{split}\end{align}

	We now move on to consider consistency. It will be useful to consider the following truncation operator to split a given displacement into core and far field contributions \cite[Lemma~7.3]{EhrlacherOrtnerShapeev2013arxiv}: 
	\begin{lemma}[Truncation Operator]\label{lem:truncation}
		For $R>0$, there exist operators $T_R \colon \left(\mathbb R^d\right)^\Lambda \to \dot{\mathscr W}^c(\Lambda)$ such that $T_R u$ has compact support in $B_R$ and, for all $R$ sufficiently large, $DT_R u(\ell) = Du(\ell)$ for all $\ell \in \Lambda \cap B_{R/2}$ and 
		\begin{align*}
			\|DT_Ru - Du\|_{\ell^2_\ctGamma} &\leq C \|Du\|_{\ell^2_\ctGamma(\Lambda\setminus B_{R/2})}, \quad \text{and}\\
			\|DT_Ru\|_{\ell^2_\ctGamma} &\leq C	\|Du\|_{\ell^2_\ctGamma(\Lambda\cap B_{R})},
		\end{align*}
		where $C$ is independent of $R$ and $u$. 
	\end{lemma}
	We use the notation of Lemma~\ref{lem:truncation} and let $v^\mathrm{co} \coloneqq T_R v$ and $v^\mathrm{ff} = v - v^\mathrm{co}$ for some $R>0$ to be chosen later. In the following, we use the fact that $\delta \mathcal G(\overline{u}) = 0$ and estimate each of the terms in the following expression:
	\begin{align}
		\label{eq:splitting}
		\Braket{\delta \mathcal G^\beta(\overline{u}), v} = 
			\Braket{\delta \mathcal G^\beta(\overline{u}) - \delta \mathcal G(\overline{u}), v^\mathrm{co}}  
				+ \Braket{\delta \mathcal G^\beta(\overline{u}) - \delta \mathcal G(\overline{u}), v^\mathrm{ff}} 
	\end{align}

	\textit{Core. }Since the core region is finite, the first term of \cref{eq:splitting} is straightforward to deal with. Here, we simply apply the convergence of the site energies directly to obtain
	\begin{align}\label{eq:core}
	\begin{split}
		&\bigg|\Braket{\delta \mathcal G^\beta(\overline{u}) - \delta \mathcal G(\overline{u}), v^\mathrm{co}}\bigg|  \\
		&\qquad \leq \sum_{\above{\ell \in \Lambda, \rho \in \Lambda - \ell\colon}{|\ell|\leq R \text{ or } |\ell + \rho| \leq R } }
		\Big|\big( \mathcal G^\beta_{\ell,\rho}(D\overline{u}(\ell)) - \mathcal G_{\ell,\rho}(D\overline{u}(\ell))\big) 
			\cdot D_\rho v^\mathrm{co}(\ell)\Big| \\
		&\qquad \leq C \beta^{-1}e^{-\frac{1}{6}\ctDist(\overline{u})\beta} 
			\sum_{\above{\ell \in \Lambda, \rho \in \Lambda - \ell\colon}{|\ell|\leq R \text{ or } |\ell + \rho| \leq R } } 
			e^{-\ctCT |\rho|}  \big| D_\rho v^{\mathrm{co}}(\ell) \big|.
	\end{split} 
	\end{align}
	Now we may use the fact that $v^\mathrm{co}$ has compact support inside $B_R$, to conclude:
	\begin{align}\label{eq:core_1}\begin{split}
		\sum_{\above{\ell \in \Lambda}{|\ell|\geq R}} \sum_{\above{\rho \in \Lambda-\ell}{|\ell+\rho|\leq R}} 
			e^{-\eta|\rho|}  \big| D_\rho v^{\mathrm{co}}(\ell) \big|
		& \leq C R^{d/2} \bigg( \sum_{|\ell|\geq R} e^{-\eta (|\ell| - R)} \bigg)^{1/2} 
			\|Dv^\mathrm{co}\|_{\ell^2_\ctGamma} \\
		&\leq C R^{d/2}R^{(d-1)/2}\|Dv^\mathrm{co}\|_{\ell^2_\ctGamma}.
	\end{split}\end{align}
	In the exact same way,
	\begin{align}\label{eq:core_2}
		&\sum_{\above{\ell \in \Lambda}{|\ell| < R}} \sum_{\rho \in \Lambda-\ell}  e^{-\eta|\rho|}  \big| D_\rho v^{\mathrm{co}}(\ell) \big| 
			\leq C R^{d/2}R^{(d-1)/2} \|Dv^\mathrm{co}\|_{\ell^2_\ctGamma}.
	\end{align}
	Combining \cref{eq:core}, \cref{eq:core_1}, \cref{eq:core_2} and Lemma~\ref{lem:truncation} we have 
	\begin{align}\label{eq:const_core}
		\bigg|\Braket{\delta \mathcal G^\beta(\overline{u}) - \delta \mathcal G(\overline{u}), v^\mathrm{co}} \bigg| 
			\leq C \beta^{-1} e^{-\frac{1}{6} \ctDist(\overline{u}) \beta } R^{d/2}R^{(d-1)/2} \|Dv\|_{\ell^2_\ctGamma(\Lambda \cap B_{R})}.
	\end{align}

	\textit{Far-field. }We now turn our attention to the far field contribution in \cref{eq:splitting}. We will replace $\overline{u}$ with some compactly supported approximation $\widetilde{u}$ and show that the error in this approximation can be bounded appropriately. We then use the fact that $\widetilde{u}$ has compact support to bound the far field contribution to \cref{eq:splitting}. 

	We define $\widetilde{u} \coloneqq T_{\widetilde{R}} \overline{u}$ for some $0 < \widetilde{R} < R$ to be chosen later and note that, by Lemma~\ref{lem:truncation}, we have
	$
		\|D\widetilde{u} - D\overline{u}\|_{\ell^2_\ctGamma} 
		\leq C \|D\overline{u}\|_{\ell^2_\ctGamma(\Lambda \setminus B_{\widetilde{R}/2})}.
	$
	Therefore, by Lemma~\ref{ham_converge}, for $\widetilde{R}$ sufficiently large, we have
	\begin{align}\label{eq:spec_perturbation}
		\mathrm{dist}(\mu, \sigma(\Ham(u))) \geq \frac{1}{4} \ctDist(\overline{u}) \quad \text{for all } u \coloneqq t \overline{u} + (1-t)\widetilde{u} \text{ and } t \in [0,1].
	\end{align}
	The inequality in \cref{eq:spec_perturbation} implies that, for every displacement along the linear path between $\ubar$ and $\widetilde{u}$, we have uniform convergence rates in the site energies as $\beta \to \infty$ (as in Lemma~\ref{lem:convergence_site_energy}). Since we have perturbed the displacement, the exponent in the convergence estimates are reduced (in this case by a factor of $2$, but this factor is arbitrary). 

	We will now estimate the error committed by replacing $\overline{u}$ with the compactly supported displacement $\widetilde{u}$. By \cref{Thm1Part1:2ndvariationerror} and \cref{eq:spec_perturbation}, we have
	\begin{align}\label{eq:partIconsistency}
		\begin{split}
			&\Braket{\delta\mathcal G^\beta(\ubar) - \delta\mathcal G(\ubar), v^\mathrm{ff}} 
				- \Braket{\delta\mathcal G^\beta(\widetilde{u}) - \delta\mathcal G(\widetilde{u}), v^\mathrm{ff}}\\
			&\qquad= \int_0^1\Braket{\left(\delta^2\mathcal G^\beta(t\ubar + (1-t)\widetilde{u}) 
				- \delta^2\mathcal G(t\ubar+ (1-t)\widetilde{u})\right) (\ubar - \widetilde{u}) ,v^\mathrm{ff}} \mathrm{d}t \\
			&\qquad\leq C \beta^{-1} e^{-\frac{1}{12}\ctDist(\overline{u})\beta} 
				\|D(\ubar - \widetilde{u})\|_{\ell^2_\ctGamma} \|Dv^\mathrm{ff}\|_{\ell^2_\ctGamma} \\
			&\qquad \leq C \beta^{-1} e^{-\frac{1}{12}\ctDist(\overline{u})\beta} 
		    	\|D\overline{u}\|_{\ell^2_\ctGamma(\Lambda \setminus B_{\widetilde{R}/2})}  \|Dv\|_{\ell^2_\ctGamma(\Lambda \setminus B_{R/2})}.
		\end{split}
	\end{align}

	Now, since we are only considering the far field behaviour of $v$ and $\widetilde{u}$ is of compact support, we are able to show that
	$\Braket{\delta\mathcal G^\beta(\widetilde{u}) - \delta\mathcal G(\widetilde{u}), v^\mathrm{ff}}$
	decays exponentially in the buffer region $B_R \setminus B_{\widetilde{R}}$:
	\begin{align}\label{eq:FF_estimate}
		\bigg| \Braket{\delta \mathcal G^\beta(\widetilde{u}) - \delta \mathcal G(\widetilde{u}), v^\mathrm{ff} } \bigg| \leq C \beta^{-1} e^{-\frac{1}{12}\ctDist(\overline{u}) \beta} \widetilde{R}^{d/2} e^{-\eta(R-\widetilde{R})} \|Dv^\mathrm{ff}\|_{\ell^2_\ctGamma}
	\end{align}
	where $\eta \coloneqq \frac{1}{2}\mathfrak{m}\min\{\ctCT,\ctTBexponent\}$. A full proof of \cref{eq:FF_estimate} is given after the conclusion of the current proof.
	
	Therefore, by applying \cref{eq:splitting}, the estimate for the core region \cref{eq:const_core} and \cref{eq:partIconsistency} and choosing $\widetilde{R}$ and $R$ sufficiently large (independently of $\beta$) we obtain
	\begin{align*}
		\Big| \Braket{\delta \mathcal G^\beta(\overline{u}) , v } 
		\Big| \leq C \beta^{-1} e^{-\frac{1}{12}\ctDist(\overline{u}) \beta} 
		\left(
		R^{d/2}R^{(d-1)/2} + \|D\overline{u}\|_{\ell^2_\ctGamma(\Lambda\setminus B_{\widetilde{R}/2})} +
		\widetilde{R}^{d/2} e^{-\eta(R - \widetilde{R})}
		\right) \| Dv \|_{\ell^2_\ctGamma}.
	\end{align*}
	We may choose $R,\widetilde{R}$ in such a way as to obtain an exponential rate of convergence as $\beta\to \infty$. 

	Applying the Inverse Function Theorem \cite[Lemma~B.1]{LuskinOrtner2013}, we can conclude that, for sufficiently large $\beta$, there exist $\ubar_\beta \in \W(\Lambda)$ and $c_1>0$ such that 
	\begin{gather*}
		\|D\ubar_\beta - D\ubar\|_{\ell^2_\ctGamma} \leq 
		C e^{-\frac{1}{12}\ctDist(\overline{u})\beta}, \quad \delta\mathcal G^\beta(\ubar_\beta) = 0 \\
		\Braket{\delta^2\mathcal G^\beta(\ubar_\beta)v,v} \geq c_1 \|Dv\|_{\ell^2_\ctGamma}^2 
	\end{gather*}
	for all $v \in \W(\Lambda)$. 

Finally we consider the error in the energy. Using an analogous argument to that of \cite[Eq.~(78)]{ChenOrtner16}, we obtain  
\begin{align}\label{eq:error_energy_1}
    |\mathcal G^\beta(\overline{u}_\beta) - \mathcal G^\beta(\overline{u})| \leq 
    C\|D\overline{u}_\beta - D\overline{u}\|_{\ell^2_\ctGamma}^2.
\end{align}
In order to deal with the model error, we consider a compactly supported displacement $T_R\overline{u}$:
\begin{align*}\begin{split}
    &\big( \mathcal G^\beta(\overline{u}) - \mathcal G^\beta(T_R\overline{u}) \big)
    - \big( \mathcal G(\overline{u}) - \mathcal G(T_R\overline{u}) \big)  
    = \int_{0}^1 \Braket{\delta \mathcal G^\beta(u_t) - \delta \mathcal G(u_t), T_R\overline{u} - \overline{u} } \textrm{d}t
    %
    %
    %
\end{split}\end{align*}
where $u_t \coloneqq (1-t)\overline{u} + tT_R\overline{u}$. Therefore, choosing $R$ sufficiently large such that 
$\mathrm{dist}(\mu,\sigma(\Ham(u_t))) \geq \frac{1}{4}\mathsf{d}(\overline{u})$ for all $t \in [0,1]$,
we may replace $u_t$ with a compactly supported displacement $\widetilde{u}$ as in \cref{eq:partIconsistency} and \cref{eq:FF_estimate} and obtain
\begin{align}\label{eq:error_energy_2}
    \left|\big( \mathcal G^\beta(\overline{u}) - \mathcal G^\beta(T_R\overline{u}) \big)
    - \big( \mathcal G(\overline{u}) - \mathcal G(T_R\overline{u}) \big)\right|
    \leq C e^{-\frac{1}{12}\mathsf{d}(\overline{u})\beta} \|D\overline{u}\|_{\ell^2_\ctGamma(\Lambda \setminus B_{R/2})}.
\end{align}
Finally, by applying Lemma~\ref{ham_converge}, we obtain
\begin{align}\label{eq:error_energy_3}\begin{split}
    &|\mathcal G^\beta(T_R\overline{u}) - \mathcal G(T_R\overline{u})| \\
    &\leq C \sum_{\ell\in\Lambda} \sum_{1\leq a\leq \numorbitals} \max_{z \in \mathscr C^- \cup \mathscr C^+} |\mathfrak{g}^\beta(z;\mu) - \mathfrak{g}(z;\mu) | \left|\big[ \mathscr R_z(T_R\overline{u}) [\Ham(T_R\overline{u}) - \Ham(x)] \mathscr R_z(x) \big]_{\ell\ell}^{aa}\right| \\
    &\leq C \beta^{-1} e^{-\frac{1}{12}\mathsf{d}\beta} \sum_{a_1,a_2} \sum_{\above{\ell_1 \in \Lambda}{\ell_2 \in\Lambda\cap B_R}} \big|\big[\Ham(T_R\overline{u}) - \Ham(x) \big]_{\ell_1\ell_2}^{a_1a_2}\big|
    \leq C R^{d/2} \|D\overline{u}\|_{\ell^2_\ctGamma} \beta^{-1} e^{-\frac{1}{12}\mathsf{d}(\overline{u})\beta}
    \end{split}
\end{align}
where $\eta$ is some positive constant and $x\colon\Lambda\to\Lambda$ denotes the identity configuration.

Combining \cref{eq:error_energy_1}, \cref{eq:error_energy_2} and \cref{eq:error_energy_3}, we obtain 
$
    |\mathcal G^\beta(\overline{u}_\beta) - \mathcal G(\overline{u})| \leq Ce^{-\frac{1}{12}\mathsf{d}(\overline{u}) \beta}
$
as required.

\begin{proof}[Proof of \cref{eq:FF_estimate}]
We will argue that site energies are close to the corresponding reference site energies. We again define $\widetilde{\Ham}(\widetilde{u})$ and $\widetilde{\Ham}^\mathrm{ref}$ as in \cref{eq:extend_hamiltonian} so that we can compare these quantities.
	
	If $\ell \in (\Lambda^\mathrm{ref} \cup \Lambda) \cap B_{\widetilde{R}}$ or $k \in (\Lambda^\mathrm{ref} \cup \Lambda) \cap B_{\widetilde{R}}$, then 
	$|\ell - k| \geq \mathrm{dist}(\ell,B_{\widetilde{R}}) + \mathrm{dist}(k,B_{\widetilde{R}})$
	and so we have
	\begin{align}\label{eq:Ham_minus_Hamref}
		\left| \left[\widetilde{\Ham}(\widetilde{u}) - \widetilde{\Ham}^\mathrm{ref}\right]^{ab}_{\ell k}\right|
			&\leq C e^{-\ctTBexponent \mathfrak{m}\left( 	|\ell| + |k| - 2\widetilde{R} ) \right)} 
	\end{align}
	Similarly, for $m \in \Lambda$, 
	\begin{align}\label{eq:Ham_minus_Hamref_der}
		\left| \left[\widetilde{\Ham}(\widetilde{u})_{,m} - \widetilde{\Ham}^\mathrm{ref}_{,m}\right]^{ab}_{\ell k}\right|
		&\leq C e^{-\frac{1}{2}\ctTBexponent \mathfrak{m}\left( |\ell| + |k| - 2\widetilde{R} \right)} 
			e^{-\frac{1}{2}\ctTBexponent \left( r_{\ell m}(\widetilde{u}) + r_{km}(\widetilde{u}) \right) }.
	\end{align}
	In the following, we use the notation 
	${\mathscr R}_z(\widetilde{u}) \coloneqq (\widetilde{\Ham}(\widetilde{u}) - z)^{-1}$,
	and extend $v^\mathrm{ff}$ by zero to $\Lambda \cup \Lambda^\mathrm{ref}$. Since $\delta \mathcal G^\beta_\mathrm{ref}(\bm{0}) = 0$, we have 
	\begin{align}
		&\Braket{\delta \mathcal G^\beta(\widetilde{u}) - \delta \mathcal G(\widetilde{u}), v^\mathrm{ff}}
		= \Braket{\delta {\mathcal G}^\beta(\widetilde{u}) - \delta {\mathcal G}^\beta_\mathrm{ref} - 
		\big( \delta {\mathcal G}(\widetilde{u}) - \delta {\mathcal G}_\mathrm{ref}\big), v^\mathrm{ff}} \nonumber\\
		& = \frac{1}{2\pi i} \sum_{a} 
			\oint_{\mathscr C^- \cup \mathscr C^+} 
			\left[\mathfrak{g}^\beta(z;\mu) - \mathfrak{g}(z;\mu)\right]
			\sum_{\above{\ell \in \Lambda\cup \Lambda^\mathrm{ref}, \rho \in \Lambda \cup \Lambda^\mathrm{ref} - \ell }{|\ell| \geq R \text{ or } |\ell + \rho|\geq R}} 
			\frac{\partial [\mathscr R_z(\widetilde{u}) - {\mathscr R}^\mathrm{ref}_z]^{aa}_{\ell\ell}}{\partial u(\ell+\rho)}
			\cdot D_\rho v^\mathrm{ff}(\ell)\mathrm{d}z \nonumber\\
		&\leq C \beta^{-1} e^{-\frac{1}{12} \ctDist(\overline{u}) \beta} \sum_{\above{\ell,k \in \Lambda\cup \Lambda^\mathrm{ref}}{|\ell| \geq R \text{ or } |k|\geq R}} \max_{z \in\mathscr C^- \cup \mathscr C^+}\left| \frac{\partial [\mathscr R_z(\widetilde{u}) - {\mathscr R}^\mathrm{ref}_z]_{\ell\ell}}{\partial u(k)} \cdot D_{k-\ell} v^\mathrm{ff}(\ell)\right|.
\label{eq:ff}
\end{align}
	We have therefore reduced the problem to considering the derivatives of the difference of two resolvent operators:
	\begin{align}
		&\frac{\partial[\mathscr R_z(\widetilde{u}) - {\mathscr R}^\mathrm{ref}_z]^{aa}_{\ell\ell}}{\partial u(k)} 
			= \left[ -\mathscr R_z(\widetilde{u})\Ham(\widetilde{u})_{,k}\mathscr R_z(\widetilde{u}) 
			+ {\mathscr R}_z^\mathrm{ref}\widetilde{\Ham}_{,k}^\mathrm{ref}{\mathscr R}^\mathrm{ref}_z \right]^{aa}_{\ell\ell}\label{eq:diff_resolvents}\\
		&=  \left[ (\mathscr R_z^\mathrm{ref} - \mathscr R_z(\widetilde{u})) \widetilde{\Ham}^\mathrm{ref}_{,k}\mathscr R_z^\mathrm{ref} 
			+ {\mathscr R}_z(\widetilde{u})(\widetilde{\Ham}^\mathrm{ref}_{,k} - \widetilde{\Ham}(\widetilde{u})_{,k}) \mathscr R_z^\mathrm{ref}
			+ \mathscr R_z(\widetilde{u})\widetilde{\Ham}(\widetilde{u})_{,k}(\mathscr R_z^\mathrm{ref} - \mathscr R_z(\widetilde{u})) \right]^{aa}_{\ell\ell}.\nonumber
	\end{align}
	In the following, we shall drop the argument $(\widetilde{u})$. Now, since 
	${\mathscr R}_z - \mathscr R^\mathrm{ref}_z = {\mathscr R}_z(\widetilde{\Ham}^\mathrm{ref} - \widetilde{\Ham}){\mathscr R}^\mathrm{ref}_z $,
	we have: for $z\in \mathscr C^- \cup \mathscr C^+$, 
	\begin{align}
		&\sum_{\above{\ell,k \in \Lambda\cup \Lambda^\mathrm{ref}}{|\ell| \geq R \text{ or } |k|\geq R}} 
			\left| \left[ ({\mathscr R}_z - \mathscr R^\mathrm{ref}_z) \widetilde{\Ham}_{,k}{\mathscr R}_z\right]_{\ell\ell}^{aa}  
			\cdot D_{k-\ell} v^\mathrm{ff}(\ell)\right| \nonumber\\
		&\leq \sum_{\above{\ell,k \in \Lambda\cup \Lambda^\mathrm{ref}}{|\ell| \geq R \text{ or } |k|\geq R}}
			\sum_{\above{\ell_1,\ell_2,\ell_3,\ell_4 \in \Lambda \cup \Lambda^\mathrm{ref}}{1\leq a_1,a_2,a_3,a_4 \leq \numorbitals}} 
			\Big|
			[{\mathscr R}_z]^{aa_1}_{\ell\ell_3} 
			( \widetilde{\Ham}^\mathrm{ref} - \widetilde{\Ham})^{a_1a_2}_{\ell_3\ell_4} 
			[\mathscr R^\mathrm{ref}_z]^{a_2a_3}_{\ell_4\ell_1}
			[\widetilde{\Ham}_{,k}]^{a_3a_4}_{\ell_1\ell_2} 
			[{\mathscr R}_z]^{a_4a}_{\ell_2\ell}\Big| 
			\big| D_{k-\ell} v^\mathrm{ff}(\ell) \big|\nonumber\\
		&\leq C \sum_{\above{\ell,k,\ell_1,\ell_2,\ell_3,\ell_4 
		}{|\ell| \geq R \text{ or } |k|\geq R}} 
			%
			e^{-\ctCT \left( r_{\ell\ell_3} + r_{\ell_4\ell_1} + r_{\ell_2\ell} \right)} 
			%
			e^{-\frac{1}{2}\ctTBexponent\mathfrak{m} \left(|\ell_3| + |\ell_4| - 2\widetilde{R}\right)}
			e^{- \ctTBexponent\left( r_{\ell_1k} + r_{\ell_2k}\right)}
			\big| D_{k-\ell} v^\mathrm{ff}(\ell) \big|\nonumber\\
		&\leq C \Bigg(\sum_{\above{\ell,k \in \Lambda\cup \Lambda^\mathrm{ref}}{|\ell| \geq R \text{ or } |k|\geq R}} \bigg(\sum_{\ell_1,\ell_3,\ell_4}
			e^{-\eta \left( r_{\ell\ell_3} + r_{\ell_4\ell_1} + 	|\ell_3| - \widetilde{R} + |\ell_4| - \widetilde{R}  + r_{\ell_1k} \right)} \bigg)^{2} \Bigg)^{1/2}
			\|Dv^\mathrm{ff}\|_{\ell^2_\ctGamma} \label{eq:1stterm}
\end{align}
	where $\eta \coloneqq\frac{1}{2}\mathfrak{m}\min\{\ctCT,\frac{1}{2}\ctTBexponent\}$.
	
We now bound the first term in the product \cref{eq:1stterm}. Here we only consider the summation over $\ell, k \in \Lambda \cup \Lambda^\mathrm{ref}$ and $|\ell|\geq R$ (the case where $|\ell|\leq R$ and $|k|\geq R$ can be treated in a similar way):
	\begin{align}
		\sum_{\above{\ell,k,\ell_1,\ell_3,\ell_4}{|\ell| \geq R} }
		e^{-\eta \left( r_{\ell\ell_3} + r_{\ell_4\ell_1} + |\ell_3| + |\ell_4| - 2\widetilde{R} + r_{\ell_1k} \right)} 
		&\leq C \bigg(\sum_{\above{\ell,\ell_3}{|\ell| \geq R }} e^{-\eta \left( r_{\ell\ell_3} + |\ell_3| - \widetilde{R} \right)} \bigg)
		\bigg(\sum_{k,\ell_4} e^{-\eta \left( r_{k\ell_4} + |\ell_4| - \widetilde{R} \right)}  \bigg)\nonumber \\
		&\leq C \widetilde{R}^d e^{-\frac{1}{2}\eta(R-\widetilde{R})}.\label{eq:1stterm_1}
	\end{align}
	Here, the $\widetilde{R}^d$ comes from the second factor in the line above. The exact same argument can be used to bound the third term in \cref{eq:diff_resolvents} similarly. 
	
	We now consider the second term in \cref{eq:diff_resolvents}: for $z \in \mathscr C^- \cup \mathscr C^+$,
	\begin{align}\label{eq:2ndterm-}
		\begin{split}
		&\sum_{\above{\ell,k \in \Lambda\cup \Lambda^\mathrm{ref}}{|\ell| \geq R \text{ or } |k|\geq R}} 
		\left| \left[ {\mathscr R}^\mathrm{ref}_z (\widetilde{\Ham}_{,k} - \widetilde{\Ham}^\mathrm{ref}_{,k}) {\mathscr R}_z\right]_{\ell\ell}^{aa}  
		\cdot D_{k-\ell} v^\mathrm{ff}(\ell)\right| \\
		&\qquad\leq
			\sum_{\above{\ell,k \in \Lambda\cup \Lambda^\mathrm{ref}}{|\ell| \geq R \text{ or } |k|\geq R}}
			\sum_{\above{\ell_1,\ell_2\in\Lambda\cup\Lambda^\mathrm{ref}}{ 1\leq a_1,a_2\leq\numorbitals}}
			\Big|[\mathscr R^\mathrm{ref}_z]^{aa_1}_{\ell\ell_1}
			[\widetilde{\Ham}_{,k} - \widetilde{\Ham}^\mathrm{ref}_{,k}]_{\ell_1\ell_2}^{a_1a_2}
			[\mathscr R_z^\mathrm{ref}]_{\ell_2\ell}^{a_2a} \Big| \Big| D_{k-\ell} v^\mathrm{ff}(\ell) \Big|\\
		&\qquad\leq C\sum_{\above{\ell,k,\ell_1,\ell_2 \in \Lambda\cup \Lambda^\mathrm{ref}}{|\ell| \geq R \text{ or } |k|\geq R}}
			%
			e^{-\ctCT( r_{\ell\ell_1} + r_{\ell_2\ell})}
			%
			e^{-\frac{1}{2}\ctTBexponent\mathfrak{m} ( |\ell_1| + |\ell_2| - 2\widetilde{R})}
			e^{-\ctTBexponent( r_{\ell_1k} + r_{\ell_2k})}
			\big|D_{k-\ell} v^\mathrm{ff}(\ell)\big|\\
		&\qquad\leq C \sum_{\above{\ell,k \in \Lambda\cup \Lambda^\mathrm{ref}}{|\ell| \geq R \text{ or } |k|\geq R}} 
			\bigg( \sum_{\ell_1} e^{-\eta( r_{\ell\ell_1} +|\ell_1| - \widetilde{R} + r_{\ell_1 k} )} \bigg) 
			e^{-\frac{1}{2}\eta r_{\ell k}} \big| D_{k-\ell} v^\mathrm{ff}(\ell) \big| \\
		&\qquad\leq C \Bigg( \sum_{\above{\ell,k,\ell_1 \in \Lambda\cup \Lambda^\mathrm{ref}}{|\ell| \geq R \text{ or } |k|\geq R}} 
			e^{-\eta( r_{\ell\ell_1} +|\ell_1| - \widetilde{R} + r_{\ell_1 k} )} \Bigg)^{1/2}
			\| D v^\mathrm{ff} \|_{\ell^2_\ctGamma}.
		\end{split}
	\end{align}
	We again show that the prefactor in this expression is bounded: by summing over $k,\ell_1$ and $\ell$ (in that order) we have,
	\begin{align}\label{eq:2ndterm_2}\begin{split}
		&\sum_{\above{\ell,k,\ell_1 \in \Lambda\cup \Lambda^\mathrm{ref}}{|\ell| \geq R \text{ or } |k|\geq R}}
			e^{-\eta( r_{\ell\ell_1} +|\ell_1| - \widetilde{R} + r_{\ell_1 k} )}  
		\leq C e^{-\frac{1}{2}\eta(R - \widetilde{R})}.
	\end{split}\end{align}

	Therefore, after collecting \cref{eq:1stterm}$-$\cref{eq:2ndterm_2}	and applying \cref{eq:ff} and \cref{eq:diff_resolvents}, we obtain \cref{eq:FF_estimate}.
\end{proof}

\subsubsection{Proof of Proposition~\ref{prop:0T_limit_infinite_domain}: \texorpdfstring{$\beta \to \infty$}{Zero Temperature Limit} in the Grand Canonical Ensemble}

	We consider a sequence, $\ubar_{\beta_j}$, of solutions to \refGCE{problem:infinite_domain}{\beta_j}{\infty}{\mu} (with $\beta_j \to \infty$ as $j \to \infty$) such that 
	$\sup_j\|D\ubar_{\beta_j}\|_{\ell^2_\ctGamma} < \infty$. 
	Noting that, after factoring out a constant shift, $\dot{\mathscr W}^{1,2}(\Lambda)$ is a Hilbert space and so we may apply the Banach-Alaoglu theorem to conclude that there exists a $\ubar \in \dot{\mathscr W}^{1,2}(\Lambda)$ such that 
	\[ 
		\ubar_{\beta_j} \rightharpoonup \ubar \quad \text{in }\dot{\mathscr W}^{1,2} \text{ as }j \to \infty
	\]
	along a subsequence (which we do not relabel). Now, because $v \mapsto D_\rho v(\ell)$ is a linear functional for all $\ell \in \Lambda$ and $\rho \in \Lambda - \ell$, we have obtained \cref{eq:theorem1part2_1}. 
	
	To simplify notation, let us define the forces 
	\begin{align}\label{eq:FlbetatoFl}
		\mathcal F_\ell^\beta(u) \coloneqq 
		\frac{\partial \mathcal G^\beta(u)}{\partial u(\ell)} 
			\quad \text{and} \quad 
			\mathcal F_\ell(u) \coloneqq \frac{\partial \mathcal G(u)}{\partial u(\ell)}.
	\end{align}

	Since $\ubar_{\beta_j}$ solves \refGCE{problem:infinite_domain}{\beta_j}{\infty}{\mu}, we have
	\begin{align*}
		0 	= \Braket{\delta \mathcal G^{\beta_j}(\ubar_{\beta_j}), v} 
			= \sum_{\ell \in \Lambda}\mathcal F_\ell^{\beta_j}(\ubar_{\beta_j})\cdot v(\ell) 
				\quad \text{for all }v \in \dot{\mathscr W}^{1,2}(\Lambda).
	\end{align*}
	Let us fix $v \in \dot{\mathscr W}^{1,2}(\Lambda)$ with compact support in $B_{R_v}$ for some $R_v>0$. Now, it is sufficient to show that \[ \mathcal F_\ell^{\beta_j}(\ubar_{\beta_j})  \to \mathcal F_\ell(\ubar) \quad \text{as } j \to\infty\]for all $\ell \in \Lambda\cap B_{R_v}$. Here, we may apply Remark~\ref{rem:spectral_pollution} and the fact that $\mu$ is uniformly bounded away from $\sigma(\Ham(\overline{u}_{\beta_j}))$ to conclude that $\mathcal{G}(\overline{u})$ is differentiable.
	
	For sufficiently large $j$, 
	\begin{align}
	    \label{eq:Fl_ub_minus_u}
	    |\mathcal F_\ell^{\beta_j}(\overline{u}_{\beta_j})
	        - \mathcal F_\ell^{\beta_j}(\overline{u})|
	        \leq C \left( 
		e^{- \ctCT R_v} +
		\| D(\overline{u}_{\beta_j} - \overline{u}) \|_{\ell^2_\ctGamma(\Lambda \cap B_{2R_v})} 
		\right).
	\end{align}
    The proof of this estimate is given below. By first choosing $R_v$ and then $j$ sufficiently large, \cref{eq:Fl_ub_minus_u} can be made arbitrarily small.
    
	Applying Lemma~\ref{lem:convergence_site_energy}, we obtain
	\begin{align}
	\label{eq:Fl_bu_minus_bu}
		|\mathcal F_\ell^{\beta_j}(\overline{u})
	        - \mathcal F_\ell(\overline{u})|	& \leq C\beta_j^{-1} e^{-\frac{1}{6}\ctDist(\overline{u})\beta_j} \sum_{k \in \Lambda} e^{-\ctCT r_{\ell k}} 
		\leq C\beta_j^{-1} e^{-\frac{1}{6}\ctDist(\overline{u})\beta_j}.
	\end{align}

	Combining \cref{eq:Fl_ub_minus_u} and \cref{eq:Fl_bu_minus_bu}, we obtain $\mathcal F_\ell^{\beta_j}(\ubar_{\beta_j}) \to \mathcal F_\ell(\ubar)$ as $j \to \infty$ and so $\mathcal F_\ell(\ubar) = 0$.

\begin{proof}[Proof of \cref{eq:Fl_ub_minus_u}]
    We will prove a more general statement: for $\beta \in(0,\infty]$ and $u_1, u_2 \in \dot{\mathscr W}^{1,2}(\Lambda)$, with $\|D(u_1 - u_2)\|_{\ell^2_\ctGamma(\Lambda \cap B_{2R_v})}$ sufficiently small, we have
	\begin{align}\label{eq:Fl_general_formula}
		|\mathcal F^\beta_\ell(u_1) - \mathcal F^{\beta}_\ell(u_2)| \leq C \left( 
		e^{- \ctCT R_v} +
		\| D(u_1 - u_2) \|_{\ell^2_\ctGamma(\Lambda \cap B_{2R_v})} 
		\right).
	\end{align}
	This result is also true in the case of periodic displacements $u_1, u_2$ as will become clear in the proof.

	Using the chain rule, we obtain the formula:
	\begin{align}\label{eq:chainrule}
		\mathcal F_\ell^\beta(u) = \sum_{\rho\in\ell-\Lambda} \mathcal G^\beta_{\ell - \rho,\rho}(Du(\ell - \rho)) - \sum_{\rho \in\Lambda-\ell} \mathcal G^\beta_{\ell,\rho}(Du(\ell)),
	\end{align}
	which is valid for both $\beta < \infty$ and $\beta = \infty$.
	
	We first notice that
	\begin{align}
	\label{eq:differencesiteengergies}
		\mathcal G^\beta_{\ell,k-\ell}(Du_1(\ell)) - \mathcal G^\beta_{\ell,k-\ell}(Du_2(\ell)) = 
		&- \frac{1}{2\pi i} \sum_a\oint_{\mathscr C^+ \cup \mathscr C^-} 
			\mathfrak{g}^\beta(z;\mu) 
			\frac{\partial \left[ \mathscr R_z(u_1) - \mathscr R_z(u_2)\right]^{aa}_{\ell\ell} }{\partial u(k)} \mathrm{d}z.
	\end{align}
	We again consider the derivative of the difference of two resolvents as in \cref{eq:diff_resolvents}, above. In the following calculations we shall ignore the indices for the atomic orbitals. By Lemma~\ref{ham_converge}, if $\|D(u_1 - u_2)\|_{\ell^2_\ctGamma(\Lambda \cap B_{2R_v})}$ is sufficiently small, we have: for $m \in \Lambda$,
	\begin{align*}
		&\left[ \mathscr R_z(u_2) - \mathscr R_z(u_1) \right]_{\ell m} 
		= \left[\mathscr R_z(u_2)\left(\Ham(u_1) - \Ham(u_2)\right)\mathscr R_z(u_1)\right]_{\ell m} \\
		&\quad\leq C
		\sum_{\above{\ell_1,\ell_2 \in \Lambda}{|\ell_1| > 2R_v \textrm{ or } |\ell_2| > 2R_v}}
			e^{-\ctCT (r_{\ell\ell_1} + r_{\ell_2 m})} 
			e^{-\gamma_0 r_{\ell_1 \ell_2}}
		    \\&\qquad\qquad\qquad 
		    + C\bigg(\sum_{\ell_1,\ell_2 \in \Lambda \cap B_{2R_v}} 
		        e^{- \ctCT (r_{\ell\ell_1} + r_{\ell_2 m})}
		    \bigg)^{1/2}
		    \bigg(\sum_{\ell_1,\ell_2 \in \Lambda \cap B_{2R_v}} \left|\left[\Ham(u_1) - \Ham(u_2) \right]_{\ell_1\ell_2}\right|^2 \bigg)^{1/2} \\
		&\quad\leq C
		\left(
    		e^{-\frac{1}{2} \min\{\ctCT,\gamma_0\}  R_v} 
    		+ \| D(u_1 - u_2) \|_{\ell^2_\ctGamma(\Lambda \cap B_{2R_v})}
		\right).
	\end{align*}
	Here, we have used the fact that $\ell \in B_{R_v}$ in the first term.
	Therefore, we have: 
	\begin{align}
	    \label{eq:diff_two_resolvents_1}
	    \begin{split}
		&\big[ (\mathscr R_z(u_2) - \mathscr R_z(u_1)) \Ham(u_2)_{,k} \mathscr R_z(u_2) \big]_{\ell\ell}\\
		&\qquad\leq C \left(
		e^{-\ctCT R_{v} }
        + \| D(u_1 - u_2) \|_{\ell^2_\ctGamma(\Lambda \cap B_{2R_v})} \right)
		\sum_{\ell_1, \ell_2 \in \Lambda}
		e^{-\gamma_0 (r_{\ell_1 k} + r_{\ell_2 k}) }
		e^{-\ctCT r_{\ell_2 \ell} } \\ 
		&\qquad \leq C \left(
		e^{-\ctCT R_{v} }
        + \| D(u_1 - u_2) \|_{\ell^2_\ctGamma(\Lambda \cap B_{2R_v})} \right)
		e^{-\frac{1}{2} \min\{\ctCT, \gamma_0\} r_{\ell k} }.
		\end{split}
	\end{align}
    Similarly, by Lemma~\ref{ham_converge}, we have
	\begin{align}\label{eq:diff_two_resolvents_3}
	\begin{split}
		&\big[ \mathscr R_z(u_1) [\Ham(u_2)_{,k} - \Ham(u_1)_{,k}] \mathscr R_z(u_2) \big]_{\ell\ell}\\
		&\qquad\leq C\left(
		e^{-\ctCT R_{v} }
        + \| D(u_1 - u_2) \|_{\ell^2_\ctGamma(\Lambda \cap B_{2R_v})} \right)
		e^{-\frac{1}{2}\min\{\ctCT, c\ctTBexponent\} r_{\ell k}}
	\end{split}
	\end{align}
	where $c = \frac{\mathfrak{m}\sqrt{3}}{2}$ is the constant from Lemma~\ref{ham_converge}.
	
	Therefore, by combining \cref{eq:diff_two_resolvents_1} and \cref{eq:diff_two_resolvents_3} and using the formula for the derivative of the difference between two resolvent operators \cref{eq:diff_resolvents} together with the chain rule formula \cref{eq:chainrule}, we have \cref{eq:Fl_general_formula}.
\end{proof}

\subsubsection{Proofs of Theorem~\ref{cor:strong_0T_limit_infinite_domain} and Proposition~\ref{cor:weak_0T_limit_infinite_domain}: \texorpdfstring{$\beta \to \infty$}{Zero Temperature Limit} in the Canonical Ensemble}

Before we proceed with the proofs of Theorem~\ref{cor:strong_0T_limit_infinite_domain} and Proposition~\ref{cor:weak_0T_limit_infinite_domain}, we first recall that $\ep_\mathrm{F}^{\beta,R}(u_R)$ denotes the Fermi level given by \cref{eq:FL}. For $\beta = \infty$, we define $\ep_\mathrm{F}^{\infty,R}(u_R)$ via the zero Fermi-temperature limit (Lemma~\ref{lem:fermilevelconverge}).

\begin{proof}[Proof of Theorem~\ref{cor:strong_0T_limit_infinite_domain}]
    We suppose that $\ubar_R$ is a strongly stable solution to \refCE{problem:canonical_finitedomain}{\infty}{R}. In particular, $\ubar_R$ is a strongly stable solution to \refGCE{problem:grand_canonical_finitedomain}{\infty}{R}{\mu} for all $\mu \in I$ where $I$ is a closed interval such that $I \cap B_\delta(\sigma(\Ham^R(\ubar_R))) = \emptyset$ for some $\delta>0$ and $\ep_\mathrm{F}^{\infty,R}(\ubar_R) \in I$. Therefore, by Theorem~\ref{thm:0T_limit_infinite_domain}, for sufficiently large $\beta$ (depending on $\delta$ and not on $\mu$), there exists a unique solution $\ubar^\mu_{R,\beta}$ to \refGCE{problem:grand_canonical_finitedomain}{\beta}{R}{\mu} satisfying 
    \begin{equation}\label{eq:cor6_1}
        \|D\ubar^\mu_{R,\beta} - D\ubar_R\|_{\ell^2_\ctGamma} \leq C e^{-\frac{1}{12}\ctDist(\overline{u}_R)\beta} \eqqcolon \tau_\beta.
    \end{equation}
    Since $I \cap B_\delta(\sigma(\Ham(\ubar_R))) = \emptyset$, the pre-factor and exponent can be chosen to depend on $\delta$ but not on $\mu \in I$. Now by Lemma~\ref{ham_converge_R}, for $\beta$ sufficiently large, we have
    \begin{align}
        |\lambda_s(\ubar^\mu_{R,\beta}) - \lambda_s(\ubar_R)| \leq C e^{-\frac{1}{12}\ctDist(\overline{u}_R)\beta}
    \end{align}
    where $\lambda_s(u)$ denotes the eigenvalues of $\Ham^R(u)$ in increasing order (for $s = 1,\dots, N_R$). In particular, $\ep_\mathrm{F}^{\beta,R}(\overline{u}^\mu_{R,\beta}) \to \ep_\mathrm{F}^{\infty,R}(\overline{u}_R)$ as $\beta \to \infty$ and so, for sufficiently large $\beta$, we have $\ep_\mathrm{F}^{\beta,R}(\ubar_{R,\beta}^\mu) \in I$ for all $\mu \in I$. 
    
    For now, we assume that the mapping $I \to I$ given by $\mu \mapsto \ep_\mathrm{F}^{\beta,R}(\overline{u}_{R,\beta}^\mu)$ is continuous for all sufficiently large $\beta$. We will prove this fact after noting that this is sufficient to conclude. Since $I$ is a compact and convex set, we can apply Brouwer's fixed point theorem to conclude that there exists $\mu^\star = \ep_\mathrm{F}^{\beta,R}(\overline{u}^{\mu^\star}_{R,\beta})\in I$. In particular, $\overline{u}_{R,\beta}^{\mu^\star}$ is a solution to \refGCE{problem:grand_canonical_finitedomain}{\beta}{R}{\mu^\star} with Fermi level $\mu^\star$. That is, $\overline{u}^{\mu^\star}_{R,\beta}$ solves \refCE{problem:canonical_finitedomain}{\beta}{R} and, by \cref{eq:cor6_1}, we have
    $    
        \|D\overline{u}^{\mu^\star}_{R,\beta} - D\overline{u}_R\|_{\ell^2_\ctGamma} \leq Ce^{-\frac{1}{12}\ctDist(\overline{u}_R)\beta}.
    $

    \textit{Continuity of $\mu \mapsto \ep_\mathrm{F}^{\beta,R}(\overline{u}_{R,\beta}^\mu)$:}
    We now wish to show that $I\to I\colon\mu \mapsto \ep_\mathrm{F}^{\beta,R}(\overline{u}_{R,\beta}^\mu)$ is continuous for all sufficiently large $\beta$. To do so, we fix $\nu\in I$ and apply the inverse function theorem on $\delta G^{\beta,R}(\,\cdot\,;\nu)$ around $\overline{u}_{R,\beta}^\mu$ for $\mu\in I$ close to $\nu$. 
    
    Firstly, we remark that $\delta^2G^{\beta,R}(\,\cdot\,;\nu)$ is Lipschitz continuous in a neighbourhood of $\overline{u}_{R,\beta}^\mu$ for all $\mu \in I$ with Lipschitz constant uniformly bounded below by a positive constant for all $\mu \in I$. 
    
    Since $\mathfrak{g}^\beta(\,\cdot\,;\nu)$ is analytic on $\mathbb C\setminus\{\nu + ir\colon r \in \mathbb R\}$, we know that $\mathfrak{g}^\beta(\,\cdot\,;\nu)$ is Lipschitz continuous on all compact sets $K \subset \mathbb C\setminus\{\nu + ir\colon r \in \mathbb R\}$. In particular, if $\mu \in I$, then
    \begin{align*}
        |\mathfrak{g}^\beta(z;\mu) - \mathfrak{g}^\beta(z;\nu)| = |\mathfrak{g}^\beta(z + \nu - \mu;\nu) - \mathfrak{g}^\beta(z;\nu)| \leq L |\mu - \nu| 
    \end{align*}
    for all $z \in \mathscr C^- \cup \mathscr C^+$ for some appropriate choice of contours as in Figure~\ref{fig:contours} with 
    $\mathrm{dist}(\mathrm{Re}(z), I) \geq \frac{1}{2}\delta$
    for all $z \in \mathscr C^- \cup \mathscr C^+$. Since $\mathfrak{g}^\beta(z;\cdot) \to 2(z - \cdot\,)$ pointwise as $\beta \to \infty$, we can conclude that the Lipschitz constant can be chosen uniformly (for sufficiently large $\beta$). Using this, together with the stability of $\overline{u}_{R,\beta}^\mu$ (where the stability constant $c_1$ is independent of $\mu \in I$), we obtain
    \begin{align*}
        \Braket{\delta^2 G^{\beta,R}(\overline{u}_{R,\beta}^\mu;\nu)v,v} &=
        \Braket{\delta^2 G^{\beta,R}(\overline{u}_{R,\beta}^\mu;\mu)v,v} - 
        \Braket{\left(\delta^2 G^{\beta,R}(\overline{u}_{R,\beta}^\nu;\nu)- \delta^2 G^{\beta,R}(\overline{u}_{R,\beta}^\nu;\mu)\right)v,v} \\
        &\geq \left(c_1 - C|\mu-\nu|\right)\|Dv\|_{\ell^2_\ctGamma} \quad \textrm{and} \\
    \Braket{\delta G^{\beta,R}(\overline{u}_{R,\beta}^\mu;\nu),v} %
    &= \Braket{\delta G^{\beta,R}(\overline{u}_{R,\beta}^\mu;\nu) - \delta G^{\beta,R}(\overline{u}_{R,\beta}^\mu;\mu),v} \\
    &\leq C|\mu - \nu|.
    \end{align*}

    Therefore, if $|\mu - \nu|$ is sufficiently small, the inverse function theorem yields the existence of   $\overline{u}_{R,\beta}^{\mu\nu}$ satisfying
    \begin{align}\label{eq:convergence_munu}
       \delta G^{\beta,R}(\overline{u}_{R,\beta}^{\mu\nu};\nu) = 0
       \quad \textrm{and} \quad 
       \|D\overline{u}_{R,\beta}^{\mu\nu} - D\overline{u}_{R,\beta}^{\mu}\|_{\ell^2_\ctGamma} \leq C|\mu-\nu|. 
    \end{align}
    In particular, $\overline{u}_{R,\beta}^{\mu\nu}$ solves \refGCE{problem:grand_canonical_finitedomain}{\beta}{R}{\nu}. By \cref{eq:cor6_1}, if $|\nu-\mu|$ is sufficiently small, we necessarily have
    $\overline{u}_{R,\beta}^{\mu\nu} \in B_{\tau_\beta}(\overline{u}_R;\|D\cdot\|_{\ell^2_\ctGamma})$
    and thus, by uniqueness of the solution 
    $\overline{u}_{R,\beta}^{\nu}$
    to \refGCE{problem:grand_canonical_finitedomain}{\beta}{R}{\nu} on
    $B_{\tau_\beta}(\overline{u}_R;\|D\cdot\|_{\ell^2_\ctGamma})$,
    we have $\overline{u}_{R,\beta}^{\mu\nu} = \overline{u}_{R,\beta}^{\nu}$. Therefore, for all $\mu,\nu \in I$ with $|\mu - \nu|$ sufficiently small, 
    $\|D\overline{u}_{R,\beta}^{\nu} - D\overline{u}_{R,\beta}^{\mu}\|_{\ell^2_\ctGamma} \leq C|\mu-\nu|$ 
    and thus, by Lemma~\ref{ham_converge_R},
    $\mathrm{dist}\big(\sigma(\Ham^R(\overline{u}_{R,\beta}^\mu)), \sigma(\Ham^R(\overline{u}_{R,\beta}^\nu))\big) \leq C|\mu - \nu|$. 
    In particular, $\mu \mapsto \ep_\mathrm{F}^{\beta,R}(\overline{u}_{R,\beta}^\mu)$ is continuous on $I$.
\end{proof}
\begin{proof}[Proof of Proposition~\ref{cor:weak_0T_limit_infinite_domain}]
   Suppose that $\ubar_{\beta_j}$ is a bounded sequence (with $\beta_j \to \infty$) of solutions to \refCE{problem:canonical_finitedomain}{\beta_j}{R}. As before, along a subsequence, there exists a weak limit $\overline{u}_{\beta_j} \rightharpoonup \overline{u}$ as $j \to \infty$. By applying Lemma~\ref{ham_converge_R} and noting that weak convergence on a finite domain implies strong convergence, we can conclude that 
    \[ |\lambda_s(\overline{u}_{\beta_j}) - \lambda_s(\overline{u})| \to 0\]
    as $j \to \infty$ for each $s = 1,\dots,N_R$. Therefore $\ep_\mathrm{F}^{\beta_j,R}(\overline{u}_{\beta_j}) \to \ep_\mathrm{F}^{\infty,R}(\overline{u})$ as $j \to \infty$. 
    
     Since, we assume that $\ep_\mathrm{F}^{\beta_j,R}(\overline{u}_{\beta_j})$ is uniformly bounded away from $\sigma(\Ham^R(\overline{u}_{\beta_j}))$, we can apply Remark~\ref{rem:spectral_pollution} to conclude $\mu \coloneqq \ep_\mathrm{F}^{\infty,R}(\overline{u}) \not \in \sigma(\Ham^R(\overline{u}))$.
    
    Since $\overline{u}_{\beta_j}$ solves \refCE{problem:canonical_finitedomain}{\beta_j}{R},
    \begin{align}\label{eq:GCERpartii_2}\begin{split}
     \left| \frac{\partial G^{\beta_j, R}(\overline{u}_{\beta_j};\mu )}{\partial \overline{u}_{\beta_j}(\ell)} \right| 
     &= \left| \frac{\partial G^{\beta_j, R}(\overline{u}_{\beta_j};\mu )}{\partial \overline{u}_{\beta_j}(\ell)} 
     - \frac{\partial G^{\beta_j, R}(\overline{u}_{\beta_j};\tau )}{\partial \overline{u}_{\beta_j}(\ell)}\Bigg|_{\tau = \ep_\mathrm{F}^{\beta_j,R}(\overline{u}_{\beta_j})} \right| \\
     &\leq C \big|\ep_\mathrm{F}^{\beta_j,R}(\overline{u}_{\beta_j}) - \mu \big| \to 0 \quad \textrm{as }j \to \infty.
    \end{split}\end{align}
   
    On the other hand, as in the proof of Proposition~\ref{prop:0T_limit_infinite_domain} (see, \cref{eq:FlbetatoFl}$-$\cref{eq:Fl_bu_minus_bu}), we have 
    \begin{align}\label{eq:GCERpartii_1}
    \frac{\partial G^{\beta_j, R}(\overline{u}_{\beta_j};\mu )}{\partial \overline{u}_{\beta_j}(\ell)} 
    \to \frac{\partial G^{\infty, R}(\overline{u};\mu )}{\partial \overline{u}(\ell)}\quad \textup{as }j \to \infty.
    \end{align}

    Therefore, by combining \cref{eq:GCERpartii_1} and \cref{eq:GCERpartii_2}, we can conclude that $\overline{u}$ is a critical point of $G^{\infty,R}(\,\cdot\,;\mu)$.
\end{proof}
\subsection{Thermodynamic Limit}

The results of \cite{ChenLuOrtner18} are analogous to Theorem~\ref{thm:0T_thermodynamic_limit_CE_GCE} and Proposition~\ref{prop:0T_thermodynamic_limit_CE_GCE} but in the case of finite Fermi-temperature. Moreover, the authors of \cite{ChenLuOrtner18} only consider the case of a Bravais lattice $\Lambda^\mathrm{ref} = \mathsf{B}\mathbb Z^d$. In this section, we prove that these results can be extended to the case of zero Fermi-temperature for insulators in the more general case where $\Lambda^\mathrm{ref}$ need not be a Bravais lattice. 

In principle, one may prove these thermodynamic limit results by using Theorem~\ref{thm:0T_limit_infinite_domain} and Proposition~\ref{prop:0T_limit_infinite_domain} to compare the zero Fermi-temperature problems with the analogous finite temperature problems and showing that the convergence rates in \cite{ChenLuOrtner18} are independent of Fermi-temperature. However, we opt for a more direct approach here because the case $\Lambda^\mathrm{ref} \not= \mathsf{B}\mathbb Z^d$ was not considered in \cite{ChenLuOrtner18} and thus a rigorous treatment would be lengthy.

\subsubsection{Proof of Theorem~\ref{thm:0T_thermodynamic_limit}: \texorpdfstring{$R \to \infty$}{Thermodynamic Limit} in the Grand Canonical Ensemble}

Throughout this proof $\beta \in (0,\infty]$ will be fixed and therefore we omit the index corresponding to Fermi-temperature on the grand potential and the site energies. Again, we will use the notation of Lemma~\ref{lem:truncation} for the truncation operator $T_R\colon \W(\Lambda) \to \dot{\mathscr W}^\textrm{c}(\Lambda)$.

\textit{Step 1: Quasi-best approximation.}

For some $r>0$ sufficiently small, we have $x + B_{2r}(\ubar)\subset \mathrm{Adm}(\Lambda)$. We may choose $R$ sufficiently large such that $T_R\ubar\in B_r(\ubar)$ and so $x + B_r(T_R\ubar)\subset\mathrm{Adm}(\Lambda)$. We know that $\mathcal G\in C^3(\textrm{Adm}(\Lambda))$ and so $\delta\mathcal G$ and $\delta^2\mathcal G$ are Lipschitz continuous on $\textrm{Adm}(\Lambda)\cap B_r(\ubar)$. In particular, 
\begin{align}
   \label{eq:op_norm}
   \|\delta\mathcal G(\ubar) - \delta\mathcal G(T_R\ubar)\| &
        \leq C \|D\ubar - DT_R\ubar\|_{\ell^2_\ctGamma} 
        \leq C \|D\overline{u}\|_{\ell^2_\ctGamma(\Lambda\setminus B_{R/2})}, 
            \quad \text{and}\\
	\|\delta^2\mathcal G(\ubar) - \delta^2\mathcal G(T_R\ubar)\| &
	    \leq C \|D\ubar - DT_R\ubar\|_{\ell^2_\ctGamma} 
	    \leq C \|D\overline{u}\|_{\ell^2_\ctGamma(\Lambda\setminus B_{R/2})}.
\end{align}
%

\textit{Step 2: Consistency.}

We fix $v \colon \Lambda_R \to \mathbb R^d$ satisfying \asNonInterR. Since $v$ is periodic and not necessarily an admissible displacement on $\Lambda$, we consider the compactly supported displacement $\TRstar v$ for some $R^\star < R$ and extend by zero to $\Lambda$. Rewriting $\Braket{\delta G^R(T_R\ubar),v}$, we have 
\begin{align}\label{eq:T1T2T3}\begin{split}
	\Braket{\delta G^R(T_R\ubar),v} 
	&= \Braket{\delta G^R(T_R\ubar),(I - \TRstar)v}
	           + \Braket{\delta G^R(T_R\ubar)-\delta \mathcal G(T_R\ubar),\TRstar v} \\
	           &\qquad\qquad+ \Braket{\delta \mathcal G(T_R\ubar)-\delta \mathcal G(\ubar),\TRstar v}
\end{split}\end{align}
We consider each of these contributions in turn. 

Replacing $T_R\ubar$ with $\TRtil\ubar$ for some $\tilde{R} < R^\star < R$ in the first term of \cref{eq:T1T2T3} gives an approximation error of $C \|D\ubar \|_{\ell^2_\ctGamma(\Lambda \cap B_R \setminus B_{\tilde{R}/2})}$. Since 
$(I - \TRstar)v = 0$ on $\Lambda \cap B_{R^\star}$ 
and 
$\TRtil\ubar = 0$ on $\Lambda \setminus B_{\tilde{R}}$,
we can bound $\Braket{\delta G^R(\TRtil\ubar),(I - \TRstar)v}$ as follows:
\begin{align}
\label{eq:T1_}
\begin{split}
    &|\Braket{\delta G^R(T_R\ubar),(I - T_{R^\star})v}| \\ 
    &\qquad\leq |\Braket{\delta G^R(\TRtil\ubar),(I - T_{R^\star})v}| 
            + |\Braket{\delta G^R(T_R\ubar) - \delta G^R(\TRtil\ubar),(I - T_{R^\star})v}| \\
   &\qquad\leq C\tilde{R}^{d/2} e^{-\eta(R^\star - \tilde{R})} \|Dv\|_{\ell^2_\ctGamma(\Lambda \setminus B_{R^\star/2})} 
   + C \|D\ubar \|_{\ell^2_\ctGamma(\Lambda \cap B_R \setminus B_{\tilde{R}/2})}\|Dv\|_{\ell^2_\ctGamma(\Lambda \setminus B_{R^\star/2})}
\end{split}
\end{align}
where $\eta \coloneqq \frac{1}{2}\mathfrak{m}\min\{\ctCT, \frac{1}{2}\ctTBexponent\}$. The first term in \cref{eq:T1_} is bounded by comparing the first derivative of the grand potential with the corresponding reference grand potential and taking the derivatives inside the contour integration in an argument that is exactly the same as in \cref{eq:FF_estimate}. Bounding the second term in \cref{eq:T1_} is done by applying Taylor's theorem and using the locality of the second derivatives of the site energies (for an identical argument see \cref{eq:partIconsistency}). 
	
Next, we consider the second term of \cref{eq:T1T2T3}. We simply apply the convergence of the site energies as $R \to \infty$ (Lemma~\ref{eq:thermodynamic_site_energies}), together with the locality of the site energies (Lemma~\ref{lem:CT}) and the fact that $T_{R^\star}v$ has compact support in $B_{R^\star}$ to conclude:
\begin{align}
\label{eq:consistency_T2_}
& \Braket{\delta G^R(T_R\ubar)-\delta \mathcal G(T_R\ubar),\TRstar v} \nonumber\\
&\qquad=\sum_{\ell \in \Lambda_R}\sum_{\rho \in \Lambda_R - \ell} \left(\mathcal G^R_{\ell,\rho}(DT_R\ubar(\ell)) - \mathcal G_{\ell,\rho}(DT_R\ubar(\ell))\right)\cdot D_\rho T_{R^\star} v(\ell) \nonumber\\
    &\qquad \qquad- 
    \sum_{\above{\ell \in \Lambda, \rho \in \Lambda - \ell}{\ell \not \in \Lambda_R \textrm{ or } \ell + \rho \not \in \Lambda_R}} 
    \mathcal G_{\ell,\rho}(DT_R\ubar(\ell))\cdot 
    D_\rho T_{R^\star} v(\ell)\nonumber\\
&\qquad\leq C\sum_{\ell \in \Lambda_R}
\sum_{\rho \in \Lambda_R - \ell} 
e^{-\eta (\mathrm{dist}(\ell, \Omega_R^c) + |\rho|)}
|D_\rho T_{R^\star}v(\ell)| 
    + C\sum_{\above{\ell \in \Lambda, \rho \in \Lambda - \ell}{\ell \not \in \Lambda_R \textrm{ or } \ell + \rho \not \in \Lambda_R}} 
    e^{-\ctCT |\rho|}
    |D_\rho T_{R^\star}v(\ell)| \nonumber\\
&\qquad\leq C(R^\star)^{d/2} e^{-\frac{1}{2}\eta(R - R^\star)}
\|DT_{R^\star}v\|_{\ell^2_\ctGamma}
\end{align}
where $\eta \coloneqq \frac{1}{2}\mathfrak{m}\min\{\ctCT, \frac{1}{2}\ctTBexponent\}$. In the final line, we have used the fact that $T_{R^\star}v$ has compact support in $B_{R^\star}$. More specifically, in the first term, we have used the fact that, if $\ell \in B_{R^\star}$ or $\ell + \rho \in B_{R^\star}$, then $\mathrm{dist}(\ell, \Omega_R^c) + |\rho| > R - R^\star$. Moreover, in the second term, we have that, if $\ell \not \in \Lambda_R$ then we must sum over $\ell + \rho \in B_{R^\star}$ and so $|\rho| > R - R^\star$ (and \textit{vice versa}).

Finally, we consider the third contribution from \cref{eq:T1T2T3}: by \cref{eq:op_norm}, we have
\begin{align}
	&|\Braket{\delta \mathcal G(T_R\ubar)-\delta \mathcal G(\ubar),\TRstar v}| 
    \leq C\|D\overline{u}\|_{\ell^2_\ctGamma(\Lambda \setminus B_{R/2})} \|D\TRstar v\|_{\ell^2_\ctGamma}. \label{eq:consistency_T3_}
\end{align}
Combining \cref{eq:T1T2T3}$-$\cref{eq:consistency_T3_} we obtain the following consistency estimate:
\begin{align}\label{eq:consistency}
    \left| \Braket{\delta G^R(T_R\ubar),v} \right| 
    \leq C\left( 
    \widetilde{R}^{d/2} e^{-\eta(R^\star - \tilde{R})} +
    \|D\ubar\|_{\ell^2_\ctGamma(\Lambda \setminus B_{\tilde{R}/2})} + 
    (R^\star)^{d/2} e^{-\frac{1}{2}\eta (R-R^\star)}  
    \right) 
    \|Dv\|_{\ell^2_\ctGamma}
\end{align}
where $\eta \coloneqq \frac{1}{2}\mathfrak{m}\min\{\ctCT, \frac{1}{2}\ctTBexponent\}$. Here, we can see that if $\|D\overline{u}\|_{\ell^2_{\ctGamma}(\Lambda \setminus B_R)} \lesssim R^{-d/2}$ (which would follow if \cref{eq:far_field_decay} holds), then we obtain a convergence rate as discussed in Remark~\ref{rem:convergence_rates}.
%

\textit{Step 3: Stability.}

We now show the following stability estimate: there exists $c_1 > 0$ such that
\begin{align}\label{eq:stab_estimate}
    \Braket{\delta^2 G^{R}(T_R\overline{u})v,v} \geq c_1 \|Dv\|_{\ell^2_\ctGamma}^2
\end{align}
for all sufficiently large $R$.

We first take a sequence $v_R$ of test functions with $\|Dv_R\|_{\ell^2_\ctGamma} = 1$ and note that $w_R \coloneqq T_R v_R \in \dot{\mathscr W}^{1,2}(\Lambda)$ and 
$\|Dw_R\|_{\ell^2_\ctGamma} \leq C \|Dv_R\|_{\ell^2_\ctGamma(\Lambda\cap B_R)} \leq C$ 
where $C$ is independent of $R$. Therefore, along a subsequence (which we do not relabel) we have $w_{R} \rightharpoonup v$ in $\dot{\mathscr W}^{1,2}(\Lambda)$ as $R \to \infty$ for some $v \in \dot{\mathscr W}^{1,2}(\Lambda)$. By \cite[Lemma~7.8]{EhrlacherOrtnerShapeev2013arxiv}, we may choose a sequence of radii $S(R)$ with $S(R) \to \infty$ as $R \to \infty$ ``sufficiently slowly'' such that
\begin{align}
    T_{S(R)}w_R \to v \quad&\textrm{strongly in } \dot{\mathscr W}^{1,2}(\Lambda),
     \label{eq:strong_convergence_FF}\\
    T_{S(R)}w_R - w_R \rightharpoonup 0 \quad&\textrm{weakly in } \dot{\mathscr W}^{1,2}(\Lambda) \textrm{ and} \label{eq:weak_convergence_FF}\\
    R - S(R) \to \infty \quad&\textrm{as } R \to \infty.
    \label{eq:RminusS}
\end{align}
We let $v^\mathrm{co}_R \coloneqq T_{S(R)}w_R$ and $v^{\mathrm{ff}}_R \coloneqq v_R - v_R^\mathrm{co}$ and expand $\Braket{\delta^2G^R(T_R\overline{u})v_R,v_R}$ as follows:
\begin{align}
        &\Braket{\delta^2G^R(T_R\overline{u})v_R,v_R} \nonumber\\
        &\qquad= 
        \Braket{\delta^2 G^R(T_R\overline{u})v_R^\mathrm{co},v_R^\mathrm{co}} 
        + 2\Braket{\delta^2 G^R(T_R\overline{u})v_R^\mathrm{co},v_R^\mathrm{ff}} 
        + \Braket{\delta^2 G^R(T_R\overline{u})v_R^\mathrm{ff},v_R^\mathrm{ff}} \nonumber\\
        &\qquad\eqqcolon \mathrm{T}_1 + 2\mathrm{T}_2 + \mathrm{T}_3.
        \label{eq:stability_expansion}
\end{align}
We shall consider each of these terms separately.

Using the fact that $v_R^\mathrm{co}$ has compact support in $B_{S(R)}$, we obtain:
\begin{align}\label{eq:T1T2_general_bound}\begin{split}
    &\Braket{\left(\delta^2 G^R(T_R\overline{u}) - \delta^2 \mathcal G(\overline{u})  \right)
    v_R^\mathrm{co},w} \\
    &\hspace{2cm}\leq C \left( e^{-\frac{1}{2}\eta(R - S(R))} + \|D\overline{u}\|_{\ell^2_\ctGamma(\Lambda \setminus B_{R/2})}\right)\|Dv_R^\mathrm{co}\|_{\ell^2_\ctGamma} \|Dw\|_{\ell^2_\ctGamma}
    \end{split}
\end{align}
where $\eta \coloneqq \frac{1}{2}\mathfrak{m}\min\{\ctCT, \frac{1}{2}\ctTBexponent\}$. The proof of this estimate is similar to that of \textit{Step 2} and is shown in \Cref{app:stab} for completeness.

\textit{$\mathrm{T}_1$: core term.} 
Using the strong stability of the solution $\overline{u}$ together with \cref{eq:T1T2_general_bound}, we may conclude that, for sufficiently large $R$, we have
\begin{align}\label{eq:T1_core}\begin{split}
    \mathrm{T}_1 
    &= \Braket{\delta^2 G^R(T_R\overline{u})v_R^\mathrm{co},v_R^\mathrm{co}} \\
    &= \Braket{\delta^2 \mathcal G(\overline{u})v_R^\mathrm{co},v_R^\mathrm{co}} - 
    \Braket{\left(\delta^2 \mathcal G(\overline{u}) - \delta^2 G^R(T_R\overline{u})  \right)v_R^\mathrm{co},v_R^\mathrm{co}} \\
    &\geq \frac{c_0}{2} \|Dv_R^\mathrm{co}\|_{\ell^2_\ctGamma}^2.
\end{split}\end{align}

\textit{$\mathrm{T}_2$: cross term}. Now, we show that the cross term in \cref{eq:stability_expansion} vanishes in the $R\to\infty$ limit. Rewriting $\mathrm{T}_2$, we have
\begin{align}\label{eq:T2_stability_R}
\begin{split}
    \mathrm{T}_2 &= \Braket{\delta^2 G^R(T_R\overline{u})v_R^\mathrm{co},v_R^\mathrm{ff}} \\
    &= \Braket{\left(\delta^2 G^R(T_R\overline{u}) 
        - \delta^2\mathcal G(\overline{u})\right)v_R^\mathrm{co}, v^\mathrm{ff}_R} 
    + \Braket{\delta^2 \mathcal G(\overline{u})(v_R^\mathrm{co}-v), v^\mathrm{ff}_R}
    + \Braket{\delta^2 \mathcal G(\overline{u})v, v^\mathrm{ff}_R} 
\end{split}
\end{align}
By \cref{eq:T1T2_general_bound}, the first term of \cref{eq:T2_stability_R} vanishes as $R\to\infty$ and, since $v_{R}^\mathrm{co}\to v$ strongly in $\dot{\mathscr{W}}^{1,2}(\Lambda)$, the second term in \cref{eq:T2_stability_R} also vanishes:
\begin{align*}
    \Braket{\delta^2 \mathcal G(\overline{u})(v_R^\mathrm{co} - v), v^\mathrm{ff}_R}
    &\leq C\|D(v^\mathrm{co}_R - v)\|_{\ell^2_\ctGamma} \|Dv^\mathrm{ff}_R\|_{\ell^2_\ctGamma} \to 0 \quad \textrm{as }R\to\infty.
\end{align*}
Finally, since $\delta^2 \mathcal G(\overline{u})v$ is a bounded linear functional on $\dot{\mathscr{W}}^{1,2}(\Lambda)$, we may apply the Riesz representation theorem to conclude that there exists $\Phi \in \dot{\mathscr W}^{1,2}(\Lambda)$ such that 
\begin{align*}
    \Braket{\delta^2 \mathcal G(\overline{u})v, v^\mathrm{ff}_R} &=
    \Braket{D\Phi,Dv_R^\mathrm{ff}}_{\ell^2_\ctGamma}.
\end{align*}
This quantity vanishes as $R \to \infty$ by the weak convergence of $v_R^\mathrm{ff} \rightharpoonup 0$ as $R\to\infty$ (see \cref{eq:weak_convergence_FF}).

\textit{$\mathrm{T}_3$: far field term}. Since $v^\mathrm{ff}_R$ only sees the far field behaviour of the test function, and not the point defect, we may replace $\delta^2  G^R(T_R\overline{u})$ by $\delta^2  G^R_\mathrm{ref}$ giving an approximation error of $C\|DT_R\overline{u}\|_{\ell^2_\ctGamma(\Lambda_R \setminus B_{S(R)/2})}$. To do this, we extend $v_R^\mathrm{ff}$ by zero to $\Lambda_R \cup \Lambda^\mathrm{ref}_R$ and note that,
\begin{align}\label{eq:T3minus_}\begin{split}
    \left|\mathrm{T}_3 - \Braket{\delta^2  G^{R}_\mathrm{ref}v^\mathrm{ff}_R,v^\mathrm{ff}_R} \right| 
    &=\left|\Braket{\left(\delta^2  G^R(T_R\overline{u}) - \delta^2  G_{\mathrm{ref}}^R\right)v^\mathrm{ff}_R,v^\mathrm{ff}_R}\right| \\
    &\leq C \|DT_R\overline{u}\|_{\ell^2_\ctGamma(\Lambda_R \setminus B_{S(R)/2})}
    \|Dv_R^{\mathrm{ff}}\|^2_{\ell^2_\ctGamma}. 
\end{split}\end{align}
It is now sufficient to prove that there exists $c_1 > 0$ such that
\begin{align}\label{eq:stability_ref}
    \Braket{\delta^2 G^{R}_\mathrm{ref}v^\mathrm{ff}_R,v^\mathrm{ff}_R} \geq c_1 \|Dv_R^\mathrm{ff}\|_{\ell^2_\ctGamma}^2.
\end{align}
A proof of this fact for Bravais lattices can be found in \cite{HudsonOrtner2011} which can be adapted to the multi-lattice setting. We give an alternative proof in \Cref{app:stab} for completeness.

Therefore, applying \cref{eq:T3minus_} and \cref{eq:stability_ref} we can conclude that
\begin{align*}
   \mathrm{T}_3 = \Braket{\delta^2  G^{R}(T_R\overline{u}) v_R^\mathrm{ff}, v_R^\mathrm{ff}} \geq \frac{c_1}{2} \|Dv_R^\mathrm{ff}\|_{\ell^2_\ctGamma}^2
\end{align*}
for all $R$ sufficiently large.

Using the fact that $\|Dv^\mathrm{co}_R\|_{\ell^2_\ctGamma}^2 + \|Dv^\mathrm{ff}_R\|_{\ell^2_\ctGamma}^2 \geq  \frac{1}{2} \|Dv_R\|_{\ell^2_\ctGamma}^2$ for all sufficiently large $R$, which follows from \cite[Lemma~7.9]{EhrlacherOrtnerShapeev2013arxiv}, allows us to conclude the proof of the stability estimate \cref{eq:stab_estimate}.

\textit{Step 4: Application of the Inverse Function Theorem.}
The consistency \cref{eq:consistency} and stability \cref{eq:stab_estimate} estimates allow us to apply the inverse function theorem \cite[Lemma~B.1]{LuskinOrtner2013} to conclude: for sufficiently large $R$, there exists $\ubar_R$ such that 
\[ 
\delta G^R(\ubar_R) = 0 \quad \text{and} \quad 
\|D\ubar_R - D\ubar\|_{\ell^2_\ctGamma}\to 0 \text{ as }R\to\infty. 
\]
Moreover, there exists a constant $c_2>0$ such that 
\[ 
\Braket{\delta^2 G^R(\ubar_R)v,v} \geq c_2 \|Dv\|_{\ell^2_\ctGamma}^2. 
\]
%

\subsubsection{Proof of Proposition~\ref{prop:0T_thermodynamic_limit}: \texorpdfstring{$R \to \infty$}{Thermodynamic Limit} in the Grand Canonical Ensemble}

Since $\sup_j \|D\overline{u}_{R_j}\|_{\ell^2_\ctGamma} <\infty$, there exists a $\ubar \in \dot{\mathscr W}^{1,2}(\Lambda)$ such that $\ubar_{R_j}\rightharpoonup \ubar$ along a subsequence as $j \to \infty$. Using Lemma~\ref{lem:spectral_pollution} and the fact $\mu$ is uniformly bounded away from $\sigma(\Ham^{R_j}(\overline{u}_{R_j}))$, we obtain $\mu \not \in \sigma(\Ham(\overline{u}))$.

We wish to show that $\delta \mathcal G(\ubar) = 0$. Using the notation from \cref{eq:FlbetatoFl} and noting that $\overline{u}_{R_j}$ solves \refGCE{problem:grand_canonical_finitedomain}{\infty}{R_j}{\mu} we have, for all $v \in \dot{\mathscr W}^{1,2}(\Lambda)$,
\begin{align}
0 = \Braket{\delta G^{R_j}(u_{R_j}), v} = \sum_{\ell\in \Lambda_{R_j}} \mathcal F^{R_j}_\ell(u_{R_j})\cdot v(\ell).
\end{align}
It is sufficient to suppose that $\mathrm{supp}(v) \subset B_{R_v}$ for some $R_v > 0$ and show that $\mathcal F_\ell^{R_j}(\overline{u}_{R_j}) \to \mathcal F_\ell(\overline{u})$ as $j \to \infty$:
\begin{align}\label{eq:GCEweakconvergence}\begin{split}
    |\mathcal F_\ell^{R_j}(\overline{u}_{R_j}) - \mathcal F_\ell(\overline{u})| &\leq
    |\mathcal F_\ell^{R_j}(\overline{u}_{R_j}) - \mathcal F^{R_j}_\ell(\overline{u})| + 
    |\mathcal F_\ell^{R_j}(\overline{u}) - \mathcal F_\ell(\overline{u})|.
\end{split}\end{align}
The first term of \cref{eq:GCEweakconvergence} may be treated in the exact same way as in \cref{eq:Fl_general_formula} 
to conclude that $|\mathcal F_\ell^{R_j}(\overline{u}_{R_j}) - \mathcal F^{R_j}_\ell(\overline{u})| \to 0$ as $j\to\infty$.

For the second term of \cref{eq:GCEweakconvergence} we may use the chain rule formula \cref{eq:chainrule} to write the forces as sums over site energies. Using the fact that $\ell \in B_{R_v}$ and Lemma~\ref{eq:thermodynamic_site_energies}, we have
\begin{align*}
    &\sum_{\rho \in \Lambda_{R_j} - \ell} 
    \left(
    \mathcal G^{R_j}_{\ell-\rho,\rho}(Du(\ell)) - \mathcal G_{\ell,\rho}(Du(\ell)) 
    \right) - \sum_{\rho \in \Lambda\setminus\Lambda_{R_j} - \ell} \mathcal G_{\ell,\rho}(Du(\ell)) \\
    &\qquad\leq \sum_{\rho \in \Lambda_{R_j} - \ell} 
    e^{-\eta(\mathrm{dist}(\ell,\Omega_{R_j}^c) + |\rho|)} 
        - \sum_{\rho \in \Lambda\setminus\Lambda_{R_j} - \ell} e^{-\ctCT|\rho|}\\
    &\qquad\leq C e^{-\frac{1}{2}\eta (R_j - R_{v})}.
\end{align*}
where $\eta \coloneqq \frac{1}{2}\mathfrak{m}\min\{\ctCT,\frac{1}{2}\ctTBexponent\}$.
We can therefore conclude that $|\mathcal F_\ell^{R_j}(\overline{u}) - \mathcal F_\ell(\overline{u})| \to 0$ as $j \to \infty$. That is, $\mathcal F_\ell(\overline{u}) = 0$ for all $\ell \in \Lambda \cap B_{R_v}$.

\subsubsection{Proofs of Theorem~\ref{thm:0T_thermodynamic_limit_CE_GCE} and Proposition~\ref{prop:0T_thermodynamic_limit_CE_GCE}: \texorpdfstring{$R \to \infty$}{Thermodynamic Limit} in the Canonical Ensemble}

\begin{proof}[Proof of Theorem~\ref{thm:0T_thermodynamic_limit_CE_GCE}]
We suppose that $\overline{u}$ is a strongly stable solution to \refGCE{problem:infinite_domain}{\infty}{\infty}{\mu}. By Theorem~\ref{thm:0T_thermodynamic_limit}, there is a sequence of solutions $\overline{u}_R$ to \refGCE{problem:grand_canonical_finitedomain}{\infty}{R}{\mu} with $\overline{u}_R \to \overline{u}$ in $\dot{\mathscr W}^{1,2}(\Lambda)$ as $R \to \infty$. This strong convergence means that, by Lemmas~\ref{ham_converge_R} and \ref{lem:spectral_pollution}, every isolated eigenvalue of $\sigma(\Ham(\overline{u}))$ is a limit point of a sequence of eigenvalues contained in $\sigma(\Ham^R(\overline{u}_R))$ and the accumulation points of every such sequence are contained in $\sigma(\Ham(\overline{u}))$. Since $\mu \not \in \sigma(\Ham(\overline{u}))$, we can find adjacent points $\underline{\ep}, \overline{\ep} \in \sigma(\Ham(\overline{u}))$ such that $\underline{\ep} < \mu < \overline{\ep}$. Now, choosing the electron number %
$N_{e,R} \coloneqq N^{\infty,R}\left(\tfrac{1}{2}(\underline{\ep} + \overline{\ep})\right)$,
and by applying Lemma~\ref{lem:fermilevelconverge}, we can conclude that $\ep_\mathrm{F}^{\infty,R}(\overline{u}_R) \to \tfrac{1}{2}(\underline{\ep} + \overline{\ep})$ as $R \to \infty$. Further, the interval between $\ep_\mathrm{F}^{\infty,R}(\overline{u}_R)$ and $\tfrac{1}{2}(\underline{\ep} + \overline{\ep})$ does not intersect $\sigma(\Ham(\overline{u}))$ for all sufficiently large $R$. That is, for all sufficiently large $R$, $\overline{u}_R$ solves \refCE{problem:canonical_finitedomain}{\infty}{R}.
\end{proof}

\begin{proof}[Proof of Proposition~\ref{prop:0T_thermodynamic_limit_CE_GCE}]
    Since $\|D\overline{u}_{R_j}\|_{\ell^2_\ctGamma}$ is uniformly bounded, along a subsequence, 
    $\overline{u}_{R_j} \rightharpoonup \overline{u}$ as $j \to \infty$
    for some $\overline{u} \in \dot{\mathscr W}^{1,2}(\Lambda)$. Now, because $\sigma(\Ham^{R_j}(\overline{u}_{R_j}))$ is uniformly bounded, 
    $\ep_\mathrm{F}^{\infty, R_j}(\overline{u}_{R_j}) \to \mu$
    along a further subsequence as $j\to\infty$ for some $\mu \in \mathbb R$.
    
    Supposing that $\sigma(\Ham^{R_j}(\overline{u}_{R_j}))$ is (eventually) bounded away from 
    $\ep_\mathrm{F}^{\infty,R_j}(\overline{u}_{R_j})$,
    we know that $\mu$ is eventually bounded away from $\sigma(\Ham^{R_j}(\overline{u}_{R_j}))$. This means that, for all $j$ sufficiently large, $\overline{u}_{R_j}$ solves \refGCE{problem:grand_canonical_finitedomain}{\infty}{R}{\nu} for all $\nu$ in a neighbourhood of $\mu$. Therefore, by Proposition~\ref{prop:0T_thermodynamic_limit}, $\overline{u}$ is a critical point of $\mathcal G(\,\cdot\,;\mu)$.
    
    We remark here that the boundedness of the sequence $(N_{e,R_j} - N_{R_j})_j$ is a necessary condition for the limit $\mu$ to be contained in the band gap. We do not state this as an assumption in Proposition~\ref{prop:0T_thermodynamic_limit_CE_GCE} because we require the stronger condition that $\mu$ is (eventually) bounded away from $\sigma(\Ham^{R_j}(\overline{u}_{R_j}))$.
\end{proof}

\appendix

\section{Symmetries of the Hamiltonian}
\label{sec:symmetries}

In addition to assumption \asTB, the Hamiltonian also satisfies the following symmetry property which is not required for this paper but is included for completeness:
\begin{enumerate}[label=(\roman*)]
	\item[] Fix a displacement $u$ and define $y \coloneqq x + u$ where $x$ is the identity configuration. For an isometry $\mathcal I \colon \mathbb R^d \to \mathbb R^d$, we let $u_{\mathcal I}$ be the displacement for which $\mathcal I\circ y = x + {u}_\mathcal{I}$. Then, for $\ell, k \in \Lambda$, there exist orthogonal $Q^\ell, Q^k \in \mathbb R^{\numorbitals\times\numorbitals}$ such that 
	\begin{equation}
		\label{eq:isometry_invariance}
		\big[\Ham(u_{\mathcal I})^{ab}_{\ell k}\big]_{a,b = 1}^{\numorbitals} 
			= Q^\ell \cdot \big[\Ham(u)^{ab}_{\ell k}\big]_{a,b = 1}^{\numorbitals} \cdot \big(Q^k\big)^T.
	\end{equation}
\end{enumerate}

\cref{eq:isometry_invariance} states that the Hamiltonian is invariant under isometries of $\mathbb R^d$ up to an orthogonal change of basis. That is \cite{SlaterKoster1954},
\begin{equation}
	\label{a:isometry_invariance} 
	\Ham(u_{\mathcal I}) = Q\cdot \Ham(u)\cdot Q^T\quad 
	\text{where} \quad Q = \mathrm{diag}\big( Q^\ell \big)_{\ell \in \Lambda}
\end{equation}
This assumption is derived in \cite[Appendix~A]{ChenOrtner16} and, for a single atomic orbital per atom, takes the more familiar form: 
$\Ham(u_{\mathcal I}) = \Ham(u)$.
This isometry invariance means that the spectrum of $\Ham(u)$ and $\Ham(u_{\mathcal I})$ agree and thus the site energies associated with $u$ are equal to that of $u_{\mathcal I}$, for example.


\section{Non-constant Number of Atomic Orbitals per Atom}
\label{app:numorbitals}

As noted in \cref{sec:tight_binding_hamiltonian}, the number of atomic orbitals per atom should depend on the atomic species. That is, for $\ell \in \Lambda$, the number of atomic orbitals corresponding to $\ell$ depends on the atomic species of $\ell$. We shall denote by $\numorbitals(\ell)$ this number of atomic orbitals. Now, the linear tight binding Hamiltonian takes the following form: 
\begin{align*}
    [\Ham(u)]_{\ell k}^{ab} = h_{\ell k}^{ab}(\bm{r}_{\ell k}(u)) \quad 
    \text{for }\ell, k \in \Lambda \textrm{, }
    1 \leq a \leq \numorbitals(\ell) \textrm{ and }
    1 \leq b \leq \numorbitals(k).
\end{align*}
We may define $\numorbitals \coloneqq \max_{\ell \in \Lambda} \numorbitals(\ell)$, and, for $\ell,k\in\Lambda$ and $1\leq a,b\leq \numorbitals$,
\begin{align*}
    [\widetilde{\Ham}(u)]_{\ell k}^{ab} \coloneqq 
    \begin{cases} 
    h_{\ell k}^{ab}(\bm{r}_{\ell k}(u)) &
    \text{if }
    1 \leq a \leq \numorbitals(\ell) \textrm{ and }
    1 \leq b \leq \numorbitals(k) \\
    0 & \textrm{if } \numorbitals(\ell)<a\leq\numorbitals \textrm{ or }\numorbitals(k)<b\leq \numorbitals.
    \end{cases}
\end{align*}
Moreover, for $\psi = \left\{\psi(\ell;a) \colon \ell \in \Lambda, a \in \{1,\dots,\numorbitals(\ell)\}\right\}$ we may define
\begin{align*}
    \widetilde{\psi}(\ell;a) \coloneqq \begin{cases}
        \psi(\ell;a) &\textrm{if } 1\leq a \leq \numorbitals(\ell) \\
        0 &\textrm{if } \numorbitals(\ell) < a \leq \numorbitals.
    \end{cases}
\end{align*}
Without loss of generality we may suppose that $\lambda \geq a_0 > 0$ for all $\lambda \in \sigma(\Ham(u))$. This can be done by artificially shifting the spectrum by adding a suitable multiple of the identity to the Hamiltonian. A detailed description of this construction is given in \cite[\S4.3]{ChenOrtnerThomas2019}.

We claim that $\Ham(u)$ may be replaced with $\widetilde{\Ham}(u)$ (which satisfies the assumptions of this paper) throughout. For example, in the site energies, for an appropriate choice of contour $\mathscr C$ that does not encircle $\{0\}$, we have
\begin{align}\label{eq:site_energies_multiple_orbitals}\begin{split}
    G_\ell^\beta(u) &= \sum_s \mathfrak{g}^\beta(\lambda_s;\mu)\sum_{a = 1}^{\numorbitals(\ell)}[\psi_s]_{\ell a}^2 
    = \sum_s \mathfrak{g}^\beta(\lambda_s;\mu)\sum_{a = 1}^{\numorbitals}[\widetilde{\psi}_s]_{\ell a}^2 \\
    &= -\frac{1}{2\pi i}\sum_a\oint_\mathscr{C} \mathfrak{g}^\beta(z;\mu)\left[\left(\widetilde{\Ham}(u) - z\right)^{-1}\right]_{\ell\ell}^{aa}\mathrm{d}z.
\end{split}
\end{align}
The calculation in \cref{eq:site_energies_multiple_orbitals} is valid because, apart from the additional zero eigenvalues that are not encircled by $\mathscr C$, the spectrum of $\Ham(u)$ is equal to that of $\widetilde{\Ham}(u)$.

\section{Potential Energy Surface at Finite Fermi Temperature}
\label{sec:pot_energy_surface}

In this section, we give a more detailed derivation of the Helmholtz free energy. Again, we have a particle system containing $N_{e,R}$ electrons and nuclei described by some admissible displacement $u \colon \Lambda_R \to \mathbb{R}^d$. 

The Helmholtz free energy is given by minimising the functional \cite{DavidMermin1965,ChenLuOrtner18}
\begin{align*}
    \mathfrak F\left( \{ \psi_s \}_{s=1}^{N_R} , \{f_s\}_{s=1}^{N_R} \right) 
        &\coloneqq 2 \sum_{s=1}^{N_R} \bigg( f_s \Braket{\psi_s | \Ham^R(u) | \psi_s} + k_B TS(f_s) \bigg),\\
    \text{where} \quad S(f) &\coloneqq f \log f + (1-f) \log(1-f)
\end{align*}
under the constraints that 
$\psi_s \colon \Lambda_R \times \{1,\,\dots,N_\textrm{b}\} \to \mathbb R$
are orthogonal orbital functions and the corresponding {occupation numbers}, $f_s$, satisfy 
$0\leq f_s \leq 1$ and $2\sum_{s=1}^{N_R} f_s = N_{e,R}$. 
Here, the factor of two accounts for the spin and the inverse Fermi-temperature is given by
$\beta \coloneqq (k_B T)^{-1}$
where $k_B$ is the Boltzmann constant. 

A simple calculation yields the existence of a minimiser 
$\{ \psi_s \}_{s=1}^{N_R}$, $\{f_s\}_{s=1}^{N_R}$
(where the dependence on $y$ is omitted) satisfying
\begin{align}
    \label{eq:Helmholtz} 
    \Ham^R(u) \psi_s = \lambda_s \psi_s \quad \text{and} \quad 
    f_s = f_\beta(\lambda_s - \ep_\mathrm{F}^{\beta,R})
\end{align}
where 
$\{\lambda_s\}_{s = 1}^{N_R}$
is some enumeration of $\sigma(\Ham^R(u))$. The {Fermi level}, 
$\ep^{\beta,R}_\textrm{F}(u) = \ep_\mathrm{F}^{\beta,R}$,
is the Lagrange multiplier for the particle number constraint. Moreover, there is a unique solution to this constraint: the {particle number} functional, $N^{\beta,R}(u,\cdot)$, is continuous and strictly increasing (in $\tau\in\mathbb R$) with 
$N^{\beta,R}(u;\tau)\to0$ as $\tau \to-\infty$ and $N(u;\tau) \to 2N_R$ as $\tau\to\infty$ 
and thus there is a unique value, $\ep_\mathrm{F}^{\beta,R}$, satisfying
\begin{align*}
    N^{\beta,R}(u;\ep_\mathrm{F}^{\beta,R}) = N_{e,R}.
\end{align*}

After writing the entropy contribution, $S(f_s)$, as
\begin{align*} 
    S(f_s) &= f_s \log f_s + (1-f_s) \log(1-f_s) \\ 
    &= f_s \log \left( e^{-\beta(\lambda_s-\ep^{\beta,R}_\mathrm{F})} (1 - f_s) \right) + (1-f_s) \log(1-f_s)\\
    &= -\beta(\lambda_s - \ep_\mathrm{F}^{\beta,R})f_s + \log(1-f_s),
\end{align*}
we may rewrite the Helmholtz free energy as 
\begin{align*}
    \begin{split}
        E^{\beta,R}(u) &= \sum_{s=1}^{N_R} \mathfrak{e}^\beta(\lambda_s,\ep^{\beta,R}_\mathrm{F}) \quad \text{where}\\
        \mathfrak{e}^\beta(\lambda;\tau) 
            &\coloneqq 2 \lambda f_\beta(\lambda-\tau) + \frac{2}{\beta} S(f_\beta(\lambda-\tau)) \\
        &= 2 \tau f_\beta(\lambda-\tau) + \frac{2}{\beta} \log\left(1-f_\beta(\lambda-\tau)\right).
    \end{split}
\end{align*}

\section{Convergence of the Fermi Level}
\label{sec:appendix_FL}

Here we prove the convergence of the Fermi level in the zero temperature limit for finite computational domains. This proof is a much simpler version of the one presented in \cite{ThomasMSc2018}.

\begin{proof}[Proof of Lemma~\ref{lem:fermilevelconverge}]

We first suppose that there exists $\ep \in \sigma(\Ham^R(u))$ such that $N^{\infty,R}(\ep) = N_{e,R}$. This means that $\underline{\ep} = \overline{\ep} = \ep$. Suppose that $\ep_\mathrm{F}^\beta \to \ep^\prime$ along a subsequence (which we do not relabel) as $\beta \to \infty$. Using the fact that $N^{\beta,R}(\cdot)$ is strictly increasing and taking $\beta \to \infty$, we obtain $N^{\infty,R}(\ep^\prime - \delta) \leq N_{e,R} \leq N^{\infty,R}(\ep^\prime + \delta)$ for all $\delta > 0$. That is, $\ep^\prime = \ep$ and we can conclude.

Now we instead suppose that $N^{\infty,R}(\ep) \not= N_{e,R}$ for all $\ep \in \sigma(\Ham(u))$. In this case, $\underline{\ep} < \overline{\ep}$ are adjacent eigenvalues. Since $N^{\beta,R}(\ep_\mathrm{F}^\beta) = N_{e,R}$, it is not hard to show that, for sufficiently large $\beta$, $\ep_\mathrm{F}^\beta \in (\underline{\ep},\overline{\ep})$. We now simply use the fact that $f_{\beta}(-\tau) = 1 - f_\beta(\tau)$ to obtain
%
%
%
\begin{align}\label{eq:FL_converge_1}\begin{split}
 \sum_{s\colon \lambda_s \geq \overline{\ep}} f_\beta(|\lambda_s - \ep_\mathrm{F}^\beta|) - \sum_{s\colon \lambda_s \leq \underline{\ep}} f_\beta(|\lambda_s - \ep_\mathrm{F}^\beta|) = \tfrac{1}{2}\left( N_{e,R} - N^{\infty,R}(\tfrac{1}{2}(\underline{\ep} + \overline{\ep}))\right).
\end{split}\end{align}
We suppose that $n \coloneqq \tfrac{1}{2}\left( N_{e,R} - N^{\infty,R}(\tfrac{1}{2}(\underline{\ep} + \overline{\ep}))\right) \geq 0$ (we will see that the other case is very similar). Now, using \cref{eq:FL_converge_1}, we have
\begin{align}\label{eq:FL_converge_2}\begin{split}
  f_\beta(|\overline{\ep} - \ep_\mathrm{F}^\beta|) &\leq \sum_{s\colon \lambda_s \geq \overline{\ep}} f_\beta(|\lambda_s - \ep_\mathrm{F}^\beta|) 
  = n + \sum_{s\colon \lambda_s \leq \underline{\ep}} f_\beta(|\lambda_s - \ep_\mathrm{F}^\beta|) 
  \leq n + m f_\beta(|\underline{\ep} - \ep_\mathrm{F}^\beta|)
\end{split}\end{align}
where $m \coloneqq \#\{{s\colon \lambda_s \leq \underline{\ep}}\}$. In the exact same way, we have
\begin{align}\label{eq:FL_converge_3}\begin{split}
  n + f_\beta(|\underline{\ep} - \ep_\mathrm{F}^\beta|) 
  \leq (N_R - m) f_\beta(|\overline{\ep} - \ep_\mathrm{F}^\beta|).
\end{split}\end{align}
Therefore, by combining \cref{eq:FL_converge_2} and \cref{eq:FL_converge_3}, we have
\begin{align}\label{eq:FL_4}
   (N_R - m)^{-1}\left( n+ f_\beta(|\underline{\ep} - \ep_\mathrm{F}^\beta|) \right)
   \leq  
   f_\beta(|\overline{\ep} - \ep_\mathrm{F}^\beta|) \leq n + m f_\beta(|\underline{\ep} - \ep_\mathrm{F}^\beta|).
\end{align}

Now we consider the case that $n = 0$ (that is, $N_{e,R} = N^{\infty,R}(\tfrac{1}{2}(\underline{\ep}+\overline{\ep}))$). In this case, we may simplify \cref{eq:FL_4} to obtain the following bound 
\begin{align*}
  c \leq  e^{\beta(|\underline{\ep} - \ep_\mathrm{F}^\beta| - |\overline{\ep} - \ep_\mathrm{F}^\beta|)} \leq C
\end{align*}
for some $c,C > 0$. This immediately gives $\left||\underline{\ep} - \ep_\mathrm{F}^\beta| - |\overline{\ep} - \ep_\mathrm{F}^\beta|\right| \lesssim \beta^{-1}$. That is, $\ep_\mathrm{F}^\beta \to \tfrac{1}{2}(\underline{\ep} + \overline{\ep})$. 

On the other hand, we now suppose that $n > 0$ (that is, $N^{\infty,R}(\tfrac{1}{2}(\underline{\ep}+\overline{\ep})) < N_{e,R}$). We now replace $|\overline{\ep} - \ep_\mathrm{F}^\beta|$ with $|\overline{\ep} -\underline{\ep}| - |\underline{\ep} - \ep_\mathrm{F}^\beta|$ in \cref{eq:FL_4} to obtain the following bound
\begin{align}
    c \leq e^{\beta(|\underline{\ep} - \ep_\mathrm{F}^\beta| - |\overline{\ep} - \underline{\ep}|)} \leq C
\end{align}
for some $c,C>0$. This means that $\left||\underline{\ep} - \ep_\mathrm{F}^\beta| - |\overline{\ep} - \underline{\ep}|\right| \lesssim \beta^{-1}$ and thus $\ep_\mathrm{F}^\beta \to \overline{\ep}$.
\end{proof}

\section{Band Structure of the Reference Hamiltonian}
\label{sec:band_structure}

Recall, the unit cell $\Gamma \subset \Lambda^\mathrm{ref}$ is finite and satisfies $\Lambda^\mathrm{ref} = \bigcup_{\gamma\in\mathbb Z^d}\left(  \Gamma + \mathsf{A}\gamma \right)$ and $\Gamma + \mathsf{A}\gamma$ are pairwise disjoint for all $\gamma \in \mathbb Z^d$. 

For $\xi \in \mathbb R^d$ and $\psi \in \ell^2(\Lambda^\mathrm{ref}\times \{1,\dots,\numorbitals\})$, we define $(U\psi)_\xi \in \ell^2(\Lambda^\mathrm{ref}\times\{1,\dots,\numorbitals\})$ by
\begin{align}
\label{eq:U}
    (U\psi)_\xi(\ell;a) = \sum_{\gamma\in\mathbb Z^d} \psi(\ell + \mathsf{A}\gamma;a) e^{-i (\ell + \mathsf{A}\gamma) \cdot \xi}.
\end{align}
We let $\Gamma^\star \subset \mathbb R^d$ be a bounded connected domain containing the origin such that $\mathbb R^d$ is the disjoint union of $\Gamma^\star + 2\pi\mathsf{A}^{-\mathrm{T}}\eta$ for $\eta \in \mathbb Z^d$. In particular, for $\xi \in \mathbb R^d$, there exist unique $\eta \in \mathbb Z^d$ and $\xi_0 \in \Gamma^\star$ such that $\xi = \xi_0 + 2\pi \mathsf{A}^{-\mathrm{T}}\eta$ and so 
$e^{-i\mathsf{A}\gamma \cdot \xi} = e^{-i\mathsf{A}\gamma \cdot \xi_0} e^{-i\mathsf{A}\gamma \cdot 2\pi\mathsf{A}^{-\mathrm{T}}\eta} =e^{-i\mathsf{A}\gamma \cdot \xi_0}$ for any $\gamma \in \mathbb Z^d$.
We therefore restrict $\xi$ to $\Gamma^\star$ and define $U \colon \ell^2(\Lambda^\mathrm{ref}\times\{1,\dots,\numorbitals\}) \to L^2(\Gamma^\star; \ell^2(\Gamma\times\{1,\dots,\numorbitals\}))$ by \cref{eq:U}. It is not hard to see that this operator is unitary where $L^2(\Gamma^\star; \ell^2(\Gamma\times\{1,\dots,\numorbitals\}))$ is a Hilbert space with norm corresponding to the following inner product:
\begin{align*}
    \Braket{\Psi,\Phi}_{L^2(\Gamma^\star;\ell^2)} &=
    \frac{1}{|\Gamma^\star|} \int_{\Gamma^\star} \Braket{\Psi_\xi,\Phi_\xi}_{\ell^2}\mathrm{d}\xi = 
    \frac{1}{|\Gamma^\star|}\sum_{\ell \in \Gamma}\sum_{1\leq a \leq \numorbitals} \int_{\Gamma^\star}\Psi_\xi(\ell;a) \overline{\Phi_\xi(\ell;a)} \mathrm{d}\xi.
\end{align*}
Therefore, we can write 
\begin{align*}
    &(U\Ham^\mathrm{ref}\psi)_\xi(\ell;a) =
    \sum_{\gamma\in\mathbb Z^d} \sum_{\above{1\leq b\leq \numorbitals}{k \in \Lambda^\mathrm{ref}}}
    \big[\Ham^\mathrm{ref}\big]_{\ell+\mathsf{A}\gamma,k}^{ab} \psi(k;b)e^{-i(\ell + \mathsf{A}\gamma)\cdot\xi} \\
    &\quad= \sum_{\above{1\leq b\leq \numorbitals}{k \in \Gamma}}\sum_{\eta\in\mathbb Z^d} \left(\sum_{\gamma\in\mathbb Z^d}
    h^{ab}_{\ell + \mathsf{A}\gamma, k + \mathsf{A}\eta}(\ell - k + \mathsf{A}(\gamma - \eta))e^{-i(\ell - k + \mathsf{A}(\gamma - \eta)\cdot\xi} \right) \psi(k+\mathsf{A}\eta;b)e^{-i(k + \mathsf{A}\eta)\cdot\xi} \\
   &\quad= \sum_{\above{1\leq b\leq \numorbitals}{k \in \Gamma}} \left(\sum_{\gamma\in\mathbb Z^d}
    h^{ab}_{\ell k}(\ell - k + \mathsf{A}\gamma)e^{-i(\ell - k +  \mathsf{A}\gamma)\cdot\xi} \right) (U\psi)_\xi(k;b)
\end{align*}
and can conclude that,
\begin{align}
\label{eq:Ham_bloch_1}
    ( U \Ham^\mathrm{ref} \psi)_\xi = \Ham^\mathrm{ref}_\xi (U\psi)_\xi
\end{align}
where $\Ham^\mathrm{ref}_\xi \colon \ell^2(\Gamma\times\{1,\dots,\numorbitals\}) \to \ell^2(\Gamma\times\{1,\dots,\numorbitals\})$ is the matrix with entries
\begin{align}
\label{eq:Ham_bloch_2}
    \big[\Ham^\mathrm{ref}_\xi\big]_{\ell k}^{ab} = \sum_{\gamma\in\mathbb Z^d} h_{\ell k}^{ab}(\ell - k + \mathsf{A}\gamma) e^{-i (\ell - k + \mathsf{A}\gamma) \cdot \xi}.
\end{align}
Therefore, by \cref{eq:Ham_bloch_1}, \cref{eq:Ham_bloch_2} and the fact that $U$ is unitary, we have that $\sigma(\Ham^\mathrm{ref})$ is composed of a finite number of spectral bands: let $\lambda_n(\xi)$ be the eigenvalues of $\Ham_\xi^\mathrm{ref}$ for $\xi \in \Gamma^\star$ and $n = 1,\dots,\numorbitals\cdot\#\Gamma$, then 
\begin{align*}
    \sigma(\Ham^\mathrm{ref}) = \bigcup_{n}\bigcup_{\xi \in \Gamma^\star} \lambda_n(\xi).
\end{align*}
It is important to note that $\lambda_n\colon \Gamma^\star \to \mathbb R$ are continuous.

A similar calculation can be carried out for $R < \infty$ to conclude that 
\begin{align*}
    \sigma(\Ham^{\mathrm{ref},R}) = \bigcup_{n}\bigcup_{\xi \in \Gamma^\star_R} \lambda_n(\xi).
\end{align*}
where $\Gamma_R^\star \subset \Gamma^\star$ is a discrete set depending on the choice of $\mathsf{B}_R$ in the sequence of finite domain approximations. 
Importantly, this means that $\sigma(\Ham^{\mathrm{ref},R}) \subset \sigma(\Ham^\mathrm{ref})$ for all $R>0$.

\section{Proof of Lemma~\ref{lem:decomp_Ham_R}: Decomposition of the Hamiltonian for \texorpdfstring{$R<\infty$}{}}
\label{app:decomp_Ham}
Supposing $\|Du_R\|_{\ell^2_\ctGamma} \leq C$ for all $R$ and taking $\delta> 0$, we can find $\ell_1^R,\dots,\ell^R_{n_R} \in \Lambda_R$ such that
\begin{align}\label{eq:delta_bound}
    \sum_{\ell \in \Lambda_R \setminus \{\ell_1^R,\dots,\ell_{n_R}^R\}} |Du_R(\ell)|_\ctGamma^2 \leq \delta^2
\end{align}
and $n_R$ is uniformly bounded above by some $n$. For simplicity of notation, we relabel so that $|\ell^R_1| \leq |\ell^R_2| \leq \dots \leq |\ell^R_{n_R}|$, and define $\ell^R_{j} \coloneqq \ell^R_{n_R}$ for all $n_R < j \leq n$. Now, there exist $n_0$ and $R_0$ such that $(\ell^R_j)_R \subset B_{R_0}$ for all $1\leq j < n_0$ and $(\ell^R_j)_R \to \infty$ for all $n_0 \leq j \leq n$. 

We define $P^R_\mathrm{loc}(u_R)$ to be the finite rank operator given by removing all the ``bounded defect states'': for $R_1> 0$, we define
\begin{align*}
   P^{R}_\mathrm{loc}(u_R)_{\ell k}^{ab} \coloneqq
   \begin{cases}
   \left[ \widetilde{\Ham}^R(u_R) - \widetilde{\Ham}^{\mathrm{ref},R} \right]_{\ell k}^{ab}  & \text{if }
   %
   \big(\ell \in \Lambda_R \cap B_{R_0} \textrm{ or } k \in \Lambda_R \cap B_{R_0} \big) 
   \textrm{ and } |\ell - k| \leq R_1 
   \\
   0    &\text{otherwise.}
   \end{cases}
\end{align*}
%
We can see that $P^{R}_\mathrm{loc}(u_R)_{\ell k}^{ab}$ is only non-zero for $(\ell, k) \in (\Lambda \cap B_{R_\delta})^2$ for some $R_\delta$ depending on $R_0$ and $R_1$ but independent of $R$. By Lemma~\ref{ham_converge_R}, we can see directly from the definition of $P^{R}_\mathrm{loc}(u_R)$ that, if $u_R \rightharpoonup u$, then $P^{R}_\mathrm{loc}(u_R)\to P^{\infty}_\mathrm{loc}(u)$ where 
$P^{\infty}_\mathrm{loc}(u)_{\ell k}^{ab} \coloneqq  \widetilde{\Ham}(u)_{\ell k}^{ab} - [\widetilde{\Ham}^{\mathrm{ref}}]_{\ell k}^{ab}$ for $|\ell - k| \leq R_1$ and either $\ell \in \Lambda_R \cap B_{R_0}$ or $k \in \Lambda_R \cap B_{R_0}$.

On the other hand, we define $P^R_\infty(u_R)$ be the finite rank operator given by removing all the ``defect states that get sent to infinity'': for $R_1> 0$, we define
\begin{align*}
   P^{R}_\infty(u_R)_{\ell k}^{ab} \coloneqq
   \begin{cases}
   \left[ \widetilde{\Ham}^R(u_R) - \widetilde{\Ham}^{\mathrm{ref},R} \right]_{\ell k}^{ab}  & \text{if }
   \big(\ell \in \{\ell^R_j\}_{j = n_0}^n \textrm{ or } k \in \{\ell^R_j\}_{j = n_0}^n \big) \textrm{ and } r^{\#}_{\ell k}(x) \leq R_1 
   \\
   0    &\text{otherwise.}
   \end{cases}
\end{align*}

We define the perturbation 
$P_\delta^R(u_R) \coloneqq \widetilde{\Ham}^R(u_R) - \widetilde{\Ham}^{\mathrm{ref},R} - P^{R}_\mathrm{loc}(u_R) - P^{R}_\infty(u_R)$
and now show that this can be bounded appropriately in the Frobenius norm.
If $\delta > 0$ sufficiently small and $\ell \not \in \{ \ell^R_j \}_{j=1}^n$, we can apply \cref{eq:delta_bound} to conclude that $|D_{k-\ell} u_R(\ell)| \leq \mathfrak{m}|\ell - k|$. Therefore, we may apply Lemma~\ref{ham_converge_R} in the following to obtain: for fixed $1\leq a,b \leq \numorbitals$, 
\begin{align}\label{eq:Pdelta_1}\begin{split}
   \sum_{\ell, k \in \Lambda_R \setminus \{\ell^R_j\}_j}
   |P_\delta^{R}(u_R)_{\ell k}^{ab}|^2 
    &\leq \sum_{\ell, k \in \Lambda_R \setminus \{\ell^R_j\}_j}
    \left|\left[\Ham^R(u_R) - \Ham^R(x)\right]_{\ell k}^{ab}\right|^2 \\
    &\leq C \sum_{\ell, k \in \Lambda_R \setminus \{\ell^R_j\}_j}  e^{-2c\gamma_0 r^\#_{\ell k}(x)} |D_{k - \ell} u_R(\ell)|^2 \leq C\delta^2
\end{split}\end{align}
where $c\coloneqq \frac{\mathfrak{m}\sqrt{3}}{4}$ is the constant from Lemma~\ref{ham_converge_R}. In the final line, we use the same idea as in \cref{eq:Ham_converge_trick} to replace the torus distance, $r_{\ell k}^\#(x)$, with $|\ell - k|$. Finally, using the off-diagonal decay of the Hamiltonian operators, we obtain: for fixed $1\leq a,b \leq \numorbitals$,
\begin{align}\label{eq:Pdelta_2}
   \sum_{\ell \in \{\ell^R_j\}_j} 
   \sum_{
   \above
        {k \in \Lambda_R \setminus \{\ell^R_j\}_j}
        {r_{\ell k}^\#(x) > R_1}
   }
   |P_\delta^{R}(u_R)_{\ell k}^{ab}|^2 
    &\leq 2h_0\sum_{\ell \in \{\ell^R_j\}_j} 
   \sum_{
   \above
        {k \in \Lambda_R \setminus \{\ell^R_j\}_j}
        {r_{\ell k}^\#(x) > R_1}
   }
    e^{-2\gamma_0\mathfrak{m}|\ell - k|} 
    \leq C e^{-\gamma_0\mathfrak{m}R_1}.
\end{align}
Therefore, after combining by \cref{eq:Pdelta_1} and \cref{eq:Pdelta_2}, and choosing $R_1$ sufficiently large, we obtain $\|P^R_\delta(u_R)\|_\mathrm{F} \leq C\delta$.

\section{Zero Temperature Limit of \texorpdfstring{$\mathfrak{g}^{\beta}(z;\tau)$}{}}
Here we prove the convergence of the functions $\mathfrak{g}^\beta(z;\mu)$ as $\beta \to \infty$.  

First, we define the analytic continuation of $\mathfrak{g}^\beta$ as in \cite{ChenOrtnerThomas2019}. We let $z \mapsto \log z$ be the principal complex logarithm (i.e.\ with $\Im{\log z} \in (-\pi, \pi]$). For $z \in \mathbb C \setminus \left\{\mu + ir \colon r \in \mathbb R, |r| \geq \pi\beta^{-1}\right\}$ with $\beta \Im{z} \in ((2k-1)\pi, (2k+1)\pi]$ for some $k \in \mathbb Z$, we define
\begin{align*}
	\mathfrak{g}^\beta(z;\mu) \coloneqq 
	\begin{cases}
	\frac{2}{\beta}\log(1-f_\beta(z-\mu)) & \text{if } \Re{z}\geq\mu \\
	\frac{2}{\beta}\left[\log(1-f_\beta(z-\mu)) + 2k \pi i\right] &\text{if } \Re{z} \leq\mu.
	\end{cases}
\end{align*}
It is shown in \cite{ChenOrtnerThomas2019} that this mapping defines an analytic function on the domain of definition. 

\label{proof_gbeta_convergence}
\begin{proof}[Proof of Lemma~\ref{lem:g_convergence_2}]
	\textit{Step 0: preliminaries.} To simplify notation and without loss of generality, we may suppose that $\mu = 0$. For $z \in \mathbb C$, we have
	\begin{align}\label{eq:1-f}
		1 - f_\beta(z) = \frac{e^{\beta z}}{1 + e^{\beta z}} = \frac{e^{\beta z} \left( 1+e^{\beta\overline{z}} \right)}{|1+e^{\beta z}|^{2}} = \frac{e^{\beta\Re{z}}}{|1+e^{\beta z}|^2}\left( e^{i\beta \Im{z}} + e^{\beta \text{Re}(z)} \right).
	\end{align}
	Therefore, there exists some real $\alpha(z)>0$ such that
	\begin{gather*}
		\Re{1 - f_\beta(z)} = \alpha(z) \left( \cos(\beta\Im{z}) + e^{\beta\Re{z}} \right) \\ \Im{1 - f_\beta(z)} = \alpha(z) \sin(\beta \Im{z}).
	\end{gather*}
	This means that, if $\cos(\beta\Im{z}) + e^{\beta\Re{z}} \not= 0$, we have
	\begin{align}\label{eq:arg}
		\Im{\log\left[1-f_\beta(z)\right]} = \arg\left(1-f_\beta(z)\right) = \tan^{-1}\left(\frac{\sin(\beta \Im{z})}{\cos(\beta \Im{z}) + e^{\beta \Re{z}}}\right).
	\end{align}
	We also note that the following inequality holds: for all $z \in \mathbb C$,
	\begin{align}\label{eq:1pluse}
		\left(1 - e^{\beta\Re{z}}\right) ^2\leq|1 +e^{\beta z}|^2 \leq \left(1 + e^{\beta \Re{z}}\right)^2.
	\end{align}
	We shall estimate both the real and imaginary parts for $z \in \mathscr C^-$ and $z \in \mathscr C^+$ separately.
	
	\textit{Step 1: $z \in \mathscr C^+$, real part.} We use the fact that $1 - f_\beta(z) = (1 + e^{-\beta z})^{-1}$, $\beta\mathrm{Re}(z) > 0$ and by using \cref{eq:1pluse}, we have $|1 +e^{-\beta z}|^{-1} \leq (1 - e^{-\beta\mathrm{Re}(z)})^{-1}$ and so
	\begin{align*}
		\Re{\mathfrak{g}^\beta(z;\mu) - \mathfrak{g}(z;\mu)}	
			&= 2\beta^{-1}\log\left|1 - f_\beta(z)\right| 
			\leq 2\beta^{-1} \log\left((1 - e^{-\beta \Re{z}})^{-1} \right) \\
			&\leq 2\beta^{-1} \left((1 - e^{-\beta \Re{z}})^{-1} - 1 \right)
			= 2\beta^{-1}  
			(1 - e^{-\beta\mathrm{Re}(z)})^{-1}
			e^{-\beta\mathrm{Re}(z)}.
	\end{align*}
	In the final line we use the fact that $\log(x) \leq x - 1$. Since $\Re{z} \geq \tfrac{1}{2}\ctDist$, we have 
	\begin{align}\label{eq:conclusion_1}
		\big|\mathrm{Re}\big(\mathfrak{g}^\beta(z;\mu) - \mathfrak{g}(z;\mu)\big)\big| \leq 2C_{\beta\ctDist}^{(1)} \beta^{-1}e^{-\beta|\Re{z}|} \quad \text{where} \quad C_{\beta\ctDist}^{(1)} \coloneqq (1-e^{-\frac{1}{2}\beta\ctDist})^{-1}.
	\end{align}
	%
	
	\textit{Step 2: $z \in \mathscr C^+$, imaginary part.} Again, since $\Re{z}$ is positive and bounded below by $\tfrac{1}{2}\ctDist$, we have
	\begin{align*}
		\cos(\beta\Im{z}) + e^{\beta\Re{z}} 
		\geq e^{\beta\mathrm{Re}(z)} - 1 
		= (1 - e^{-\beta\mathrm{Re}(z)})e^{\beta\mathrm{Re}(z)}
		\geq \big(C_{\beta\ctDist}^{(1)}\big)^{-1} e^{\beta\Re{z}}.
	\end{align*}
	Therefore, applying \cref{eq:arg}, together with $|\tan^{-1}(x)| = \tan^{-1}|x|$ and the fact that $\tan^{-1}$ is increasing, we have
	\begin{align}\label{eq:conclusion_2}
		\left|\Im{\mathfrak{g}^\beta(z;\mu) - \mathfrak{g}(z;\mu)}\right| &\leq
		 2\beta^{-1}\tan^{-1}\left(C_{\beta\ctDist}^{(1)} e^{-\beta \Re{z}}\right) \leq 2C_{\beta\ctDist}^{(1)} \beta^{-1} e^{-\beta |\Re{z}|}.
	\end{align}
	%
	
	\textit{Step 3: $z \in \mathscr C^-$, real part.} We can rewrite the real component as follows
	\begin{align}
		\label{eq:realpart_}
		\begin{split}
			\Re{ \mathfrak{g}^\beta(z;\mu) - \mathfrak{g}(z;\mu)} 
			&= 2\beta^{-1}\Big(\log\big|e^{\beta\Re{z}}\big| 
							- \log\big|1+e^{\beta z}\big|\Big) - 2\Re{z} \\
			&= - 2 \beta^{-1} \log\big|1+e^{\beta z}\big|.
		\end{split}
	\end{align}
	After applying \cref{eq:1pluse}, using $1 - \frac{1}{x} \leq \log(x) \leq x - 1$ and noting that $\Re{z} < 0$ and is bounded away from zero, we have 
	\begin{align}\label{eq:logabs}
	\begin{split}
		\log|1+e^{\beta z}| &\leq \log(1+e^{\beta \Re{z}}) \leq e^{-\beta |\Re{z}|} \quad \text{and} \\
		\log|1+e^{\beta z}| &\geq \log(1-e^{-\beta|\Re{z}|})
			\geq -\frac{e^{-\beta|\Re{z}|}}{1-e^{-\beta|\Re{z}|}}
			\geq - C_{\beta\ctDist}^{(1)}e^{-\beta|\mathrm{Re}(z)|}. 		
	\end{split}
	\end{align}
	Therefore, combining \cref{eq:realpart_} with \cref{eq:logabs}, we can conclude that
	\begin{align}\label{eq:conclusion_3}
	\left|\Re{ \mathfrak{g}^\beta(z;\mu) - \mathfrak{g}(z;\mu)}\right| \leq 2 C_{\beta\ctDist}^{(1)} \beta^{-1} e^{-\beta|\Re{z}|}.
	\end{align}
	
	\textit{Step 4: $z \in \mathscr C^-$, imaginary part.} Suppose that $\omega \in \mathbb C$ with $\beta \Im{\omega} \in (-\pi,\pi]$ and $\beta (z - \omega) = 2k \pi i$ for some $k \in \mathbb Z$. Now, since $e^{\beta z} = e^{\beta\omega}$, we have
	\begin{align*}
		\mathfrak{g}^\beta(z;\mu) - \mathfrak{g}(z;\mu) 
		&= 2\beta^{-1} \big( \log(1-f_\beta(z - \mu)) + 2k \pi i - \beta z \big) \\
		&= 2\beta^{-1} \big( \log(1-f_\beta(\omega - \mu)) - \beta \omega  \big) \\
		&= \mathfrak{g}^\beta(\omega;\mu) - \mathfrak{g}(\omega;\mu).
	\end{align*}
	Thus, it is sufficient to consider $z$ with $\beta \Im{z} \in (-\pi,\pi]$ in the following.
	
	\textit{Step 4: Case 1. }Let us fix $\delta > 0$ such that $e^{\beta\Re{z}} < \delta < 1$ and consider the case that 
	\begin{align}
	\label{eq:etaassumption}
	\left| \cos(\beta \Im{z}) \right| \leq \delta.
	\end{align}
 	Without loss of generality, we may suppose that $\beta \Im{z}>0$ (the other case can be treated in the exact same way). Since, $\max\{1 - \frac{2}{\pi}x, \frac{2}{\pi} x - 1 \}\leq |\cos(x)|$ on $[0,\pi]$, we can apply \cref{eq:etaassumption}, to obtain
 	\begin{equation}
 		\frac{\pi}{2}(1 - \delta) \leq \beta \Im{z} \leq \frac{\pi}{2}(1 + \delta).
 	\end{equation}
 	After noting that $\frac{\pi}{2} = \tan^{-1}(x) + \frac{1}{\tan^{-1}(x)}$ and $|\tan^{-1}(x)| \leq |x|$, we apply \cref{eq:arg} to obtain
	\begin{align}\label{eq:imaginary_1}\begin{split}
		\left| \arg(1-f_\beta(z)) -\tfrac{\pi}{2}\right| 
		&= \left| \tan^{-1}\left(\frac{\sin(\beta\Im{z})}{\cos(\beta\Im{z}) + e^{\beta\Re{z}}}\right) - \frac{\pi}{2} \right| \\
		&= \left| \tan^{-1}\left(\frac{\cos(\beta\Im{z}) + e^{\beta\Re{z}}}{\sin(\beta\Im{z})}\right)\right| \\
		&\leq \left|\frac{\cos(\beta\Im{z}) + e^{\beta\Re{z}}}{\sin(\beta\Im{z})}\right| 
		\leq \frac{\delta + e^{-\beta|\Re{z}|}}{\sin(\tfrac{\pi}{2}(1 + \delta))} 
		\leq \frac{2\delta}{\sin(\tfrac{\pi}{2}(1+\delta))}
	\end{split}\end{align}
	In the exact same way, if $|\cos(\beta\Im{z}) + e^{-\beta|\Re{z}|}| \leq \delta$, then we obtain
	\begin{align}\label{eq:imaginary_2}\begin{split}
			|\beta\Im{z} - \tfrac{\pi}{2}| &\leq \tfrac{\pi}{2}( \delta + e^{-\beta|\Re{z}|}) \leq \pi \delta \quad \text{and} \\
		\left| \arg(1-f_\beta(z)) -\tfrac{\pi}{2}\right| &\leq \frac{\delta}{\sin(\tfrac{\pi}{2}(1 + \delta + e^{-\beta|\Re{z}|}))} \leq  \frac{\delta}{\sin(\tfrac{\pi}{2}(1 + 2\delta))}
		\end{split}\end{align}

	\textit{Step 4: Case 2. }On the other hand, suppose that 
	\[
	\left| \cos(\beta \Im{z}) \right| \geq \delta \quad \text{and} \quad \left| \cos(\beta \Im{z}) + e^{-\beta|\Re{z}|}\right| \geq \delta. 
	\]
	In particular, $\beta\Im{z} \not \in\{\pm\tfrac{\pi}{2}\}$ and $\cos(\beta\Im{z}) + e^{-\beta|\Re{z}|} \not =0$. Therefore, after noting that 
	$\frac{\mathrm{d}}{\mathrm{d}x}\tan^{-1}(x) = \frac{1}{1 + x^2} \leq 1$,
	we can conclude that
	\begin{align}
	\label{eq:argminusimbeta}
	\begin{split}
	&|\arg(1-f_\beta(z)) - \beta\Im{z}| \\
	&\quad=
	\left|\tan^{-1}\left(\frac{\sin(\beta\Im{z})}{\cos(\beta\Im{z}) + e^{-\beta|\Re{z}|}}\right) - \tan^{-1}(\tan\left({\beta \Im{z}}\right))\right| \\ 
	&\quad\leq \left|  \frac{\sin(\beta\Im{z})}{\cos(\beta\Im{z}) + e^{-\beta|\Re{z}|}} - \tan(\beta\Im{z}) \right| \\
	&\quad= \left|  \frac{e^{-\beta|\Re{z}|} \tan(\beta\Im{z})}{\cos(\beta\Im{z}) + e^{-\beta|\Re{z}|}} \right| \leq \delta^{-2} e^{-\beta |\Re{z}|}.
	\end{split}
	\end{align}
	Thus, after choosing $\delta \coloneqq e^{-\frac{1}{3}\beta|\Re{z}|}$, we can conclude that 
	\begin{align}\label{eq:conclusion_4}
	\begin{split}
		&\left|\Im{\mathfrak{g}^\beta(z; \mu) - \mathfrak{g}(z;\mu)}\right| \leq 2 C_{\beta\ctDist}^{(2)} \beta^{-1} e^{-\frac{1}{3} \beta |\Re{z}| }	\quad \text{where} \\
		&\qquad\qquad C_{\beta\ctDist}^{(2)} \coloneqq \pi + 2 \sin\big(\tfrac{\pi}{2}(1 + 2e^{-\frac{1}{6}\ctDist\beta})\big)^{-1}  
	\end{split}
	\end{align}
	Since $C^{(1)}_{\beta\ctDist}$ and $C^{(2)}_{\beta\ctDist}$ are decreasing in $\beta$, we may combine \cref{eq:conclusion_1}, \cref{eq:conclusion_2}, \cref{eq:conclusion_3} and \cref{eq:conclusion_4} to conclude: for all $\beta_0>0$
	\begin{equation}
		|\mathfrak{g}^\beta(z;\mu) - \mathfrak{g}(z;\mu)| \leq C_{\beta_0\ctDist} \beta^{-1} e^{-\frac{1}{3}\beta|\Re{z}|} \quad \text{where} \quad C_{\beta_0\ctDist} \coloneqq 2\max\left\{ C^{(1)}_{\beta_0\ctDist}, C^{(2)}_{\beta_0\ctDist} \right\}
	\end{equation}
	for all $z \in \mathscr C^- \cup \mathscr C^+$ and $\beta \geq \beta_0$.
\end{proof} 

\section{Thermodynamic Limit of the Site Energies and Derivatives}
\label{app:thermodynamic_limit_site_energies}

\begin{proof}[Proof of Lemma~\ref{eq:thermodynamic_site_energies}]
This proof follows the ideas of \cite[Proof of Theorem~10]{ChenOrtner16}.
	
We first note that for $\ell, k \in \Lambda_R$, we have
\begin{align*}
    &[\Ham^R(u) - \Ham(u)]^{ab}_{\ell k}
    = \sum_{\above{\alpha \in \mathbb Z^d}{\alpha \not = 0}} h^{ab}_{\ell k}(\bm{r}_{\ell k}(u) + \mathsf{M}_R\alpha)\\
    &\qquad\leq \ctTBprefactor \sum_{{\alpha \not= 0}} e^{-\ctTBexponent|\bm{r}_{\ell k}(u) + \mathsf{M}_R\alpha|}
    \leq C \exp\left(-\frac{1}{2}\gamma_0 \min_{\alpha \not = 0}\left|\bm{r}_{\ell k}(u) + \mathsf{M}_R\alpha\right| \right). 
\end{align*}
%

Similarly, for $1 \leq j \leq \ctHamregularity$ and $\bm{m} = (m_1,\dots,m_j) \in \Lambda^j$, we have
\begin{align*}
    &\left|\frac{\partial^j [\Ham^R(u)-\Ham(u)]^{ab}_{\ell k} }{\partial[u(m_1)]_{i_1}\dots\partial[u(m_j)]_{i_j}}\right| \\
    &\qquad \leq C e^{-\frac{1}{2}\gamma_0\sum_{l=1}^j
    \left(\min_{\alpha \not=0}|\bm{r}_{\ell m_l}(u) + \mathsf{M}_R\alpha| 
    + \min_{\alpha \not=0}|\bm{r}_{k m_l}(u) + \mathsf{M}_R\alpha|\right)}.
\end{align*}

We define the extension of $\Ham^R(u)$ to $\Lambda \times \Lambda$ by
\begin{align*}
	\left[\widetilde{\Ham}^R(u)\right]_{\ell k} \coloneqq
	\begin{cases}
	\left[{\Ham}^R(u)\right]_{\ell k} &\text{if } \ell, k \in \Lambda_R\\ 
	0 &\text{otherwise}
	\end{cases}
\end{align*}
To simplify notation, we write $\widetilde{\mathscr R}_z^R \coloneqq (\widetilde{\Ham}^R(u) - z)^{-1}$. Noting that $|[\widetilde{\Ham}^R(u)]^{ab}_{\ell k}| \leq h_0 e^{-\ctTBexponent r_{\ell k}^\#(u)}$, we have $|[\widetilde{\mathscr R}_z^R]^{ab}_{\ell k}| \leq C e^{-\ctCT r_{\ell k}^\#(u)}$ as in Remark~\ref{rem:CT_R}.

We may fix a simple closed contour $\mathscr C$ as in \cref{eq:distance_contour_0T} and write:
\begin{equation}\label{eq:RtoinftyResolvents}
\begin{split}
    &\mathcal G^\beta_{\ell}(Du(\ell)) - \mathcal G^{\beta,R}_{\ell}(Du(\ell)) 
    = -\frac{1}{2\pi i}\sum_a\oint_\mathscr{C} \mathfrak{g}^\beta(z;\mu) 
    \left[\mathscr R_z(u) - \widetilde{\mathscr R}_z^R(u)\right]_{\ell\ell}^{aa} \dxx{z} \\
	&\qquad\leq C \sum_{bc}\sum_{\ell_1, \ell_2 \in \Lambda}
	\Big[\mathscr R_z(u)\Big]_{\ell\ell_1}^{ab}
	\left[\widetilde{\Ham}^R(u)- {\Ham}(u)\right]^{bc}_{\ell_1\ell_2}
	\left[\widetilde{\mathscr R}_z^R(u)\right]_{\ell_2\ell}^{ca} \\
	&\qquad\leq C \sum_{\ell_1, \ell_2 \in \Lambda_R} e^{-\ctCT \left(r_{\ell \ell_1}(u) + r_{\ell\ell_2}^\#(u)\right)}
	e^{- \frac{1}{2}\gamma_0 \min_{\alpha \not = 0} |\bm{r}_{\ell_1\ell_2}(u) + \mathsf{M}_R\alpha|} \\
	    &\qquad\qquad + C\sum_{\above{\ell_1, \ell_2 \in \Lambda}{\ell_1 \not\in\Lambda_R \,\mathrm{ or }\, \ell_2 \not\in\Lambda_R}} e^{-\ctCT r_{\ell\ell_1}(u)} e^{-\ctTBexponent r_{\ell_1\ell_2}(u)} e^{-\ctCT r_{\ell\ell_2}^\#(u)}
	\\
	&\qquad\leq C e^{-\frac{1}{2}\mathfrak{m}\eta \mathrm{dist}(\ell, \mathbb{R}^d \setminus \Omega_R)} 
\end{split}
\end{equation}
where $\eta \coloneqq \frac{1}{2}\mathfrak{m}\min\{\ctCT,\frac{1}{2}\ctTBexponent\}$. In the last line we used the fact that, for $\ell_1,\ell_2 \in \Lambda_R$, we have 
$r_{\ell \ell_1}(u) + r_{\ell\ell_2}^\#(u) + \min_{\alpha \not = 0} |\bm{r}_{\ell_1\ell_2}(u) + \mathsf{M}_R\alpha| \geq \frac{1}{2}\mathfrak{m}\mathrm{dist}(\ell,\mathbb{R}^d\setminus \Omega_R)$ and, for $\ell_1 \in \Lambda \setminus \Lambda_R$ or $\ell_2 \in \Lambda \setminus \Lambda_R$, we have $r_{\ell \ell_1}(u) + r_{\ell_1\ell_2}(u) \geq \mathfrak{m}\mathrm{dist}(\ell,\mathbb{R}^d\setminus \Omega_R)$.

The case where $j\geq 1$ can be treated similarly. For $j=1$, we have
\begin{align} \label{eq:Rsiteenergies_1}
	\frac{\partial \mathcal G_{\ell}^{\beta,R}(Du(\ell))}{\partial [u(m)]_{i}} 
	    -  \frac{\partial \mathcal G^{\beta}_{\ell}(Du(\ell))}{\partial [u(m)]_{i}} &=
    \frac{1}{2\pi i}\sum_a\oint_\mathscr{C} \mathfrak{g}^\beta(z;\mu) \frac{\partial}{\partial u(m)}\left[\mathscr R_z(u) - \widetilde{\mathscr R}_z^R(u)\right]_{\ell\ell}^{aa} \dxx{z}.
\end{align} 
Now, after dropping the $(u)$ in the notation for the resolvent operators, we have
\begin{align*}
	&\frac{\partial}{\partial u(m)}\left[\mathscr R_z(u) - \widetilde{\mathscr R}_z^R(u)\right]_{\ell\ell}^{aa}\\
	&= \Bigg[ \mathscr R_z{\Ham}_{,m}\mathscr R_z \left[ \widetilde{\Ham}^R - {\Ham}\right] \widetilde{\mathscr R}_z^R - \mathscr R_z \left[ \widetilde{\Ham}^R_{,m} - {\Ham}_{,m}\right] \widetilde{\mathscr R}_z^R + \mathscr R_z \left[ \widetilde{\Ham}^R - {\Ham}\right] \widetilde{\mathscr R}_z^R \widetilde{\Ham}^R_{,m} \widetilde{\mathscr R}_z^R  \Bigg]_{\ell\ell}^{aa}. 
\end{align*}
We show that each of these terms can be bounded as required: 
\begin{align}
	&\left[ \mathscr R_z{\Ham}_{,m}\mathscr R_z \left( \widetilde{\Ham}^R - {\Ham}\right) \widetilde{\mathscr R}_z^R \right]_{\ell\ell}^{aa} \nonumber\\
	&\leq C \sum_{\above{\ell_1,\ell_2 \in \Lambda}{\ell_3,\ell_4 \in \Lambda_R}}  
	e^{-\ctCT \left( r_{\ell \ell_1}(u) + r_{\ell_2 \ell_3}(u) + r^\#_{\ell_4 \ell}(u)\right)} 
	e^{-\gamma_0\left( r_{\ell_1 m}(u) + r_{m\ell_2}(u) \right)} 
	e^{ -\tfrac{1}{2}\gamma_0 \min_{\alpha\not=0} |\bm{r}_{\ell_3\ell_4}(u) + \mathsf{M}_R\alpha|} \nonumber\\
	&\qquad+ C\sum_{\above{\ell_1,\ell_2,\ell_3, \ell_4 \in \Lambda}{\ell_3 \not\in \Lambda_R \,\text{ or }\, \ell_4 \not\in \Lambda_R}}
	e^{-\ctCT \left( r_{\ell \ell_1}(u) + r_{\ell_2 \ell_3}(u) + r^\#_{\ell_4 \ell}(u)\right)} 
	e^{-\gamma_0\left( r_{\ell_1 m}(u) + r_{m\ell_2}(u) + r_{\ell_3\ell_4}(u) \right)} \nonumber\\
	&\leq C e^{-\frac{1}{2}\mathfrak{m}\eta 
	\mathrm{dist}(\ell,\mathbb R^d \setminus \Omega_R)}
	e^{-\frac{1}{2}\eta r_{\ell m}(u)} \nonumber\\
	&\qquad \times \bigg( \sum_{\ell_3,\ell_4 \in \Lambda_R} e^{-\frac{1}{2}\eta \left( 
	r_{m\ell_3}(u) + r_{\ell_3\ell_4}^\#(u) + r_{\ell_4\ell}^\#(u)\right)} + 
	\sum_{\above{\ell_3, \ell_4 \in \Lambda}{\ell_3 \not\in \Lambda_R \,\text{ or }\, \ell_4 \not\in \Lambda_R}} e^{-\frac{1}{2}\eta \left( 
	r_{m\ell_3}(u) + r_{\ell_3\ell_4}(u) + r_{\ell\ell_4}^\#(u)\right)} \bigg) \nonumber\\
	&\leq C e^{-\frac{1}{2}\mathfrak{m}\eta \mathrm{dist}(\ell,\mathbb R^d \setminus \Omega_R)} e^{-\frac{1}{2} \eta r_{\ell m}(u)}.\label{eq:Rsiteenergies_2}
\end{align}
In the second to last line we used the fact that, for $\ell_3,\ell_4 \in \Lambda_R$, we have 
$r_{m \ell_3}(u) + r_{m\ell_4}^\#(u) + \min_{\alpha \not = 0} |\bm{r}_{\ell_3\ell_4}(u) + \mathsf{M}_R\alpha| \geq \frac{1}{2}\mathfrak{m}\mathrm{dist}(\ell,\mathbb{R}^d\setminus \Omega_R)$ and, for $\ell_3 \in \Lambda \setminus \Lambda_R$ or $\ell_4 \in \Lambda \setminus \Lambda_R$, we have $r_{\ell \ell_1}(u) + r_{\ell_1m}(u) + r_{m\ell_2}(u) + r_{\ell_2\ell_3}(u) + r_{\ell_3\ell_4} \geq \mathfrak{m}\mathrm{dist}(\ell,\mathbb{R}^d\setminus \Omega_R)$.
	%
%
%
%
%

Similarly, we have
\begin{align}
	&\left[ \mathscr R_z \left( \widetilde{\Ham}^R_{,m} - {\Ham}_{,m}\right) \widetilde{\mathscr R}_z^R \right]_{\ell\ell}^{aa} \nonumber\\
	&\qquad\leq C\sum_{\ell_1,\ell_2\in \Lambda_R}
	e^{-\ctCT\left(r_{\ell\ell_1}(u) + r_{\ell_2\ell}^\#(u)\right) }
	e^{-\frac{1}{2}\ctTBexponent \left( \min_{\alpha\not=0}|\bm{r}_{\ell_1m}(u) + \mathsf{M}_R\alpha| +
	\min_{\alpha\not=0}|\bm{r}_{\ell_2m}(u) + \mathsf{M}_R\alpha|\right)} \nonumber\\
%
	&\qquad\qquad + 
	C\sum_{\above{\ell_1,\ell_2\in \Lambda}
	{\ell_1\not\in\Lambda_R \,\textrm{or}\, \ell_2\not\in\Lambda_R}} 
	e^{-\ctCT\left(r_{\ell\ell_1}(u) + r_{\ell_2\ell}(u)\right) }
	e^{-\ctTBexponent \left( r_{\ell_1m}(u) + r_{\ell_2m}(u) \right)} \nonumber\\
    &\qquad\leq C e^{-\frac{1}{2}\mathfrak{m}\eta\mathrm{dist}(\ell,\mathbb R^d \setminus \Omega_R)}e^{-\frac{1}{2}\eta r_{\ell m}^\#(u)}. \label{eq:Rsiteenergies_3}
\end{align}
	By using \cref{eq:Rsiteenergies_1}, \cref{eq:Rsiteenergies_2} and \cref{eq:Rsiteenergies_3}, together with the boundedness of $\mathfrak{g}^\beta(\,\cdot\,;\mu)$ along the contour $\mathscr C$, we can therefore conclude. 
	
	For $j\geq 2$, similar arguments can be made but notation becomes tedious. Since no new ideas are used, we omit the proof. 
\end{proof}

\section{Proof of Theorem~\ref{thm:0T_thermodynamic_limit}: Stability }
\label{app:stab}

\begin{proof}[Proof of \cref{eq:T1T2_general_bound}]
Since $v_R^\mathrm{co}$ is an admissible displacement on $\Lambda$, we have
\begin{align}
        &\Braket{\left(\delta^2\mathcal G(T_R\overline{u}) - \delta^2 G^R(T_R\overline{u})\right)v_R^\mathrm{co},w} \nonumber\\
        &= \sum_{\ell \in \Lambda \cap B_{S(R)}} \sum_{\rho_1,\rho_2 \in \Lambda_R - \ell} 
            D_{\rho_1}v_R^\mathrm{co}(\ell)^T\big( \mathcal G_{\ell, \rho_1\rho_2}(DT_R \overline{u}(\ell)) - \mathcal G^R_{\ell, \rho_1\rho_2}(DT_R \overline{u}(\ell)) \big) D_{\rho_2}w(\ell) \nonumber\\
        &\qquad + \sum_{\ell \in \Lambda \cap B_{S(R)}} \sum_{\above{\rho_1,\rho_2 \in \Lambda - \ell}{\ell + \rho_1 \not\in \Lambda_R \textrm{ or }\ell + \rho_2 \not\in \Lambda_R}} 
            D_{\rho_1}v_R^\mathrm{co}(\ell)^T \mathcal G_{\ell, \rho_1\rho_2}(DT_R \overline{u}(\ell)) D_{\rho_2}w(\ell) \nonumber\\
        &\qquad + \sum_{\ell \in \Lambda_R \setminus B_{S(R)} } \sum_{\rho_1,\rho_2 \in \Lambda\cap B_{S(R)} - \ell } 
            D_{\rho_1}v_R^\mathrm{co}(\ell)^T\big( \mathcal G_{\ell, \rho_1\rho_2}(DT_R \overline{u}(\ell)) - \mathcal G^R_{\ell, \rho_1\rho_2}(DT_R \overline{u}(\ell)) \big) D_{\rho_2}w(\ell) \nonumber \\
        &\qquad + \sum_{\ell \in \Lambda \setminus \Lambda_R} \sum_{\rho_1,\rho_2 \in \Lambda - \ell}
            D_{\rho_1}v_R^\mathrm{co}(\ell)^T \mathcal G_{\ell, \rho_1\rho_2}(DT_R \overline{u}(\ell)) D_{\rho_2}w(\ell) \nonumber\\
        &\eqqcolon \mathrm{T}_{1} + \mathrm{T}_{2} + \mathrm{T}_{3} + \mathrm{T}_{4}.
         \label{eq:stability_core}
\end{align}
We shall consider each term separately. Firstly, we note that $\mathrm{T}_{1}$ and $\mathrm{T}_{3}$ in \cref{eq:stability_core} can be bounded using the convergence of the site energies as $R\to \infty$ (Lemma~\ref{eq:thermodynamic_site_energies}):
\begin{align*}
    |\mathrm{T}_{1}| \leq 
    C\sum_{\above{\ell \in \Lambda \cap B_{S(R)}}{\rho_1,\rho_2 \in \Lambda_R - \ell}} e^{-\eta(R + |\rho_1| + |\rho_2|)} |D_{\rho_1}v^\mathrm{co}_R(\ell)||D_{\rho_2}w(\ell)|
        \leq Ce^{-\eta R} \|Dv^\mathrm{co}_R\|_{\ell^2_\ctGamma}
        \|Dw\|_{\ell^2_\ctGamma}
\end{align*}
where $\eta \coloneqq \frac{1}{2}\mathfrak{m}\min\{\ctCT, \frac{1}{2}\ctTBexponent\}$. Similarly, we have $|\mathrm{T}_{3}| \lesssim e^{-\gamma R}\|Dv_R^\mathrm{co}\|_{\ell^2_\ctGamma}\|Dw\|_{\ell^2_\ctGamma}$.

For $\mathrm{T}_{2}$, we apply the locality of the site energies to conclude that
\begin{align*}
    |\mathrm{T}_{2}| &\leq 
   C\sum_{\ell \in \Lambda \cap B_{S(R)}} \sum_{\above{\rho_1,\rho_2 \in \Lambda - \ell}{\ell + \rho_1 \not\in \Lambda_R \textrm{ or }\ell + \rho_2 \not\in \Lambda_R}} 
        e^{-\ctCT(|\rho_1| + |\rho_2|)} |D_{\rho_1}v_R^\mathrm{co}(\ell)||D_{\rho_2}w(\ell)| \\
        &\leq Ce^{-\frac{1}{2}\ctCT (R-S(R))} \|Dv_R^\mathrm{co}\|_{\ell^2_\ctGamma}
        \|Dw\|_{\ell^2_\ctGamma}.
\end{align*}
Here, we have used the fact that $|\rho_1| \geq R- S(R)$ or $|\rho_2| \geq R-S(R)$.

Now we move on to consider $\mathrm{T}_{4}$. Using the fact that $v_R^\mathrm{co}(\ell) = 0$ for all $\ell \in \Lambda\setminus B_{S(R)}$, together with the locality of the site energies to conclude that
\begin{align*}
    \mathrm{T}_{4} &= \sum_{\ell \in \Lambda \setminus \Lambda_R} \sum_{\rho_1,\rho_2 \in \Lambda \cap B_{S(R)}-\ell} D_{\rho_1}v_R^\mathrm{co}(\ell)^T \mathcal G_{\ell,\rho_1\rho_2}(DT_R\overline{u}(\ell))D_{\rho_2}w(\ell)\\
    & \leq \sum_{\ell \in \Lambda \setminus \Lambda_R} \sum_{\rho_1,\rho_2 \in \Lambda \cap B_{S(R)}- \ell} e^{-\ctCT(R - S(R))} |D_{\rho_1} v^\mathrm{co}_R(\ell)| |D_{\rho_2} w(\ell)| \\
    &\leq  e^{-\ctCT(R - S(R) + |\rho_1| + |\rho_2|)} \|Dv_R^\mathrm{co}\|_{\ell^2_\ctGamma}
    \|Dw\|_{\ell^2_\ctGamma}.
\end{align*}
%
%
%
%
%
%

We now replace $T_R\overline{u}$ with $\overline{u}$ and show that this introduces an approximation error that is a constant multiple of $\|D\overline{u}\|_{\ell^2_{\ctGamma}(\Lambda \setminus B_{R/2})}$. For $R$ sufficiently large, we have
\begin{gather}
    \delta^2 \mathcal{G}(T_R\overline{u}) - \delta^2 \mathcal{G}(\overline{u}) 
    = \int_{0}^1 \delta^3 \mathcal{G}(tT_R\overline{u}+(1-t)\overline{u})[T_R\overline{u}-\overline{u}] \mathrm{d}t; \quad \textrm{and}\label{eq:2ndderivative}\\
    \mathrm{dist}(\mu, \sigma(\Ham(u))) \geq \frac{1}{4} \ctDist(\overline{u}) \quad \textrm{for all }u \coloneqq tT_R\overline{u} + (1-t)\overline{u} \textrm{ and }t\in [0,1].\label{eq:uniform_locality_estimates}
    %
    %
\end{gather}
The fact that the spectrum is uniformly bounded away from $\mu$ along the path between $T_R\overline{u}$ and $\overline{u}$ results from Lemma~\ref{lem:perturbation_spec} and ensures that the exponents in the locality estimates are uniform in $t$. By applying \cref{eq:2ndderivative}, we obtain
\begin{align}
    &\Braket{\delta^2 \mathcal{G}(T_R\overline{u})v_R^\mathrm{co},w} - \Braket{\delta^2 \mathcal{G}(\overline{u})v_R^\mathrm{co},w}\nonumber\\
    &\qquad\leq C
    \sum_{\ell \in \Lambda} \sum_{\rho_1,\rho_2,\rho_3 \in \Lambda - \ell}
    e^{-\ctCT(|\rho_1| + |\rho_2| + |\rho_3|)}
    |D_{\rho_1}v_R^\mathrm{co}(\ell)|
    |D_{\rho_2}w(\ell)|
    |D_{\rho_3}(T_R\overline{u} - \overline{u})(\ell)|\nonumber \\
    &\qquad\leq C \|D(\overline{u}-T_R\overline{u})\|_{\ell^2_\ctGamma} \|Dv_R\|_{\ell^2_\ctGamma} \|Dw\|_{\ell^2_\ctGamma} \nonumber\\ 
    &\qquad\leq C \|D\overline{u}\|_{\ell^2_\ctGamma(\Lambda \setminus B_{R/2})} \|Dv_R\|_{\ell^2_\ctGamma} \|Dw\|_{\ell^2_\ctGamma}.
    \label{eq:T1_conclusion2}
\end{align}
Here, we have used H\"older's inequality and the fact that
\[
    \|Du\|_{\ell^p_\ctGamma} \coloneqq \bigg(\sum_{\ell\in \Lambda}\sum_{\rho \in \Lambda - \ell} e^{-\ctGamma|\rho|}
    |D_\rho u(\ell)|^p\bigg)^{1/p} \leq C_p \|Du\|_{\ell^2_\ctGamma} \quad \textrm{for all } p \geq 2.
\]
This result is easy to prove by showing that $\|Du\|_{\ell^p_{{p\ctGamma}}} \leq \|Du\|_{\ell^2_{2\ctGamma}}$ and using the fact that the norms $\|D\cdot\|_{\ell^2_\ctGamma}$ are all equivalent for $\ctGamma > 0$.
\end{proof}

\begin{proof}[Proof of \cref{eq:stability_ref}]
Let $n\in\mathbb N$, $\Omega_R$ be the continuous domain corresponding to $\Lambda_R^\mathrm{ref}$ and define 
\begin{align*}
\Omega_{nR} &\coloneqq \bigcup_{\above{\alpha \in \mathbb Z^d}{|\alpha|_\infty < n}} (\Omega_R + \mathsf{M}_R\alpha) \quad \textrm{and} \quad
\Lambda^\mathrm{ref}_{nR} \coloneqq \Lambda^\mathrm{ref} \cap \Omega_{nR}.
\end{align*}

For $\ell \in \Lambda^\mathrm{ref}_{nR}$ we let $\Lambda^\mathrm{ref}_R(\ell) \subset \Lambda^\mathrm{ref}_{nR}$ be a translation of $\Lambda_R^\mathrm{ref}$ such that for all $k \in \Lambda^\mathrm{ref}_{nR} \setminus \Lambda_{R}^\mathrm{ref}(\ell)$, we have $|\ell - k| > R$. Since we are considering the reference domain, we can extend $v^\mathrm{ff}_R$ by periodicity to $\Lambda^\mathrm{ref}_{nR}$, and use the translational symmetry to conclude that,
\begin{align*}
    &\Braket{\delta^2 {G}_{\mathrm{ref}}^{nR} v^\mathrm{ff}_R,v^\mathrm{ff}_R} =
    \sum_{\above
    {\alpha \in \mathbb Z^d}{|\alpha|_\infty < n}}
    \sum_{\ell \in \Lambda^\mathrm{ref}_R}
    \sum_{\rho_1, \rho_2 \in \Lambda^\mathrm{ref}_{nR} - (\ell + \mathsf{M}_R\alpha)}
    D_{\rho_1} v_R^\mathrm{ff}(\ell)^T 
    [\mathcal{G}_\mathrm{ref}^{nR}]_{\ell,\rho_1\rho_2}
    D_{\rho_2} v_R^\mathrm{ff}(\ell) \\
    &\qquad= (2n+1)^d
    \Braket{\delta^2 {G}_{\mathrm{ref}}^{R} v^\mathrm{ff}_R,v^\mathrm{ff}_R} \\
    &\qquad\quad + (2n+1)^d    
    \sum_{\ell \in \Lambda^\mathrm{ref}_R}
    \sum_{\rho_1,\rho_2 \in \Lambda^\mathrm{ref}_{R}(\ell) - \ell}
        D_{\rho_1} v_R^\mathrm{ff}(\ell)^T 
        \big( 
        [\mathcal{G}_\mathrm{ref}^{nR}]_{\ell,\rho_1\rho_2} - [\mathcal{G}_\mathrm{ref}^{R}]_{\ell,\rho_1\rho_2}
        \big)
        D_{\rho_2} v_R^\mathrm{ff}(\ell)\\
    &\qquad\quad + \sum_{\above
    {\alpha \in \mathbb Z^d}{|\alpha|_\infty < n}}
    \sum_{\ell \in \Lambda^\mathrm{ref}_R }
    \sum_{\above{\rho_1,\rho_2 \in \Lambda^\mathrm{ref}_{nR} - (\ell + \mathsf{M}_R\alpha)}{\ell + \rho_1 \not \in \Lambda^\mathrm{ref}_R(\ell) \textrm{ or }\ell + \rho_2 \not \in \Lambda^\mathrm{ref}_R(\ell)}}
        D_{\rho_1} v_R^\mathrm{ff}(\ell)^T 
        [\mathcal{G}_\mathrm{ref}^{nR}]_{\ell,\rho_1\rho_2}
        D_{\rho_2} v_R^\mathrm{ff}(\ell).
\end{align*}
We note that for $\alpha \in \mathbb Z^d$ with $|\alpha|_\infty < n$ and $\ell \in \Lambda^\mathrm{ref}_R$, we have $\Lambda^\mathrm{ref}_{nR} - (\ell + \mathsf{M}_R\alpha) \subset \Lambda^\mathrm{ref}_{(2n+1)R}$. Therefore, by applying the locality and convergence results for the site energies, we have
\begin{align}\label{eq:stability_RminusnR}
\begin{split}
    &\left|\Braket{\delta^2  G_{\mathrm{ref}}^R v^\mathrm{ff}_R,v^\mathrm{ff}_R} - \frac{1}{(2n+1)^d} \Braket{\delta^2  G_{\mathrm{ref}}^{nR} v^\mathrm{ff}_R,v^\mathrm{ff}_R}\right|   \\
    %
    %
        %
    %
    &\qquad \leq C \sum_{\ell\in\Lambda^\mathrm{ref}_R} \sum_{\rho_1,\rho_2 \in \Lambda^\mathrm{ref}_R(\ell) - \ell} 
    e^{-\eta(R + |\rho_1| + |\rho_2|)} 
    |D_{\rho_1}v_R^\mathrm{ff}(\ell)|
    |D_{\rho_2}v_R^\mathrm{ff}(\ell)| \\ 
        &\qquad\qquad +\frac{Ce^{-\frac{1}{2}\ctCT R}}{(2n+1)^d} 
        \sum_{\ell \in \Lambda^\mathrm{ref}_{nR}}
        \sum_{\rho_1,\rho_2 \in \Lambda^\mathrm{ref}_{(2n+1)R} - \ell}
        e^{-\frac{1}{2}\ctCT(|\rho_1| + |\rho_2|)}
        |D_{\rho_1}v_R^\mathrm{ff}(\ell)|
        |D_{\rho_2}v_R^\mathrm{ff}(\ell)| \\
    &\qquad \leq C (e^{-\eta R} + e^{-\ctCT R}) 
    \|Dv^\mathrm{ff}_R\|_{\ell^2_\ctGamma(\Lambda_R^\mathrm{ref})}^2
\end{split}
\end{align}
where $\eta \coloneqq \frac{1}{2}\mathfrak{m}\min\{\ctCT, \frac{1}{2}\ctTBexponent\}$. Here, we have used the fact that for $m\in\mathbb N$ and $w \colon \Lambda^\mathrm{ref}_R \to \mathbb R^d$ that is extended periodically to $\Lambda^\mathrm{ref}_{mR}$, we have $\|Dw\|_{\ell_\ctGamma^2(\Lambda^\mathrm{ref}_{mR})} \leq C m^{d/2} \|Dw\|_{\ell^2_{\ctGamma}(\Lambda^\mathrm{ref}_R)}$.

Now we may choose a smooth cut off function $\phi \colon \mathbb R^d \to [0,1]$ depending on $nR$ such that 
\begin{align}\label{eq:phi_cut_off}
\phi = 0 \textrm{ on } B_r(\Omega_{nR}^c)  
\quad \textrm{and} \quad 
\phi = 1 \textrm{ on } \Omega_{nR} \setminus B_{2r}(\Omega_{nR}^c)
\end{align}
for some $r >0$ such that $3r < R$. Then, 
\begin{align}
&\Braket{\delta^2 G_{\mathrm{ref}}^{nR}v^\mathrm{ff}_R,v^\mathrm{ff}_R}\nonumber\\
&=\Braket{\delta^2  G_{\mathrm{ref}}^{nR}
\phi v^\mathrm{ff}_R,\phi v^\mathrm{ff}_R} 
    + 2\Braket{\delta^2  G_{\mathrm{ref}}^{nR}(1-\phi)v^\mathrm{ff}_R,\phi v^\mathrm{ff}_R} 
    +\Braket{\delta^2  G_{\mathrm{ref}}^{nR}(1-\phi)v^\mathrm{ff}_R,(1-\phi)v^\mathrm{ff}_R}  \nonumber\\
&\eqqcolon \mathrm{T}_{1} + \mathrm{T}_{2} + \mathrm{T}_{3}.\label{eq:stability_R_FF}
\end{align}
We consider each of these terms in turn. 

Firstly, we note that since $\phi$ is of compact support, we can approximate $\mathrm{T}_{1}$ by the corresponding infinite domain quantity. Using the fact that $\phi = 0$ on the buffer region $\Omega_{nR}\setminus B_r(\Omega_{nR}^c)$, we can conclude the approximation error is exponentially small in $r$:
\begin{align}\label{eq:T3minus}\begin{split}
    \mathrm{T}_{1} - \Braket{\delta^2 \mathcal G_\mathrm{ref}\phi v^\mathrm{ff}_R,\phi v^\mathrm{ff}_R} &=
    \Braket{\left(\delta^2  G_{\mathrm{ref}}^{nR} - \delta^2\mathcal G_\mathrm{ref}\right)\phi v^\mathrm{ff}_R,\phi v^\mathrm{ff}_R}\\ 
    &\leq C(e^{-\frac{1}{2}\eta r} + e^{-\frac{1}{2}\ctCT r})\|D(\phi v^\mathrm{ff}_R)\|_{\ell^2_\ctGamma(\Lambda^\mathrm{ref})}^2
    %
%
\end{split}\end{align}
where $\eta \coloneqq \frac{1}{2}\mathfrak{m}\min\{\ctCT, \frac{1}{2}\ctTBexponent\}$. This is a simple calculation in which the left hand side can be expanded in terms of site energies and we can use the fact that if $\ell \in \Lambda^\mathrm{ref}_{nR}$ and $\rho_1,\rho_2 \in \Lambda^\mathrm{ref}_{nR}-\ell$, for which $\phi(\ell) \not= 0$ or $\phi(\ell + \rho_l)\not=0$ for both $l = 1,2$ then $\mathrm{dist}(\ell,\Omega_{nR}^c) + |\rho_1| + |\rho_2| > r$. Moreover, if $\ell \not \in \Lambda_{nR}^\mathrm{ref}$ then we consider $\rho_1,\rho_2$ for which $\phi(\ell + \rho_l)\not=0$ for both $l = 1,2$ and thus $|\rho_1| > R$.

Using the fact that the reference configuration is stable, and using \cref{eq:phi_cut_off}, we can conclude that, for sufficiently large $R$, we have
\begin{align}\label{eq:T31}\begin{split}
    \mathrm{T}_{1} &\geq 
    \left(c_\mathrm{stab} - Ce^{-\frac{1}{2}\eta r} \right) 
    \|D(\phi v_R^\mathrm{ff})\|_{\ell^2_\ctGamma(\Lambda^\mathrm{ref})}^2\\
    &\geq
   Cn^{d} \left(c_\mathrm{stab} - Ce^{-\frac{1}{2}\eta r} \right) \|Dv_R^\mathrm{ff}\|_{\ell^2_\ctGamma(\Lambda_{R}^\mathrm{ref})}^2.
\end{split}\end{align}

We now show that, after dividing by $n^d$, $\mathrm{T}_{2}$ and $\mathrm{T}_{3}$ can be made arbitrarily small for sufficiently large $R$ and $n$. We fix $w \colon \Lambda^\mathrm{ref}_R \to \mathbb R^d$ and extend periodically to $\Lambda^\mathrm{ref}_{nR}$ and consider 
\begin{align}\label{eq:Grefn}\begin{split}
    &\Braket{\delta^2  G^{nR}_\mathrm{ref}(1 - \phi)v_R^\mathrm{ff}, w} \\
    &\qquad = \sum_{\above{\ell \in \Lambda^\mathrm{ref}_{nR}\colon}{\mathrm{dist}\left(\ell,\Omega_{nR}^c\right)\leq 3r}} 
    \sum_{\rho_1,\rho_2 \in \Lambda^\mathrm{ref}_{nR}-\ell}
    D_{\rho_1}(1-\phi)v_R^\mathrm{ff}(\ell)^T
    \left[\mathcal G^{nR}_\mathrm{ref}\right]_{\ell, \rho_1\rho_2}
    D_{\rho_2}w(\ell) \\
    %
    &\qquad\qquad + \sum_{\above{\ell \in \Lambda^\mathrm{ref}_{nR}\colon}{\mathrm{dist}\left(\ell,\Omega_{nR}^c\right)>3r}} 
    \sum_{\rho_1,\rho_2 \in \Lambda^\mathrm{ref}_{nR}-\ell} D_{\rho_1}(1-\phi)v_R^\mathrm{ff}(\ell)^T
    \left[\mathcal G^{nR}_\mathrm{ref}\right]_{\ell, \rho_1\rho_2}
    D_{\rho_2}w(\ell).  
\end{split}
\end{align}

For the first term of \cref{eq:Grefn} we may use the fact that the set of $\ell \in \Lambda^\mathrm{ref}_{nR}$ for which $\mathrm{dist}(\ell,\Omega_{nR}^c)\leq 3r$ has at most $C n^{d-1}$ elements to conclude,
\begin{align}\label{eq:Grefn_1}
\begin{split}
&\sum_{\above{\ell \in \Lambda^\mathrm{ref}_{nR}\colon}{\mathrm{dist}\left(\ell,\Omega_{nR}^c\right)\leq 3r}} 
\sum_{{\rho_1,\rho_2 \in \Lambda^\mathrm{ref}_{nR}-\ell}} D_{\rho_1}(1-\phi)v_R^\mathrm{ff}(\ell)^T
    \left[\mathcal G^{nR}_\mathrm{ref}\right]_{\ell, \rho_1\rho_2}
    D_{\rho_2}w(\ell) \\
    &\qquad \leq  C \sum_{
        \above
        {\ell \in \Lambda^\mathrm{ref}_{nR}\colon}
        {\mathrm{dist}\left(\ell,\Omega_{nR}^c\right)\leq 3r}
    } 
    \sum_{
    \above
        {\rho_1,\rho_2 \in \Lambda^\mathrm{ref}_{nR} - \ell \colon}
        {|\rho_1|, |\rho_2| \leq 2R }
    } 
    e^{-\ctCT(|\rho_1| + |\rho_2|)}
    |D_{\rho_1}(1-\phi)v_R^\mathrm{ff}(\ell)|
    |D_{\rho_2}w(\ell)| \\
        &\qquad\qquad +  C \sum_{\ell \in \Lambda^\mathrm{ref}_{nR}} 
    \sum_{
    \above
        {\rho_1,\rho_2 \in \Lambda^\mathrm{ref}_{nR} - \ell \colon}
        {|\rho_1| > 2R \textrm{ or } |\rho_2| > 2R }
    }
    e^{-\ctCT(|\rho_1| + |\rho_2|)}
    |D_{\rho_1}(1-\phi)v_R^\mathrm{ff}(\ell)|
    |D_{\rho_2}w(\ell)|  \\
    &\qquad \leq C\left(n^{d-1} + n^{d} e^{-\ctCT R}\right) \|Dv_R^\mathrm{ff}\|_{\ell^2_\ctGamma(\Lambda^\mathrm{ref}_R)} 
    \|Dw\|_{\ell^2_\ctGamma(\Lambda^\mathrm{ref}_R)}.
\end{split}
\end{align}
%

Similarly, for the second term in \cref{eq:Grefn}, we use the fact that, for $\ell \in \Lambda^\mathrm{ref}_{nR}$ with $\mathrm{dist}(\ell, \Omega_{nR}^c) > 3r$, by \cref{eq:phi_cut_off}, we only have to sum over $\rho_1 \in \Lambda^\mathrm{ref}_{nR} -\ell$ for which $\mathrm{dist}(\ell + \rho_1, \Omega_{nR}^c) \leq 2r$. That is, $|\rho_1| \geq r$, and so we obtain an 
approximation error exponentially small in $r$ (as in the second term of \cref{eq:Grefn_1}):
\begin{align}\label{eq:Grefn_2}
\begin{split}
  &\sum_{\above{\ell \in \Lambda^\mathrm{ref}_{nR}\colon}{\mathrm{dist}\left(\ell,\Omega_{nR}^c\right)>3r}} 
    \sum_{\rho_1,\rho_2 \in \Lambda^\mathrm{ref}_{nR}-\ell} D_{\rho_1}(1-\phi)v_R^\mathrm{ff}(\ell)^T
    \left[\mathcal G^{nR}_\mathrm{ref}\right]_{\ell, \rho_1\rho_2}
    D_{\rho_2}w(\ell)  \\
    &\qquad\qquad \leq C n^{d} e^{-\frac{1}{2}\ctCT r}
    \|Dv^\mathrm{ff}_R\|_{
    \ell^2_\ctGamma(\Lambda^\mathrm{ref}_R)}
    \|Dw\|_{
    \ell^2_\ctGamma(\Lambda^\mathrm{ref}_R)}.
\end{split}
\end{align}

By combining \cref{eq:Grefn}, \cref{eq:Grefn_1} and \cref{eq:Grefn_2}, for both $w = \phi v_R^\mathrm{ff}$ and $w = (1 - \phi) v_R^\mathrm{ff}$, we have 
\begin{align}\label{eq:T32plusT33}\begin{split}
    |\mathrm{T}_{2}| + |\mathrm{T}_{3}| \leq C \left( n^{d-1} + n^{d} e^{-\frac{1}{2}\ctCT r}\right)
    \|Dv_R^\mathrm{ff}\|_{\ell_\ctGamma^2}^2.
\end{split}
\end{align}
Therefore, by \cref{eq:stability_R_FF}, \cref{eq:T31} and \cref{eq:T32plusT33}, 
\begin{align}\label{eq:stab_nd}\begin{split}
    \frac{1}{n^d}\Braket{\delta^2
    G^{nR}_\mathrm{ref}v_R^\mathrm{ff}, v_R^\mathrm{ff}} 
    &\geq C\left( \widetilde{c}_\mathrm{stab} -( e^{-\frac{1}{2}\eta r} + n^{-1} + e^{-\tfrac{1}{2}\ctCT r})\right) \|Dv_R^\mathrm{ff}\|_{\ell^2_\ctGamma}^2
\end{split}
\end{align}
where $\widetilde{c}_\mathrm{stab}$ is a positive constant multiple of $c_\mathrm{stab}$.

Applying \cref{eq:stab_nd} together with \cref{eq:stability_RminusnR} we can choose $n$, $r$ then $R$ sufficiently large to conclude.
\end{proof}

\small 
\bibliography{refs}
\bibliographystyle{siam}

\end{document}